\documentclass[a4paper,reqno,11pt]{amsart}

\usepackage{amsfonts}
\usepackage{tikz}
\usepackage{tkz-graph}
\usetikzlibrary{graphs}
\usepackage{mathtools}
\usetikzlibrary{graphs.standard}
\usepackage[margin=0.8in]{geometry}
\usepackage{accents}
\usepackage{array,calc,bbold}
\usepackage{mathrsfs}
\usepackage{stmaryrd}
\mathtoolsset{showonlyrefs}
\usepackage{pgf}
\usetikzlibrary{shapes,cd,arrows,automata,decorations.pathreplacing,angles,quotes,calc}
\usepackage{tkz-euclide}
\usepackage{caption}
\usepackage[utf8]{inputenc}
\usepackage{subcaption}
\usepackage{amssymb}
\usepackage{amsmath,calc,graphicx}
\usepackage{url}
\usepackage[hidelinks]{hyperref}
\usepackage{amssymb}
\usepackage{amsthm}
\usepackage{multirow}
\usepackage{float}
\usepackage{enumerate}
\usepackage{enumitem}
\usepackage[normalem]{ulem}
\usepackage{hyperref}
\usepackage{textgreek}

\hypersetup{
  colorlinks   = true, 
  urlcolor     = MyBlue, 
  linkcolor    = MyBlue, 
  citecolor   = darkspringgreen 
}
\definecolor{darkgreen}{rgb}{0.0,0.5,0.0}
\definecolor{darkspringgreen}{rgb}{0.05, 0.5, 0.06}
\usepackage[]{youngtab}

\usepackage{xcolor}
\definecolor{MyBlue}{rgb}{0.0,0.0,0.55}
\colorlet{NextBlue}{MyBlue!20}
\colorlet{SecondBlue}{MyBlue!40}
 \usepackage[foot]{amsaddr}
\numberwithin{equation}{section}
\numberwithin{figure}{section}
\theoremstyle{plain}

\newtheorem{theorem}{Theorem}[section]
\newtheorem{lemma}[theorem]{Lemma}
\newtheorem{corollary}[theorem]{Corollary}
\newtheorem{proposition}[theorem]{Proposition}
\newtheorem{definition}[theorem]{Definition}

\newtheorem{rhprob}{Riemann-Hilbert Problem}

\theoremstyle{remark}
\newtheorem{remark}[theorem]{Remark}

\usepackage{latexsym}
\usepackage[giveninits=true,uniquename=false,uniquelist=false,style=alphabetic,sortcites=false,natbib=true,backend=biber, maxbibnames=99]{biblatex}
\usepackage{xparse}
\usepackage{float,soul,ulem,cancel}
\providecommand*\email[1]{\href{mailto:#1}{#1}}

\DeclareMathOperator{\sff}{{\sf f}}
\DeclareMathOperator{\sfg}{{\sf g}}
\DeclareMathOperator{\sfw}{{\sf w}}

\newcommand\Psix{\textrm{P}_{\textrm{VI}}}
\newcommand\Pthree{\textrm{P}_{\textrm{III}}}

\usepackage{xcolor}
\definecolor{MyBlue}{rgb}{0.0,0.0,0.55}
\colorlet{NextBlue}{MyBlue!20}
\colorlet{SecondBlue}{MyBlue!40}

\NewDocumentCommand{\tens}{t_}
 {%
  \IfBooleanTF{#1}
   {\tensop}
   {\otimes}%
 }
\NewDocumentCommand{\tensop}{m}
 {%
  \mathbin{\mathop{\otimes}\displaylimits_{#1}}%
 }

\newcounter{rhp}

   \numberwithin{rhp}{section} 

\usepackage{xparse}
\usepackage{float}
\usepackage{todonotes}
\usepackage{etoolbox}

\usepackage{framed}
\definecolor{shadecolor}{rgb}{0.9, 0.9, 0.86}
\usetikzlibrary{arrows,positioning,cd,calc,math,decorations.pathreplacing,decorations.markings,decorations.pathmorphing,patterns
}
\tikzset{wave/.style={decorate, decoration=snake}}

\usepackage{xcolor}

\usepackage{subcaption}

\DeclareMathOperator{\res}{Res}

\newcommand{\pain}[1]{Painlev\'e}

\newcommand{\x}[0]{\xi}

\title{A Tau function for $q$-Painlev\'e VI as a Fredholm determinant}
\author{Harini Desiraju}
\address[Harini Desiraju]{Mathematical Institute, Andrew Wiles Building, Woodstock Rd, Oxford OX2 6GG, United Kingdom.}
\email{\href{mailto:harini.desiraju@maths.ox.ac.uk}{harini.desiraju@maths.ox.ac.uk}}

\author{Pieter Roffelsen}
\address[Pieter Roffelsen]{School of Mathematics and Statistics, University of Sydney, Camperdown, NSW 2006, Australia.}
\email{\href{mailto:pieter.roffelsen@sydney.edu.au}{pieter.roffelsen@sydney.edu.au}}
\date{}

\begin{document}
\maketitle
\begin{abstract}
We give an analytic construction of a tau function for the $q$-difference sixth Painlev\'e equation ($q\Psix$) as a Fredholm determinant through the general Riemann-Hilbert problem associated with it.
We show that the tau function is an analytic function on its domain of definition that vanishes at a particular point if and only if the corresponding Riemann-Hilbert problem is point-wise not solvable there. We express the corresponding $q\Psix$ transcendents in terms of the tau function as well as three copies of it with some of the parameters shifted. Then the vanishing of any of these four tau functions corresponds to the transcendents taking value in a specific corresponding exceptional line on the initial value space of $q\Psix$. Finally, we derive an asymptotic expansion of the tau function for small times $t$.
\end{abstract}

\section{Introduction}

Painlev\'e equations are a class of integrable equations that have been studied extensively since the late seventies for the fact that they form universal nonlinear models in the field of integrable systems.
Their solutions, called Painlev\'e transcendents, naturally arise in several areas of mathematical physics including statistical mechanics, random matrix theory and conformal field theory \cite{conte2012painleve,iwasaki2013gauss}.

Each Painlev\'e equation governs  isomonodromic deformation within one or more classes of linear ODEs and the corresponding \textit{Riemann-Hilbert problems} (RHPs) allow for the  global asymptotic analysis of their solutions \cite{jimbo1981monodromy,fokasitskapaev}. A fundamental role in the theory of Painlev\'e equations is played by the corresponding tau functions, which encode pointwise solvability of the relevant RHPs as well as asymptotic and geometric data of these integrable systems \cite{jimbo1981monodromy,harnad2021tau}. However, due to the transcendental nature of the tau functions, an explicit representation of tau functions was elusive for a long time.

The seminal papers \cite{gavrylenko2018fredholm, CGL2017} proved that the general tau function of the Painlev\'e VI ($\Psix$) equation can be written as a Fredholm determinant.
These expressions were instrumental in proving long-standing problems such as the connection problem for Painlev\'e tau functions
 as well as proving their representations in terms of conformal blocks \cite{Gamayun2012}. The pith of their construction is as follows. They first recast the RHP of $\Psix$ equation, which is classically defined on a union of contours, to a RHP on the unit circle.
The tau function is then written as the Fredholm determinant of a trace class operator constructed in terms of Toeplitz operators naturally associated with this RHP. 
 The approach of \cite{gavrylenko2018fredholm,CGL2017} has since been extended to other Painlev\'e equations \cite{GavrylenkoLisovyi2017,GIL2018b,GIL2018a,desiraju2019tau,Desiraju2021} and to integrable systems on the torus \cite{DMDG2020,desiraju2022painleve,DMDG2022,delmonte2025modulartransformationstaufunctions}. However, these extensions all deal with continuous as opposed to discrete integrable systems.



Recently, the Riemann-Hilbert approach to Painlev\'e equations has been extended to several $q$-difference equations \cite{gavrylenko2025riemannhilbertproblemsfredholmdeterminants,JR21,joshi_rof_qpvi,ORS20}. These developments set the stage for extending the Fredholm determinant framework to tau functions associated with discrete Painlev\'e equations. Our focus is on the $q$-difference sixth Painlev\'e equation ($q\Psix$), discovered by Jimbo and Sakai \cite{jimbo1996q} as the system governing a connection matrix preserving deformation of a linear $q$-difference system. Mano \cite{mano_qpvi} observed that the connection matrix factorises into two connection matrices of hypergeometric systems for an open subset of the relevant parameter space.  Ohyama et al. \cite{ORS20} established such factorisation in complete generality,
and called them Mano decompositions. The technology of Mano decompositions was used in \cite{roffelsen2024segre} to solve the connection problem for $q\Psix$. In the current paper, we will use it to rewrite the general Riemann-Hilbert problem associated with $q\Psix$ to one posed on a fixed circle and use it to define an associated Fredholm determinant. We argue that this Fredholm determinant is a tau function for $q\Psix$. Indeed, it is analytic on its domain of definition and vanishes at a time $t$ if and only if the corresponding Riemann-Hilbert problem is pointwise not solvable there. This is equivalent to the associated Painlev\'e transcendent taking value in a particular exceptional line on the initial value space at time $t$. These are among the defining properties of tau functions in the differential theory \cite{miwatau81,okamoto81,palmer}.

We note that tau functions for $q\Psix$ have been introduced in the literature before. Tsuda and Matsuda \cite[Definition 1.3]{tsudamasuda} introduce a tau function for $q\Psix$ as a formal function on a sublattice of the Picard group of the initial value space of $q\Psix$ that is compatible with the action of the extended
 affine Weyl symmetry goup of $D_5$ type. They deduce bilinear relations for this tau function and derive explicit expressions for the dependent variables of $q\Psix$ in terms of the tau function.
Jimbo et al. \cite{JNS2017} define a tau function for $q\Psix$ as a formal series with coefficients in terms of conformal blocks. They derive expressions for the dependent variables of $q\Psix$ in terms of their tau function and conjecture bilinear relations analogues to those in \cite{tsudamasuda}. Whilst we were not able to ascertain the exact relation between our tau function and theirs, the expressions for the dependent variables of $q\Psix$  in terms of our and their tau functions are essentially identical.

Tau functions for $q$-Painlev\'e equations also appear in the physics literature; they are related to the grand canonical topological string partition functions \cite{bonelli2018quantumcurvesqdeformedpainleve}, to cluster integrable systems \cite{bershtein2024clusterreductionsmutationsqpainleve}, and to observables in string theory \cite{nosaka2023q,bonelli2022m2,Bonelli_2021}. There are several conjectures regarding their explicit relations \cite{Bershtein_2015,Bershtein_2017,bonelli2022m2,nosaka2023q}. Notably the tau functions of $q$-Painlev\'e VI and III ($q\Psix$ and $q\Pthree$) have a conjectural description in terms of conformal blocks \cite{Bershtein_2017,JNS2017}, and they satisfy conjectural bilinear relations that one can verify using the conformal block expression \cite{JNS2017,stoyan2025blowuprelationsqpainlevevi}. 
However, a proof of the expression in terms of conformal blocks requires the Fredholm determinant representation of the tau function. 
Further development of the Riemann-Hilbert and Fredholm determinant theory of discrete Painlev\'e equations may shed light on these conjectures.

Finally, we note that Gavrylenko \cite{gavrylenko2025riemannhilbertproblemsfredholmdeterminants} recently presented a construction of a $q$-Painlev\'e III$_3$ ($q$PIII$_3$) tau function as a Fredholm determinant and related it to the corresponding dependent variables of the Painlev\'e equation. We use and adapt some of the techniques introduced in \cite{gavrylenko2025riemannhilbertproblemsfredholmdeterminants} in the present paper.

\subsection{Main results}
Take a real number $q$ with $0<q<1$, as well as parameters $\theta:=(\theta_0,\theta_t,\theta_1,\theta_\infty)\in\mathbb{C}^4$, then the $q\Psix$ equation is given by
\begin{equation}\label{eq:qpvi}
q\Psix:\ \begin{cases}\  \sff\overline{\sff}&=\dfrac{(\overline{\sfg}-q^{+\theta_0}t)(\overline{\sfg}-q^{-\theta_0}t)}{(\overline{\sfg}-q^{\theta_\infty-1})(\overline{\sfg}-q^{-\theta_\infty})},\\
  \   \sfg\overline{\sfg}&=\dfrac{(\sff-q^{+\theta_t}t)(\sff-q^{-\theta_t}t)}{q(\sff-q^{+\theta_1})(\sff-q^{-\theta_1})},
    \end{cases} 
\end{equation}
where $\sff,\sfg:\mathcal{T}\rightarrow \mathbb{CP}^1$ are complex functions defined on a domain $\mathcal{T}$ invariant under multiplication by $q$ and we have used the abbreviated notation $\sff=\sff(t)$, $\sfg=\sfg(t)$, $\overline{\sff}=\sff(qt)$, $\overline{\sfg}=\sfg(qt)$, for $t\in \mathcal{T}$. In this paper, we will only consider one time domain, which is the complex plane cut along the positive real axis, $\mathbb{C}\setminus \mathbb{R}_{\geq 0}$, minus four $q$-lines,
\begin{equation}\label{eq:tdomain}
  \mathcal{T}:=(\mathbb{C}\setminus \mathbb{R}_{\geq 0})\setminus (q^\mathbb{Z}\cdot\{q^{+\theta_t+\theta_1},q^{+\theta_t-\theta_1},q^{-\theta_t+\theta_1},q^{-\theta_t-\theta_1}\}).
\end{equation}
We further make the following assumptions on the parameters,
\begin{equation}\label{eq:param_assumptions_1}
2\theta_0,2\theta_t,2\theta_1,2\theta_\infty\notin \Lambda,\qquad \Lambda:=\mathbb{Z}\cdot 1+\mathbb{Z}\cdot \frac{2\pi i}{\log q}.
\end{equation}
The removal of the four $q$-lines from the domain in \eqref{eq:tdomain}, together with parameter conditions \eqref{eq:param_assumptions_1}, ensure that the standard linear problem corresponding to $q\Psix$, given in equation \eqref{eq:laxpairA}, is non-resonant \cite{joshi_rof_qpvi}.

We introduce two further complex constants, 
\begin{equation}\label{eq:int_constants}
    s_{0t}\in\mathbb{C}^*,\quad \sigma_{0t}\in\mathbb{C},
\end{equation}
satisfying
\begin{equation}\label{eq:sigma_ineq}
    0<\Re\sigma_{0t}<\tfrac{1}{2},
\end{equation}
and
\begin{equation}\label{eq:sigmaconds}
    \sigma_{0t}\not\equiv \pm (\theta_0+\theta_t),\pm (\theta_0-\theta_t),\pm (\theta_\infty+\theta_1),\pm (\theta_\infty-\theta_1) \mod{\Lambda},
\end{equation}
where $\Lambda$ is the lattice introduced in equation \eqref{eq:param_assumptions_1}.
We will refer to $s_{0t}$ and $\sigma_{0t}$ as the twist parameter and intermediate exponent respectively.
We further make use of the following exponentiated variables,
\begin{equation}\label{eq:kappa_def}
    \kappa_0=q^{\theta_0},\quad \kappa_t=q^{\theta_t},\quad \kappa_1=q^{\theta_1},\quad \kappa_\infty=q^{-\theta_\infty},\quad \kappa_{0t}=q^{\sigma_{0t}}.
\end{equation}
Whereas $\theta$ and $q$ parametrise the $q\Psix$ equation, the twist parameter and intermediate exponent parametrise its solutions, $\sff=\sff(t;\sigma_{0t},s_{0t},\theta)$, $\sfg=\sfg(t;\sigma_{0t},s_{0t},\theta)$, through an associated Riemann-Hilbert problem, as explained in Section \ref{sec:setup}, see in particular Definition \ref{def:def_sols}. Explicit formulas for $\sff$ and $\sfg$ in terms of $\{\sigma_{0t},s_{0t},\theta\}$ are given in \cite[Theorem 2.14]{roffelsen2024segre} in the form of complete, convergent, asymptotic expansions, which we recall in Theorem \ref{thm:asymptotic}.

In this paper, we introduce, for the first time, a tau function for $q\Psix$ in the form of a Fredholm determinant, 
$\tau_{W}(t)$ in Definition \ref{def:widomconstant}. 
This tau function vanishes at a time $t$ if and only if the corresponding Riemann-Hilbert problem is point-wise not solvable there, see Proposition \ref{prop:solvabilitytau}.
As in \cites{JNS2017,tsudamasuda}, not only the tau function itself, but copies of it with different shifts of the parameters also play an important role,
\begin{equation}\label{def:tau-shifts}
\begin{split}
    \tau_1(t) &:= \tau(t;  \theta_0, \theta_t, \theta_{1}, \theta_{\infty},\sigma_{0t}, s_{0t}) \equiv \tau_{W}(t),\\
    \tau_2(t) &:= \tau(t;  \theta_0, \theta_t, \theta_{1}, \theta_{\infty}-1,\sigma_{0t}, s_{0t}),\\
          \tau_3(t) &:= \tau(t;  \theta_0+\tfrac{1}{2}, \theta_t, \theta_{1}, \theta_{\infty}-\tfrac{1}{2},\tfrac{1}{2}-\sigma_{0t}, -s_{0t}^{-1}),\\
    \tau_4(t) &:= \tau(t;  \theta_0-\tfrac{1}{2}, \theta_t, \theta_{1}, \theta_{\infty}-\tfrac{1}{2},\tfrac{1}{2}-\sigma_{0t}, -s_{0t}^{-1}),
      \end{split}
\end{equation}

Our first main result relates the geometry of the initial value space with these four tau functions.
The initial value space $\mathcal{X}_t$ of $q\Psix$ can be realised by blowing up $\mathbb{P}^1\times \mathbb{P}^1$ at eight points and removing the support of the unique anti-canonical divisor, see Definition \ref{def:initial_value_space}. Then, under appropriate numbering of the exceptional lines $E_k$, $1\leq k\leq 8$, lying above the eight points blown up, see \eqref{eq:exceptional_para}, we have the following result.
\begin{theorem}\label{thm:fundamentalrelationtauRHP}
  For any time $t_*\in \mathcal{T}$ and $1\leq j\leq 4$, the point $(\sff(t_*),\sfg(t_*))\in \mathcal{X}_t$ lies on the exceptional line $E_j$ if and only if $t=t_*$ is a zero of the tau function $\tau_j(t)$ defined in \eqref{def:tau-shifts}.
\end{theorem}

Furthermore, analogous to \cite[Theorem 3.3]{JNS2017},  the solution of $q\Psix$ can be written in terms of the tau function and its shifts \eqref{def:tau-shifts} as explained in the following two theorems.
\begin{theorem}\label{thm:g-tau}
The Painlev\'e transcendent $\sfg(t)$ is given in terms of the tau functions as 
\begin{align}\label{eqthm:g-tau}
    \sfg(t) &= \frac{\tau_2(t) \tau_1(t/q) - \tau_1(t) \tau_2(t/q)}{q^{\theta_{\infty}} \tau_2(t) \tau_1(t/q) -  q^{1-\theta_{\infty}} \tau_1(t) \tau_2(t/q)},
\end{align}
where $\tau_1$, $\tau_2$ are defined in \eqref{def:tau-shifts}.
\end{theorem}
\begin{theorem}\label{thm:f-tau}
The Painlev\'e transcendent $\sff(t)$ is given in terms of the tau functions as 
\begin{equation}\label{eqthm:f-tau}
    \sff(t) =   r_{0t}^{-1}(-t)^{1-2 \sigma_{0t}}m_{\sff}\frac{\tau_3(t) \tau_4 (t)}{\tau_2(t) \tau_1(t)},
    \end{equation}
    with constant multiplier
    \begin{equation*}
        m_{\sff}=\frac{\kappa_{0t} (\kappa_t\kappa_{0t} -\kappa_{0}) (\kappa_{0}\kappa_t \kappa_{0t} -1) (\kappa_{1}\kappa_{0t} -\kappa_{\infty})}{\kappa_{0}\kappa_t  (\kappa_{\infty}\kappa_{0t} -\kappa_{1}) (\kappa_{0t}-\kappa_{0t}^{-1})^2 },
   \end{equation*}
where $\tau_1$, $\tau_2$, $\tau_3$, $\tau_4$ are defined in \eqref{def:tau-shifts}.
\end{theorem}

An important consequence of expressing the tau function as a Fredholm determinant is that its asymptotic expansion can be computed to any order. We find the following expansion.
\begin{theorem}\label{thm:tau-asymp}
    The tau function admits an asymptotic expansion as $t\rightarrow 0$ of the form
\begin{equation}\label{eq:tau-asymp}
    \tau(t)=\sum_{k=-\infty}^\infty\sum_{n=k^2}^\infty \widehat{\tau}_{n,k} r_{0,t}^{k} (-t)^{n+ 2k \sigma_{0t}},
\end{equation}
where the coefficients $\widehat{\tau}_{n,k}=\widehat{\tau}_{n,k}(\theta,\sigma_{0t})$ are independent of $s_{0t}$, and the dependence on $s_{0t}$ solely enters the asymptotics through
\begin{equation}
    r_{0t}=c_{0t}\times s_{0t}, \label{rcs}
\end{equation}
where
\begin{equation}\label{eq:def_c0t}
c_{0t}=q^{2\sigma_{0t}(1-\theta_t-\theta_1)}\frac{\Gamma_q(1-2\sigma_{0t})^2}{\Gamma_q(1+2\sigma_{0t})^2}\prod_{\epsilon=\pm1}\frac{ \Gamma_q(1+\theta_t+\epsilon\,\theta_0+\sigma_{0t})  \Gamma_q(1+\theta_1+\epsilon\,\theta_\infty+\sigma_{0t})}{ \Gamma_q(1+\theta_t+\epsilon\,\theta_0-\sigma_{0t})  \Gamma_q(1+\theta_1+\epsilon\,\theta_\infty-\sigma_{0t})}.
\end{equation}
The first few coefficients are given by
 \begin{equation}\label{eq:tau-asymp-coeff}
 \begin{split}
 \widehat{\tau}_{0,0}&=1,\\
\widehat{\tau}_{1,1} &= \frac{q (\kappa_{0} \kappa_{0t}-\kappa_{t}) (\kappa_{0t}-\kappa_{0} \kappa_{t}) (\kappa_{1}-\kappa_{0t} \kappa_{\infty}) (\kappa_{1} \kappa_{\infty}-\kappa_{0t})}{\kappa_{0} \left(\kappa_{0t}^2-1\right)^2 \kappa_{1} \kappa_{\infty} \kappa_{t} \left(\kappa_{0t}^2 q-1\right)^2}, \\
       \widehat{\tau}_{1,0} &= \frac{1}{{\kappa_{0} \left(\kappa_{0t}^2-1\right)^2 \kappa_{1} \kappa_{\infty} \kappa_{t} (q-1)^2}}\Big(q \big(\kappa_{t} \big(2 \left(\kappa_{0}^2+1\right) \kappa_{0t}^2 \kappa_{1} \kappa_{\infty}^2+2 \left(\kappa_{0}^2+1\right) \kappa_{0t}^2 \kappa_{1}\\
        &-\kappa_{\infty} \left(\kappa_{0}^2 \kappa_{0t} \left(\kappa_{0t}^2+1\right) \left(\kappa_{1}^2+1\right)+\kappa_{0} \left(\kappa_{0t}^2-1\right)^2 (\kappa_{1}-1)^2+\left(\kappa_{0t}^3+\kappa_{0t}\right) \left(\kappa_{1}^2+1\right)\right)\big)\\
        &+\kappa_{0} \kappa_{t}^2 \left(-\left(\left(\kappa_{0t}^3+\kappa_{0t}\right) \kappa_{1} \kappa_{\infty}^2\right)-\left(\kappa_{0t}^3+\kappa_{0t}\right) \kappa_{1}+\kappa_{\infty} \left(\kappa_{0t}^4 ((\kappa_{1}-1) \kappa_{1}+1)+2 \kappa_{0t}^2 \kappa_{1}+(\kappa_{1}-1) \kappa_{1}+1\right)\right)\\
        &+\kappa_{0} \big(-\left(\left(\kappa_{0t}^3+\kappa_{0t}\right) \kappa_{1} \kappa_{\infty}^2\right)-\left(\kappa_{0t}^3+\kappa_{0t}\right) \kappa_{1}+\kappa_{\infty} \left(\kappa_{0t}^4 ((\kappa_{1}-1) \kappa_{1}+1)+2 \kappa_{0t}^2 \kappa_{1}+(\kappa_{1}-1) \kappa_{1}+1\right)\big)\Big), \\
       \widehat{\tau}_{1,-1} &=\frac{\kappa_{0t}^4 q (\kappa_{0t} \kappa_{t}-\kappa_{0}) (\kappa_{0} \kappa_{0t} \kappa_{t}-1) (\kappa_{0t} \kappa_{1}-\kappa_{\infty}) (\kappa_{0t} \kappa_{1} \kappa_{\infty}-1)}{\kappa_{0} \left(\kappa_{0t}^2-1\right)^2 \kappa_{1} \kappa_{\infty} \kappa_{t} \left(q-\kappa_{0t}^2\right)^2}.
       \end{split}
    \end{equation}
\end{theorem}

\begin{remark}\label{rem:comparisonwithliterature}
    The definition of the constant $c_{0t}$ in equation \eqref{eq:def_c0t} has an additional factor $q^{2\sigma_{0t}(1-\theta_t-\theta_1)}$ compared to \cite[Theorem 2.14]{roffelsen2024segre}. This normalising factor ensures that all the shifts of the twist parameter take the simplest possible form, $s_{0t}\mapsto \pm s_{0t}^{\pm 1}$, in \eqref{def:tau-shifts}. We note that in \cite[\S 3.3]{JNS2017}, the twist parameter is not transformed at all among the different tau functions. This difference is due to the fact that we impose  
$0<\Re\sigma_{0t}<\tfrac{1}{2}$ in \eqref{eq:sigma_ineq}, which is preserved by the shift $\sigma_{0t}\mapsto \tfrac{1}{2}-\sigma_{0t}$ that we use in the definition of $\tau_{3,4}$, but not by the shifts $\sigma_{0t}\mapsto \sigma_{0t}\pm\tfrac{1}{2}$ used in \cite[\S 3.3]{JNS2017} to define some of their tau functions. 
See Section \ref{sec:conclusion} for further comparison between our tau function and the one in \cite{JNS2017}.
\end{remark}

\subsection{Overview of the paper}
The paper is organised as follows.
We start from the RHP for $q\Psix$ in Section~\ref{sec:setup}, and recall the construction in \cite{roffelsen2024segre} to recast the RHP on a circle. In Section~\ref{sec:widom} we define the Fredholm determinant, known as the Widom constant, for the jump function of the RHP on the circle. In Section~\ref{sec:sol-asymp}, we prove Theorems \ref{thm:fundamentalrelationtauRHP}, \ref{thm:g-tau}, \ref{thm:f-tau} and \ref{thm:tau-asymp} in the respective subsections \ref{subsec:tauinitial}, \ref{subsec:proofthmg}, \ref{subsec:prooftheoremf} and \ref{subsec:proofminor}.

\subsection{Acknowledgements}
We thank Pavlo Gavrylenko and Nalini Joshi for useful discussions.  H.D's work was supported by Australian Research Council Discovery Project \#DP200100210, SMRI postdoctoral fellowship, and Marie Skłodowska-Curie Postdoctoral Fellowship \#101203697. P.R.'s research was supported by the
Australian Research Council Discovery Project \#DP210100129 and by the Australian Government through the Office of National Intelligence NISDRG Grant
\#NI240100145.

\section{Setup: The $q\Psix$ RHP recast on a circle} \label{sec:setup}
We start by introducing a Riemann-Hilbert problem for $q\Psix$ that captures its general solution. Such a Riemann-Hilbert problem was obtained in \cite{joshi_rof_qpvi} from the general Riemann-Hilbert theory of linear $q$-difference equations developed by Birkhoff \cite{birkhoff1913}. We then recast it into a Riemann-Hilbert problem posed on a circle, following \cite{roffelsen2024segre}.

\subsection{Notation}
We fix a real $0<q<1$.
The $q$-Pochhammer symbol is the (convergent) product
\begin{equation*}
(z;q)_\infty=\prod_{k=0}^{\infty}{(1-q^kz)}\qquad (z\in\mathbb{C}).
\end{equation*}
 The $\mathit{q}$-theta function 
\begin{equation}\label{eq:thetasym}
\theta_q(z)=(z;q)_\infty(q/z;q)_\infty\quad (z\in \mathbb{C}^*),\quad \mathbb{C}^*:=\mathbb{C}\setminus\{0\},
\end{equation}
is analytic on $\mathbb{C}^*$, with essential singularities at $z=0, \infty$, and has simple zeros on the $q$-line $q^\mathbb{Z}$. It satisfies the equation
\begin{equation}\label{eq:qtheta_identities}
\theta_q(qz)=-\frac{1}{z}\theta_q(z)=\theta_q(1/z).
\end{equation}

We make use of the $q$-gamma function
\begin{equation*}
    \Gamma_q(x)=(1-q)^{1-x}\frac{(q;q)_\infty}{(q^x;q)_\infty},
\end{equation*}
which satisfies
\begin{equation}\label{def:thetaq}
\theta_q(q^x)=\frac{(1-q)(q;q)_\infty^2}{\Gamma_q(x)\Gamma_q(1-x)}.
\end{equation}

For $n\in\mathbb{N}^*$, we use the common abbreviation for repeated products of these functions 
\begin{align*}
\theta_q(z_1,\ldots,z_n)&=\theta_q(z_1)\cdot \ldots\cdot \theta_q(z_n),\\
\Gamma_q(z_1,\ldots,z_n)&=\Gamma_q(z_1)\cdot\ldots\cdot \Gamma_q(z_n),\\
(z_1,\ldots,z_n;q)_\infty&=(z_1;q)_\infty\cdot\ldots\cdot (z_n;q)_\infty.
\end{align*}
We further denote
\begin{equation}\label{eq:defiFq}
 F_q(\alpha, \beta;\gamma
; z) :=\;_{2}\phi_1 \left[\begin{matrix} 
q^\alpha, q^\beta \\ 
q^\gamma \end{matrix} 
; q,z \right],
\end{equation}
where $_2\phi_1$ is the standard Heine $q$-hypergeometric function, for $\alpha,\beta,\gamma\in\mathbb{C}$ with $\gamma\notin \mathbb{Z}_{\leq 0}$.

\subsection{Set up} \label{subsec:setup}

We introduce a Riemann-Hilbert problem that captures the general solution of the $q\Psix$ equation described in \eqref{eq:qpvi}. Rather than deriving the Riemann-Hilbert problem from a Lax pair representation, as usually done in the literature, we first define the Riemann-Hilbert problem and derive the Lax pair from it, in Proposition \ref{prop:lax}.

The Riemann-Hilbert problem will depend on the parameters $\theta$ and $q$ of $q\Psix$, as well as the twist parameter $s_{0t}$ and intermediate exponent $\sigma_{0t}$ introduced in equation \eqref{eq:int_constants}. They will parametrise the general solution of $q\Psix$ under the construction that we proceed to detail.

The central ingredient of the $q\Psix$ Riemann-Hilbert problem is the connection matrix.
We provide the connection matrix in Mano-decomposed form \cite{mano_qpvi,ORS20}, based on the explicit formulas in \cite[\S 4.3]{roffelsen2024segre}, as
\begin{equation}\label{eq:connectionfactorisation}
    C(z,t)=D(t)C^i\left(\frac{z}{t}\right)(-t)^{\sigma_{0t} \sigma_3} \begin{bmatrix}r_{0t} & 0\\
    0 & 1\end{bmatrix}C^e(z),
\end{equation}
with $z$ varying in $\mathbb{C}^*$ and $t$ varies in $\mathbb{C}\setminus\mathbb{R}_{\geq 0}$. The factor $D(t)$ is the diagonal matrix given by
\begin{equation}\label{eq:defi_D}
    D(t)=q^{\frac{1}{2}u(u-1)} (-t)^{1+\theta_0\sigma_3},\qquad u=\log_q(-t),
\end{equation}
and $C^i(z)$ and $C^e(z)$ are connection matrices coming from hypergeometric systems,
\begin{align*}
     C^i(z)&=
    \begin{bmatrix}
   c_{11}^i\,\theta_q(q^{-\sigma_{0t}-\hspace{0.4mm}\theta_0\hspace{0.7mm}}z) &  c_{12}^i\,\theta_q(q^{+\sigma_{0t}-\hspace{0.4mm}\theta_0\hspace{0.7mm}}z)\\
    c_{21}^i\,\theta_q(q^{-\sigma_{0t}+\hspace{0.4mm}\theta_0\hspace{0.7mm}}z) &  c_{22}^i\,\theta_q(q^{+\sigma_{0t}+\hspace{0.4mm}\theta_0\hspace{0.7mm}}z)\\
    \end{bmatrix},\\
        C^e(z)&=
    \begin{bmatrix}
   c_{11}^e\,\theta_q(q^{-\theta_\infty+\sigma_{0t}}z) &  c_{12}^e\,\theta_q(q^{+\theta_\infty+\sigma_{0t}}z)\\
    c_{21}^e\,\theta_q(q^{-\theta_\infty-\sigma_{0t}}z) &  c_{22}^e\,\theta_q(q^{+\theta_\infty-\sigma_{0t}}z)\\
    \end{bmatrix},
\end{align*}
where the function $\theta_q$ is defined in \eqref{def:thetaq}, and the matrices of coefficients $c^i$ and $c^e$ are
\begin{align*}
    c^i&=\begin{bmatrix}
        \displaystyle\frac{\Gamma_q(+2\sigma_{0t},-2\theta_0)}{\Gamma_q(-\hspace{0.2mm}\theta_t\hspace{0.2mm}+\hspace{0.2mm}\sigma_{0t}\hspace{0.2mm}-\hspace{0.4mm}\theta_0\hspace{0.7mm},1+\hspace{0.2mm}\theta_t\hspace{0.2mm}-\hspace{0.2mm}\sigma_{0t}\hspace{0.2mm}-\hspace{0.4mm}\theta_0\hspace{0.7mm})} &
        \displaystyle\frac{\Gamma_q(-2\sigma_{0t},-2\theta_0)}{\Gamma_q(-\hspace{0.2mm}\theta_t\hspace{0.2mm}-\hspace{0.2mm}\sigma_{0t}\hspace{0.2mm}-\hspace{0.4mm}\theta_0\hspace{0.7mm},1+\hspace{0.2mm}\theta_t\hspace{0.2mm}-\hspace{0.2mm}\sigma_{0t}\hspace{0.2mm}-\hspace{0.4mm}\theta_0\hspace{0.7mm})}\\
        \displaystyle\frac{\Gamma_q(+2\sigma_{0t},+2\theta_0)}{\Gamma_q(-\hspace{0.2mm}\theta_t\hspace{0.2mm}+\hspace{0.2mm}\sigma_{0t}\hspace{0.2mm}+\hspace{0.4mm}\theta_0\hspace{0.7mm},1+\hspace{0.2mm}\theta_t+\hspace{0.2mm}\sigma_{0t}+\hspace{0.4mm}\theta_0\hspace{0.7mm})} &
        \displaystyle\frac{\Gamma_q(-2\sigma_{0t},+2\theta_0)}{\Gamma_q(-\hspace{0.2mm}\theta_t\hspace{0.2mm}-\hspace{0.2mm}\sigma_{0t}\hspace{0.2mm}+\hspace{0.4mm}\theta_0\hspace{0.7mm},1+\hspace{0.2mm}\theta_t-\hspace{0.2mm}\sigma_{0t}+\hspace{0.4mm}\theta_0\hspace{0.7mm})}
    \end{bmatrix},\\
    c^e&=\begin{bmatrix}
        \displaystyle\frac{\Gamma_q(+2\theta_\infty,+2\sigma_{0t})}{\Gamma_q(-\theta_1+\theta_\infty+\sigma_{0t},1+\theta_1+\theta_\infty+\sigma_{0t})} &
        \displaystyle\frac{\Gamma_q(-2\theta_\infty,+2\sigma_{0t})}{\Gamma_q(-\theta_1-\theta_\infty+\sigma_{0t},1+\theta_1-\theta_\infty+\sigma_{0t})}\\
        \displaystyle\frac{\Gamma_q(+2\theta_\infty,-2\sigma_{0t})}{\Gamma_q(-\theta_1+\theta_\infty-\sigma_{0t},1+\theta_1+\theta_\infty-\sigma_{0t})} &
        \displaystyle\frac{\Gamma_q(-2\theta_\infty,+2\sigma_{0t})}{\Gamma_q(-\theta_1-\theta_\infty-\sigma_{0t},1+\theta_1-\theta_\infty-\sigma_{0t})}
    \end{bmatrix}.
\end{align*}

Both $C^i(z)$ and $C^e(z)$ are analytic matrix functions on $\mathbb{C}^*$, they satisfy the $q$-difference equations
\begin{equation*}
    C^i(qz)=-\frac{1}{z}\;q^{\theta_0 \sigma_3}C^i(z)q^{\sigma_{0t}\sigma_3},\qquad
    C^e(qz)=-\frac{1}{z}\;q^{-\sigma_{0t} \sigma_3}C^e(z)q^{\theta_\infty\sigma_3},
\end{equation*}
and their determinants are given by
\begin{equation}\label{eq:cdets}
\begin{aligned}
    |C^i(z)|&=a^i\times \theta_q(q^{+\theta_t}z,q^{-\theta_t}z), & a^i&=\frac{q^{-\theta_t}(q-1)^2}{(q^{\sigma_{0t}}-q^{-\sigma_{0t}})(q^{\theta_0}-q^{-\theta_0})},\\
    |C^e(z)|&=a^e\times \theta_q(q^{+\theta_1}z,q^{-\theta_1}z),
    & a^e&=\frac{q^{-\theta_1}(q-1)^2}{(q^{\sigma_{0t}}-q^{-\sigma_{0t}})(q^{\theta_\infty}-q^{-\theta_\infty})}.
\end{aligned}
\end{equation}
Consequently, $C(z,t)$ is analytic in $z\in\mathbb{C}^*$ and $t\in \mathbb{C}\setminus\mathbb{R}_{\geq 0}$, satisfies the following $q$-difference equations with respect to $z$ and $t$,
\begin{align*}
    C(qz,t)&=\frac{t}{z^2}q^{\theta_0\sigma_3}C(z,t)q^{\theta_\infty\sigma_3},\\
    C(z,qt)&=z\;C(z,t),
\end{align*}
and its determinant is given by
\begin{equation*}
|C(z,t)|=a\; q^{u(u-1)}t^2\times \theta_q\left(q^{+\theta_t}\frac{z}{t},q^{-\theta_t}\frac{z}{t},q^{+\theta_1}z,q^{-\theta_1}z\right),
\end{equation*}
where $a=r_{0t}a^i a^e\in\mathbb{C}^*$.

The Riemann-Hilbert problem for $q\Psix$ can then be formulated as follows, recalling the definition of the time domain $\mathcal{T}$ in equation \eqref{eq:tdomain}.
\begin{rhprob}\label{rhp:standard} 
For $t\in \mathcal{T}$, find $2\times 2$ matrix-valued functions $\Psi_\infty(z,t)$ and $\Psi_0(z,t)$ which satisfy the following conditions.
  \begin{enumerate}[label={{\rm (\roman *)}}]
  \item $\Psi_\infty(z,t)$  is analytic on $\mathbb{CP}^1\setminus \{0\}$ and $\Psi_0(z,t)$ is analytic on \begin{equation}\label{eq:domain_analytic}
      \mathbb{C}\setminus(q^{\mathbb{Z}_{\leq 0}}\cdot \{q^{\theta_t}t,q^{-\theta_t}t,q^{\theta_1},q^{-\theta_1}\}).
  \end{equation}
    \item $\Psi_\infty(z,t)$ and $\Psi_0(z,t)$ are related by
\begin{equation}\label{eq:jump}
   \Psi_\infty(z,t)=\Psi_0(z,t)C(z,t).
\end{equation} 
 \item $\Psi_\infty(z,t)$ is normalised at infinity as
              \begin{equation*}
		\Psi_\infty(z,t)=I+\mathcal{O}(z^{-1})\qquad (z\rightarrow \infty).
              \end{equation*}
    \end{enumerate}
\end{rhprob}

    We cut away the four $q$-lines from the $t$-domain in equation \eqref{eq:tdomain}, which ensures that the four semi $q$-lines in equation \eqref{eq:domain_analytic} do not mutually intersect. This condition is necessary to guarantee point-wise uniqueness of solutions of the Riemann-Hilbert problem as in the following lemma.
\begin{lemma}
    For any fixed $t\in \mathcal{T}$, RHP \ref{rhp:standard} has at most one solution.
\end{lemma}
\begin{proof}
The proof is standard but we repeat it here for convenience of the reader. Given a solution $(\Psi_\infty(z,t),\Psi_0(z,t))$, we note that the determinants $|\Psi_\infty(z,t)|$ and $|\Psi_0(z,t)|$
are respectively analytic on $\mathbb{CP}^1\setminus \{0\}$ and \eqref{eq:domain_analytic}, and satisfy
\begin{align*}
  |\Psi_\infty(z,t)|&=|\Psi_0(z,t)|\times a\,q^{u(u-1)}t^2 \theta_q\left(q^{+\theta_t}\frac{z}{t},q^{-\theta_t}\frac{z}{t},q^{+\theta_1}z,q^{-\theta_1}z\right), \\
  |\Psi_\infty(z,t)&=1+\mathcal{O}(z^{-1})\qquad (z\rightarrow \infty).
\end{align*}
It follows that
\begin{equation}\label{eq:det}
\begin{aligned}
   |\Psi_\infty(z,t)|&=\left(q^{1+\theta_t}\frac{t}{z},q^{1-\theta_t}\frac{t}{z},q^{1+\theta_1}\frac{1}{z},q^{1-\theta_1}\frac{1}{z};q\right)_\infty,\\
   aq^{u(u-1)}t^2\times |\Psi_0(z,t)|&= \left(q^{+\theta_t}\frac{z}{t},q^{-\theta_t}\frac{z}{t},q^{+\theta_1}z,q^{-\theta_1}z;q\right)_\infty^{-1}.
\end{aligned}
\end{equation}
In particular, $\Psi_\infty(z,t)^{-1}$ and $\Psi_0(z,t)^{-1}$ are respectively analytic on $\mathbb{CP}^1\setminus\{\infty\}$, and
\begin{equation*}
      \mathbb{C}\setminus(q^{\mathbb{Z}_{> 0}}\cdot \{q^{\theta_t}t,q^{-\theta_t}t,q^{\theta_1},q^{-\theta_1}\}).
  \end{equation*}

It follows from the above statements and the fact that the four $q$-lines in \eqref{eq:domain_analytic} do not mutually intersect,
that, given any other solution $(\Psi_\infty'(z,t),\Psi_0'(z,t))$ of the RHP,  the matrix function
\begin{equation} \label{eq:matrixG}   G(z):=\Psi_\infty'(z,t)\Psi_\infty(z,t)^{-1}=\Psi_0'(z,t)\Psi_0(z,t)^{-1}
\end{equation}
is an analytic matrix function on $\mathbb{CP}^1$ satisfying $G(\infty)=I$. By Liouville's theorem, $G(z)\equiv I$, so that the solutions are the same and the uniqueness statement in the lemma follows.
\end{proof}

We have the following existence and uniqueness result for RHP \ref{rhp:standard}.

\begin{proposition}\label{prop:rhpsolvability}
RHP \ref{rhp:standard} has a unique solution $(\Psi_\infty(z,t),\Psi_0(z,t))$, with both components meromorphic in $t\in \mathcal{T}$. For any $t_*\in \mathcal{T}$, $\Psi_\infty(z,t)$ has a pole at $t=t_*$ if and only if $\Psi_0(z,t)$ has a pole at $t=t_*$, in which case both are regular at $t=q^{\pm 1}t_*$.
\end{proposition}
\begin{proof}
Since the jump matrix $C(z,t)$ is analytic in $t\in \mathcal{T}$,
 the analytic Fredholm alternative, see \cite[Proposition 4.3]{zhourhp} and also \cite[Corollary 3.1]{fokasitskapaev}, says that the solution of the RHP either exists for no $t\in \mathcal{T}$, or it does for almost all $t\in \mathcal{T}$ and is a meromorphic function in $t\in \mathcal{T}$. In \cite[Theorem 2.12]{joshi_rof_qpvi}, it is shown that, for any $t_*\in \mathcal{T}$, the RHP is solvable for $t=t_*$ or $t=q t_*$. The proposition follows from these two facts.
\end{proof}

We are now in position to reconstruct the Lax pair of $q\Psix$ in \cite{jimbosakai} from the solution of RHP \ref{rhp:standard}.
\begin{proposition}\label{prop:lax}
    The matrix functions
\begin{align*}
Y_0(z,t)&=z^{\log_q(-t)+\log_q(-1)}\Psi_0(z,t)z^{\theta_0\sigma_3},\\
Y_\infty(z,t)&=z^{\log_q(z/q)}\Psi_\infty(z,t) z^{-\theta_\infty\sigma_3},
\end{align*}
satisfy identical $q$-difference systems with respect to $z$ and $t$,
\begin{subequations}\label{eq:laxpair}
\begin{align}
   Y(qz,t)&=A(z,t)Y(z,t),\label{eq:laxpairA}\\
   Y(z,qt)&=B(z,t)Y(z,t). \label{eq:laxpairB}
\end{align}
\end{subequations}

The matrix function $A(z,t)$ is a degree two matrix polynomial in $z$,
\begin{equation}\label{eq:lin_sys}
   A(z,t)=A_0(t)+zA_1(t)+z^2 A_2
\end{equation}
with coefficients meromorphic in $t\in \mathcal{T}$ satisfying
\begin{equation}\label{eq:Amats}
A_2= q^{-\theta_\infty \sigma_3},\quad A_0(t)=H(t)t  q^{+\theta_0 \sigma_3}  H(t)^{-1},\quad H(t):=\Psi_0(0,t).
\end{equation}
The determinant of $A(z,t)$ is given by
\begin{equation}\label{eq:detA}
|A(z,t)|=(z-q^{+\theta_t}t)(z-q^{-\theta_t}t)(z-q^{+\theta_1})(z-q^{-\theta_1}).
\end{equation}
The function $B(z,t)$ in \eqref{eq:laxpairB} is a rational matrix function in $z$ of the form
\begin{equation*}
    B(z,t)=\frac{z^2I+zB_0(t)}{(z-q^{1+\theta_t}t)(z-q^{1-\theta_t}t)},\qquad B_0(t)=q^{2}t^{2}H(qt)H(t)^{-1},
\end{equation*}
where the matrix $B_0(t)$ is meromorphic in $t\in \mathcal{T}$, satisfying
\begin{equation*}
    |B_0(t)|=q^2 t^2,\qquad \operatorname{Tr} B_0(t)=-(q^{1+\theta_t}+q^{1-\theta_t})t.
\end{equation*}
\end{proposition}
\begin{proof}
Define the matrix function
\begin{align}
 A(z,t)&:=Y_\infty(qz,t)Y_\infty(z,t)^{-1}\nonumber\\
 &\,=  z^2 \Psi_\infty(qz,t)q^{-\theta_\infty\sigma_3}\Psi_\infty(z,t)^{-1} \label{eq:Amatrixdef}\\
 &\,= t\Psi_0(qz,t)q^{\theta_0\sigma_3}C(z,t)\Psi_\infty(z,t)^{-1}\nonumber\\
 &\,= t\Psi_0(qz,t)q^{\theta_0\sigma_3}\Psi_0(z,t)^{-1}\nonumber\\
 &\,=Y_0(qz,t)Y_0(z,t)^{-1}.\nonumber
\end{align}
It is a meromorphic function in $t\in \mathcal{T}$. Furthermore, it follows from the above expressions that $A(z,t)$ is analytic in $z\in \mathbb{C}$ and satisfies 
\begin{equation*}
    A(z,t)=z^2(q^{-\theta_\infty\sigma_3}+\mathcal{O}(z^{-1}))\qquad (z\rightarrow \infty),
\end{equation*}
and
\begin{equation*}
A(0,t)=\Psi_0(0,t)tq^{\theta_0\sigma_3}\Psi_0(0,t)^{-1}.
\end{equation*}
This shows that $A(z,t)$ is a degree two matrix polynomial in $z$. Finally the expression for the determinant of $A(z,t)$ follows from equations \eqref{eq:det} and the defining equations of $A(z,t)$,  \eqref{eq:Amatrixdef}.

Similarly, we define the $B$ matrix by
\begin{align*}
 B(z,t)&:=Y_\infty(z,qt)Y_\infty(z,t)^{-1}\\
 &\,=\Psi_\infty(z,qt)\Psi_\infty(z,t)^{-1}\\
 &\,=z\Psi_0(z,qt)C(z,t)\Psi_\infty(z,t)^{-1}\\
 &\,=z\Psi_0(z,qt)\Psi_0(z,t)^{-1}\\
 &\,=Y_0(z,qt)Y_0(z,t)^{-1}.
\end{align*}
It is a meromorphic function in $t\in \mathcal{T}$. Furthermore, it follows from the above expressions that it can only have poles in $z$ at $z=q^{1\pm \theta_t}t$, necessarily simple, and
\begin{align*}
 & &  B(z,t)&=I+\mathcal{O}(z^{-1}) &  &(z\rightarrow \infty), &\\
&  &  B(z,t)&=\mathcal{O}(z) & &(z\rightarrow 0). &
\end{align*}
Consequently, $B(z,t)$ is a rational matrix function in $z$ of the form indicated in the proposition, for some $B_0(t)$. The expressions for the determinant and trace of $B_0(t)$ follow from
\begin{equation*}
|zI+B_0(t)|= (z-q^{1+\theta_t}t)(z-q^{1-\theta_t}t),
\end{equation*}
which in turn follows from equations \eqref{eq:det}.
\end{proof}

Following Jimbo and Sakai \cite{jimbosakai}, we introduce coordinates  on $A(z,t)$ that define solutions to $q\Psix$.
\begin{definition}\label{def:def_sols}
We define $\sff=\sff(t;\sigma_{0t},s_{0t},\theta)$, $\sfg=\sfg(t;\sigma_{0t},s_{0t},\theta)$ and $\sfw=\sfw(t;\sigma_{0t},s_{0t},\theta)$ as the unique functions, meromorphic in $t\in\mathcal{T}$, specified by
\begin{subequations}\label{eq:coordinates_linear}
\begin{align}
    A_{12}(z,t)&=q^{\theta_\infty} \sfw(z-\sff),\\
    A_{22}(\sff,t)&=q(\sff-q^{+\theta_1})(\sff-q^{-\theta_1})\sfg.
\end{align}
\end{subequations}
\end{definition}
 Compatibility of the Lax pair \eqref{eq:laxpair} implies that $(\sff,\sfg)$ satisfy $q\Psix$ \cite{jimbosakai}, whereas $\sfw$ satisfies the auxiliary equation
\begin{equation}\label{eq:auxiliary}
    \frac{\overline{\sfw}}{\sfw}=\frac{q^{1-\theta_\infty}\overline{\sfg}-1}{q^{\theta_\infty}\overline{\sfg}-1}.
\end{equation}

The asymptotics of $\sff$ and $\sfg$ are described in the following theorem, which is a corollary of \cite[Thm. 2.14]{roffelsen2024segre}.
\begin{theorem} \label{thm:asymptotic}
As $t\rightarrow 0$ in $T$, $\sff$ and $\sfg$ have complete asymptotic expansions of the form,
\begin{align*}
    \sff(t)&=\sum_{n=1}^\infty\sum_{k=-n}^n F_{n,k}r_{0t}^k(- t)^{n+2k\sigma_{0t}},\\
    \sfg(t)&=\sum_{n=1}^\infty\sum_{k=-n}^n G_{n,k}r_{0t}^k(- t)^{n+2k\sigma_{0t}},
\end{align*}
which are uniformly absolutely convergent on $\{t\in \mathcal{T}:|t|<\epsilon\}$, for some $\epsilon>0$, where
\begin{align*}
    F_{1,\pm 1}&=q^{-\theta_t}\frac{\bigl(q^{\theta_t+\theta_0\mp\sigma_{0t}}-1\bigr)\bigl(q^{\theta_t-\theta_0\mp\sigma_{0t}}-1\bigr)\bigl(q^{\theta_1+\theta_\infty\mp\sigma_{0t}}-1\bigr)}{\bigl(q^{\theta_1+\theta_\infty\pm \sigma_{0t}}-1\bigr)\bigl(q^{\sigma_{0t}}-q^{-\sigma_{0t}}\bigr)^2},\\
    F_{1,0}&=\frac{2\bigl(q^{\theta_t}+q^{-\theta_t}\bigr)-\bigl(q^{\theta_0}+q^{-\theta_0}\bigr)\bigl(q^{ \sigma_{0t}}+q^{- \sigma_{0t}}\bigr)}{\bigl(q^{ \sigma_{0t}}-q^{- \sigma_{0t}}\bigr)^2},\\
    G_{1,0}&=\frac{2\bigl(q^{ \theta_0}+q^{- \theta_0}\bigr)-\bigl(q^{ \theta_t}+q^{- \theta_t}\bigr)\bigl(q^{ \sigma_{0t}}+q^{- \sigma_{0t}}\bigr)}{\bigl(q^{ \sigma_{0t}}-q^{- \sigma_{0t}}\bigr)^2}q^{-1},\\
    G_{1,\pm 1}&=-q^{-1\mp\sigma_{0t}}F_{1,\pm 1},
\end{align*}
and the higher order coefficients may be computed recursively via the $q\Psix$ equation \eqref{eq:qpvi}, where each coefficient
\begin{equation*}
    F_{n,k}=F_{n,k}(\theta_0,\theta_t,\theta_1,\theta_\infty,\sigma_{0t}),\quad
    G_{n,k}=G_{n,k}(\theta_0,\theta_t,\theta_1,\theta_\infty,\sigma_{0t})\qquad (-n\leq k \leq n, n\geq 1),
\end{equation*}
only depends on the parameters $\theta_0,\theta_t,\theta_1,\theta_\infty$ and $\sigma_{0t}$.
\end{theorem}

\begin{remark}
  Theorem \ref{thm:asymptotic} yields analytic continuation of $\sff(t)$ and $\sfg(t)$ to the universal covering of a small disc punctured at $t=0$. Subsequently, the $q\Psix$ equation yields unique meromorphic continuation of $\sff(t)$ and $\sfg(t)$ to the full universal covering of $\mathbb{C}^*$. Writing
  \begin{equation}\label{eq:solutionsfreeconstants}
      \sff(t)=\sff(t;\sigma_{0t},s_{0t}),\quad \sfg(t)=\sfg(t;\sigma_{0t},s_{0t})
  \end{equation}
  the non-linear monodromy along an anti-clockwise loop around $t=0$ is given by
 \begin{equation*}
      \sff(e^{2\pi i}t;\sigma_{0t},s_{0t})=\sff(t;\sigma_{0t},e^{4\pi i\sigma_{0t}}s_{0t}),\quad \sfg(e^{2\pi i}t;\sigma_{0t},s_{0t})=\sfg(t;\sigma_{0t},e^{4\pi i\sigma_{0t}}s_{0t}).
  \end{equation*}
\end{remark}

\begin{remark}
   The solution in Theorem \ref{thm:asymptotic}, with two free integration constants $\{\sigma_{0t},s_{0t}\}$, captures the general solution of $q\Psix$. Namely, for fixed $t$, when varying the integration constants, a dense open subset of the initial value space $\mathcal{X}_t$, introduced in Definition \ref{def:initial_value_space} is traced out. This follows from \cite[Prop. 2.6 and Thm. 2.14]{roffelsen2024segre}. 
\end{remark}

\begin{remark}
    The set up  of the Riemann-Hilbert problem, and particular the choice of connection matrix, was designed such that the asymptotics around $t=0$ are described by integration constants $\{\sigma_{0t},s_{0t}\}$ that are true constants as opposed to mere $q$-periodic functions of $t$. Whilst the asymptotics of the solutions around infinity can in principle be obtained from \cite[Thm. 2.23]{roffelsen2024segre}, they involve integration constants that are $q$-periodic functions of $t$. Explicit formulas for these integration constants involve inverse elliptic functions.
\end{remark}

\subsection{Decomposed RHP} Following \cite{roffelsen2024segre}, we transform the original RHP into one with the main jump matrix decomposed over two jumps. Such a decomposed form of the RHP also appears in \cite[\S 3.1]{JNS2017}.
To this end, we introduce a third solution to the Lax pair \eqref{eq:laxpair}. Recalling the jump condition \eqref{eq:jump} and the factorisation of the connection matric $C(z,t)$ in \eqref{eq:connectionfactorisation}, we define a matrix function $\Psi_{0t}(z,t)$ by
\begin{equation}\label{eq:jump_split}
  \begin{aligned}
    \Psi_\infty(z,t)&=\Psi_{0t}(z,t)C^e(z),\\
    \Psi_{0t}(z,t)&=\Psi_{0}(z,t)D(t)C^i\left(\frac{z}{t}\right)(-t)^{\sigma_{0t} \sigma_3} \begin{bmatrix}r_{0t} & 0\\
    0 & 1\end{bmatrix}.
\end{aligned}  
\end{equation}
 Then
 \begin{equation*}
     Y_{0t}(z,t)=e^{\pi i\log_q(z)}z^{\frac{1}{2}\log_q(z/q)}\Psi_{0t}(z,t)z^{-\sigma_{0t}\sigma_3},
 \end{equation*}
is another solution of \eqref{eq:laxpair}. In other words,
\begin{align*}
    A(z,t)&=-z\;\Psi_{0t}(qz,t)q^{-\sigma_{0t}\sigma_3}\Psi_{0t}(z,t)^{-1},\\
    B(z,t)&=\Psi_{0t}(z,qt)\Psi_{0t}(z,t)^{-1}.
\end{align*}

It follows from equations \eqref{eq:det} that the determinant of $\Psi_{0t}(z,t)$ is given by
\begin{equation*}
    |\Psi_{0t}(z,t)|=\frac{\left(q^{1+\theta_t}\frac{t}{z},q^{1-\theta_t}\frac{t}{z};q\right)_\infty}{a^e\left(q^{+\theta_1}z,q^{-\theta_1}z;q\right)_\infty}.
\end{equation*}
Furthermore, $\Psi_{0t}(z,t)$ is analytic on $\mathbb{C}^*$, away from the semi $q$-lines $q^{\mathbb{Z}_{\leq 0}}\cdot \{q^{\theta_1},q^{-\theta_1}\}$, and $\Psi_{0t}(z,t)$ is invertible on $\mathbb{C}^*$, away from the semi $q$-lines $q^{\mathbb{Z}_{>0}}\cdot \{q^{\theta_t}t,q^{-\theta_t}t\}$.

To formulate the decomposed RHP, we introduce some notation. For any Jordan curve $\gamma$, we denote by $D_{\text{in}}(\gamma)$ and $D_{\text{ex}}(\gamma)$ the respective in- and exterior of $\gamma$, so that
\begin{equation*}
    D_{\text{in}}(\gamma)\sqcup \gamma\sqcup D_{\text{ex}}(\gamma)=\mathbb{C}.
\end{equation*}
We choose three analytic Jordan curves $\gamma_i$, $\gamma_{0t}$, and $\gamma_e$, satisfying the following properties.
The contour $\gamma_e$ separates points on the $q$-lines $q^{\mathbb{Z}\pm\theta_1}$ as follows,
\begin{align*}
q^{\mathbb{Z}_{>0}}\cdot \{q^{\theta_1},q^{-\theta_1}\}&\subseteq D_{\text{in}}(\gamma_{e}),\\
q^{\mathbb{Z}_{\leq 0}}\cdot \{q^{\theta_1},q^{-\theta_1}\}&\subseteq D_{\text{ex}}(\gamma_{e}).
\end{align*}
The contour $\gamma_i$ scales with $t$ as $\gamma_i=t \gamma_i^0$, and separates points on the $q$-lines $q^{\mathbb{Z}\pm\theta_t}t$ as follows,
\begin{align*}
q^{\mathbb{Z}_{>0}}\cdot \{q^{\theta_t} t,q^{-\theta_t} t\}&\subseteq D_{\text{in}}(\gamma_{i}),\\
q^{\mathbb{Z}_{\leq 0}}\cdot \{q^{\theta_t} t,q^{-\theta_t} t\}&\subseteq D_{\text{ex}}(\gamma_{i}).
\end{align*}
We impose that $t$ varies in a bounded open domain
\begin{equation}\label{eq:Tdelta}
    \mathcal{T}_{\delta}:=\{t\in \mathcal{T}: |t|<\delta\},
\end{equation}
for some fixed $\delta>0$, where we recall the definition of $\mathcal{T}$ in \eqref{eq:tdomain}, and $\gamma_{0t}$ is such that
\begin{equation*}
    \gamma_i\subset D_{\text{in}}(\gamma_{0t}),\quad \gamma_e\subset D_{\text{ex}}(\gamma_{0t}),
\end{equation*}
for all $t\in \mathcal{T}_\delta$, see Figure \ref{fig:analytic_decomp_I}. By decreasing $\delta$ if necessary, we may ensure that $\gamma_{0t}$ is a circle.

\begin{figure}[ht]
	\centering
	\begin{tikzpicture}[scale=0.8]
	\draw[->] (-6,0)--(6,0) node[right]{$\Re{z}$};
	\draw[->] (0,-6)--(0,6) node[above]{$\Im{z}$};
	\tikzstyle{star}  = [circle, minimum width=3.5pt, fill, inner sep=0pt];
	\tikzstyle{starsmall}  = [circle, minimum width=3.5pt, fill, inner sep=0pt];

	\draw[domain=-1.3:6,smooth,variable=\x,red] plot ({exp(-\x*ln(2))*2*4/3*cos((pi/8) r)},{exp(-\x*ln(2))*2*4/3*sin((pi/8) r)});	
	\draw[domain=-1.3:6,smooth,variable=\x,red] plot ({exp(-\x*ln(2))*2*4/3*cos((pi/8) r)},{-exp(-\x*ln(2))*2*4/3*sin((pi/8) r)});

 \draw[blue,thick,decoration={markings, mark=at position 0.21 with {\arrow{>}}},
	postaction={decorate}] (0,0) ellipse (5.4cm and 5.4cm);

 \draw[blue,thick,decoration={markings, mark=at position 0.21 with {\arrow{>}}},
	postaction={decorate}] (0,0) ellipse (2.75cm and 2.75cm);

 \draw[black,thick,,dashed,decoration={markings, mark=at position 0.21 with {\arrow{>}}},
	postaction={decorate}] (0,0) ellipse (4.5cm and 4.5cm);

    \node[starsmall]     (or) at ({0},{0} ) {};
	\node     at ($(or)+(0.2,0.4)$) {$0$};

    \node[starsmall]     (qk1) at ({sqrt(sqrt(2))*2*3/2*cos((-pi/8) r)},{-sqrt(sqrt(2))*2*3/2*sin((-pi/8) r)} ) {};
	\node     at ($(qk1)+(0.3,-0.4)$) {$q^{1+\theta_1}$};
	\node[star]     (k1) at ({2*2*3/2*cos((-pi/8) r)},{-2*2*3/2*sin((-pi/8) r)} ) {};
	\node     at ($(k1)+(0.3,-0.3)$) {$q^{\theta_1}$};
	
	\node[starsmall]     (qkm1) at ({sqrt(sqrt(2))*2*3/2*cos((-pi/8) r)},{sqrt(sqrt(2))*2*3/2*sin((-pi/8) r)} ) {};
	\node     at ($(qkm1)+(0.3,0.4)$) {$q^{1-\theta_1}$};
	\node[star]     (km1) at ({2*2*3/2*cos((-pi/8) r)},{2*2*3/2*sin((-pi/8) r)} ) {};
	\node     at ($(km1)+(0.4,0.35)$) {$q^{-\theta_1}$};

    \draw[domain=-0.95:6,smooth,variable=\x,red] plot ({exp(-\x*ln(2))*2*5/3*cos((-3.8*pi/8+5*pi/4+\x*0) r)},{-exp(-\x*ln(2))*2*5/3*sin((-3.8*pi/8+5*pi/4+\x*0) r)});	
    \draw[domain=-0.95:6,smooth,variable=\x,red] plot ({exp(-\x*ln(2))*2*5/3*cos((-3.8*pi/8+5*pi/4+\x*0) r)},{exp(-\x*ln(2))*2*5/3*sin((-3.8*pi/8+5*pi/4+\x*0) r)});

	\node[star]     (qkt) at ({sqrt(sqrt(2))*2*4/5*cos((-3.8*pi/8+5*pi/4-1/4*0) r)},{sqrt(sqrt(2))*2*4/5*sin((-3.8*pi/8+5*pi/4-1/4*0) r)} ) {};
	\node     at ($(qkt)+(0.53,0.4)$) {$q^{1+\theta_t}t$};	
	\node[star]     (kt) at ({2*2*0.83*cos((-3.8*pi/8+5*pi/4-0) r)},{2*2*0.83*sin((-3.8*pi/8+5*pi/4-0) r)} ) {};
	\node     at ($(kt)+(0.4,0.4)$) {$q^{\theta_t}t$};

	\node[star]     (qktm) at ({sqrt(sqrt(2))*2*4/5*cos((-3.8*pi/8+5*pi/4) r)},{-sqrt(sqrt(2))*2*4/5*sin((-3.8*pi/8+5*pi/4) r)} ) {};
	\node     at ($(qktm)+(0.48,-0.38)$) {$q^{1-\theta_t}t$};
	\node[star]     (ktm) at ({2*2*0.83*cos((-3.8*pi/8+5*pi/4) r)},{-2*2*0.83*sin((-3.8*pi/8+5*pi/4) r)} ) {};
	\node     at ($(ktm)+(0.4,-0.35)$) {$q^{-\theta_t}t$};

 	\node[blue]     at ({5.8*cos((0.21*2*pi) r)},{5.8*sin((0.21*2*pi) r)}) {$\boldsymbol{\gamma_e}$};

 	\node[black]     at ({4.83*cos((0.21*2*pi) r)},{4.83*sin((0.21*2*pi) r)}) {$\boldsymbol{\gamma_{0t}}$};

 	\node[blue]     at ({3.2*cos((0.21*2*pi) r)+0.15},{3.2*sin((0.21*2*pi) r)})  {$\boldsymbol{\gamma_i}$};

 	\node[blue]     at ({5.65*cos((0.3*2*pi) r)},{5.65*sin((0.3*2*pi) r)}) {${\scriptstyle\boldsymbol{-}}$};

 	\node[black]     at ({4.75*cos((0.3*2*pi) r)},{4.75*sin((0.3*2*pi) r)}) {${\scriptstyle\boldsymbol{-}}$};

 	\node[blue]     at ({3.00*cos((0.3*2*pi) r)},{3.00*sin((0.3*2*pi) r)})  {${\scriptstyle\boldsymbol{-}}$};

 	\node[blue]     at ({5.14*cos((0.3*2*pi) r)},{5.14*sin((0.3*2*pi) r)}) {${\scriptstyle\boldsymbol{+}}$};

 	\node[black]     at ({4.25*cos((0.3*2*pi) r)},{4.25*sin((0.3*2*pi) r)}) {${\scriptstyle\boldsymbol{+}}$};

 	\node[blue]     at ({2.5*cos((0.3*2*pi) r)},{2.5*sin((0.3*2*pi) r)})  {${\scriptstyle\boldsymbol{+}}$};

	\end{tikzpicture}
	\caption{Topological representation of Jordan curves $\gamma_e$, $\gamma_{0t}$ and $\gamma_i$ relative to each other, the origin and points in the discrete $q$-lines $q^{\mathbb{Z}}\cdot x$, $x\in\{q^{\pm \theta_t}t,q^{\pm \theta_1}\}$. For the sake of simplicity, all the contours are displayed as circles, but we emphasise that $\gamma_e$ and $\gamma_i$ need not be.}
	\label{fig:analytic_decomp_I}
\end{figure}
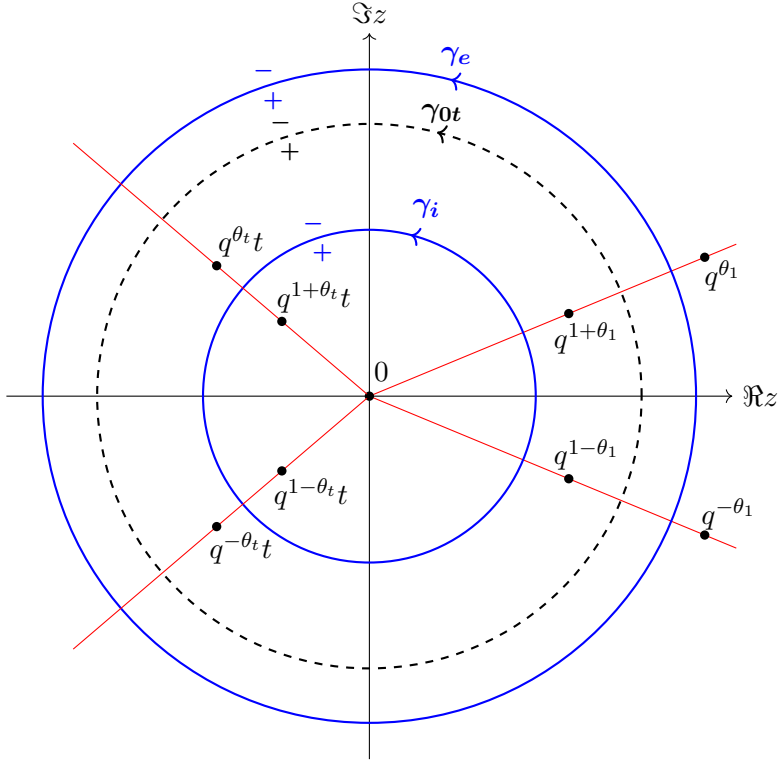

We put together the matrix functions $\Psi_0(z,t)$, $\Psi_{0t}(z,t)$ and $\Psi_\infty(z,t)$ into one piece-wise analytic matrix function,
\begin{equation*}
    \Psi(z,t)=\begin{cases}
        \Psi_0(z,t) & z\in D_{\text{in}}(\gamma_i),\\
        \Psi_{0t}(z,t) & z\in D_{\text{out}}(\gamma_i)\cap D_{\text{in}}(\gamma_e),\\
        \Psi_\infty(z,t) & z\in D_{\text{out}}(\gamma_e).
    \end{cases}
\end{equation*}
Then $\Psi(z,t)$ is a meromorpic function in $t\in \mathcal{T}_\delta$, and piece-wise analytic in $z$ on $\mathbb{CP}^1$, forming the unique solution to the following Riemann-Hilbert problem.

\begin{rhprob}\label{rhp:factorised} For $t\in \mathcal{T}_\delta$,
find a $2\times 2$ matrix-valued function $\Psi(z,t)$ which satisfies the following properties.
  \begin{enumerate}[label={{\rm (\roman *)}}]
  \item $\Psi(z,t)$ is analytic in $z$ on $\mathbb{C}\setminus (\gamma_i\sqcup \gamma_e)$, with continuous boundary values as $z$ approaches $\gamma_i$ or $\gamma_e$ from either side according to the orientations in Figure \ref{fig:analytic_decomp_I}.
    \item On the contours $\gamma_i$ and $\gamma_e$, $\Psi(z,t)$ satisfies the jump conditions
\begin{equation}\label{eq:jumprhp}
  \begin{aligned}
    \Psi_-(z,t)&=\Psi_+(z,t)C^e(z), & &(z\in \gamma_e),\\
    \Psi_-(z,t)&=\Psi_+(z,t)D(t)C^i\left(\frac{z}{t}\right)(-t)^{\sigma_{0t} \sigma_3} \begin{bmatrix}r_{0t} & 0\\
    0 & 1\end{bmatrix}, & &(z\in \gamma_i).
\end{aligned}  
\end{equation}
 \item $\Psi(z,t)$ is normalised at infinity by
              \begin{equation*}
		\Psi(z,t)=I+\mathcal{O}(z^{-1})\qquad (z\rightarrow \infty).
              \end{equation*}
    \end{enumerate}
\end{rhprob}

\subsection{RHP on a circle}\label{sec:rhpcircle}
In this section, we transform RHP \ref{rhp:factorised} into a RHP with only one jump, on the circle $\gamma_{0t}$. We accomplish this by constructing parametrices which solve the individual jump conditions in \eqref{eq:jumprhp} and then quotienting the global solution by these parametrices.

\subsubsection{Exterior parametrix}\label{subsub:exterior-param}
We start with the outer contour $\gamma_e$. So, we look for a pair of matrix-valued functions $\Psi_\infty^e(z)$ and $\Psi_0^e(z)$, analytic on the respective out and inside of $\gamma_e$, whose boundary values on $\gamma^e$ are well-defined and related by
\begin{equation*}
   \Psi_\infty^e(z)= \Psi_0^e(z)C^e(z)\quad (z\in \gamma_e),
\end{equation*}
with $\Psi_\infty^e(z)=(I+\mathcal{O}(z^{-1}))$ as $z\rightarrow \infty$.
 The unique solution to this problem, which we refer to as the external parametrix, is given by
\begin{align}
    \Psi_\infty^e(z)&=\widetilde{\Psi}_\infty^e(z) \begin{bmatrix} \big(q^{1+\theta_1}/z;q\big)_\infty & 0\\
    0 & \big(q^{1-\theta_1}/z;q\big)_\infty
    \end{bmatrix},\\    \Psi_0^e(z)&=\widetilde{\Psi}_0^e(z)\begin{bmatrix} \big(q^{-\theta_1}z;q\big)_\infty^{-1} & 0\\
    0 & \big(q^{+\theta_1}z;q\big)_\infty^{-1} 
    \end{bmatrix},\label{explicit:psi0e}
\end{align}
where, recalling the definition of the function $F_q(\cdot)$ in equation \eqref{eq:defiFq},
\begin{equation}\label{eq:psiinfedefi}
\begin{aligned}
\widetilde{\Psi}_{\infty,{11}}^e(z)&=F_q\left(-\theta_1+\theta_\infty+\sigma_{0t},-\theta_1+\theta_\infty-\sigma_{0t};2 \theta_\infty;
\frac{q^{1+\theta_1}}{z}\right),\\
\widetilde{\Psi}_{\infty,{12}}^e(z)&=\frac{r_1^e}{z}F_q\left(1+\theta_1-\theta_\infty+\sigma_{0t},1+\theta_1-\theta_\infty-\sigma_{0t};2-2\theta_\infty;
\frac{q^{1-\theta_1}}{z}\right),\\
\widetilde{\Psi}_{\infty,{21}}^e(z)&=\frac{r_2^e}{z}F_q\left(1-\theta_1+\theta_\infty+\sigma_{0t},1-\theta_1+\theta_\infty-\sigma_{0t};2+2\theta_\infty;
\frac{q^{1+\theta_1}}{z}\right),\\
\widetilde{\Psi}_{\infty,{22}}^e(z)&=F_q\left(\theta_1-\theta_\infty+\sigma_{0t},\theta_1-\theta_\infty-\sigma_{0t};-2\theta_\infty;
\frac{q^{1-\theta_1}}{z}\right),
\end{aligned}
\end{equation}
with the parameters
\begin{equation}\label{eq:defir12e}
    r_1^e=\frac{q\;\beta^e}{q^{\theta_\infty}-q^{1-\theta_\infty}}, \qquad  r_2^e=\frac{q\;\gamma^e}{q^{-\theta_\infty}-q^{1+\theta_\infty}}.
\end{equation}
Similarly, the matrix elements of $\widetilde{\Psi}_{0}^e(z)$ are
\begin{equation}
\begin{split}
\widetilde{\Psi}_{0,{11}}^e(z)&= h_{11}^eF_q\left(-\theta_1+\theta_\infty-\sigma_{0t},1-\theta_1-\theta_\infty-\sigma_{0t};1-2\sigma_{0t};
q^{+\theta_1}z\right),\\
\widetilde{\Psi}_{0,{12}}^e(z)&= h_{12}^eF_q\left(+\theta_1+\theta_\infty+\sigma_{0t},1+\theta_1-\theta_\infty+\sigma_{0t};1+2\sigma_{0t};
q^{-\theta_1}z\right),\\
\widetilde{\Psi}_{0,{21}}^e(z)&= h_{21}^eF_q\left(-\theta_1-\theta_\infty-\sigma_{0t},1-\theta_1+\theta_\infty-\sigma_{0t};1-2\sigma_{0t};
q^{+\theta_1}z\right),\\
\widetilde{\Psi}_{0,{22}}^e(z)&= h_{22}^eF_q\left(+\theta_1-\theta_\infty+\sigma_{0t},1+\theta_1+\theta_\infty+\sigma_{0t};1+2\sigma_{0t};
q^{-\theta_1}z\right),
\end{split}\label{exp:psihat0e}
\end{equation}
with the parameters
\begin{equation}\label{eq:he}
h^e=\frac{1}{1-q}\begin{bmatrix}
1-q^{\theta_1+\theta_\infty+\sigma_{0t}} & 1-q^{\theta_1+\theta_\infty-\sigma_{0t}}\\
1-q^{\theta_1-\theta_\infty+\sigma_{0t}} & 1-q^{\theta_1-\theta_\infty-\sigma_{0t}}
\end{bmatrix},
\end{equation}
and
\begin{align}
\beta^e&=\frac{(q^{\theta_1+\theta_\infty+\sigma_{0t}}-1)(q^{\theta_1+\theta_\infty-\sigma_{0t}}-1)}{q^{\theta_1}(q^{+\theta_\infty}-q^{-\theta_\infty})},\label{eq:betae}\\
\gamma^e&=\frac{(q^{\theta_1-\theta_\infty+\sigma_{0t}}-1)(q^{\theta_1-\theta_\infty-\sigma_{0t}}-1)}{q^{\theta_1}(q^{-\theta_\infty}-q^{+\theta_\infty})}.\nonumber
\end{align}

Furthermore, the linear system corresponding to the exterior parametrix
\begin{equation}\label{ext-hyp}
    Y^e(qz)=A^e(z)Y^e(z),\quad A^e(z)=A_0^e+zA_1^e,
\end{equation}
where
\begin{equation}\label{eq:A0e}
    A_0^e=\begin{bmatrix}
    \alpha^e & \beta^e\\
    \gamma^e& \delta^e\\
    \end{bmatrix},\quad
      A_1^e=q^{-\theta_\infty \sigma_3},
\end{equation}
with the constants
\begin{align*}
     \alpha^e&=\frac{q^{\theta_1}+q^{-\theta_1}-(q^{-\theta_\infty+\sigma_{0t}}+q^{-\theta_\infty-\sigma_{0t}})}{q^{-\theta_\infty}-q^{+\theta_\infty}},\\
     \delta^e&=\frac{q^{\theta_1}+q^{-\theta_1}-(q^{+\theta_\infty+\sigma_{0t}}+q^{+\theta_\infty-\sigma_{0t}})}{q^{+\theta_\infty}-q^{-\theta_\infty}},
\end{align*}
can be expressed in terms of $\Psi_\infty^e(z)$ and $\Psi_0^e(z)$ by the respective formulas,
\begin{align*}
    A^e(z)&=z\;\Psi_\infty^e(qz)q^{-\theta_\infty \sigma_3}\Psi_\infty^e(z)^{-1},\\
    &=-\Psi_0^e(qz)q^{-\sigma_{0t}\sigma_3}\Psi_0^e(z)^{-1}.
\end{align*}
We refer to the pair $(\Psi_\infty^e(z),\Psi_0^e(z))$ as the exterior parametrix.

\subsubsection{Interior parametrix}\label{subsub:interior-param}
Next, we consider the model problem defined by the jump condition on the contour $\gamma_i$. We first construct analytic matrix functions, $\Psi_\infty^i(\zeta)$ and $\Psi_0^i(\zeta)$, on the respective exterior and interior of the Jordan curve $\gamma_i^0$, whose boundary values on $\gamma_i^0$ are well-defined and related by
\begin{equation*}
   \Psi_\infty^i(\zeta)= \Psi_0^i(\zeta)C^i(\zeta)\qquad (\zeta\in \gamma_i^0),
\end{equation*}
and $\Psi_\infty^i(\zeta)=I+\mathcal{O}(\zeta^{-1})$ as $\zeta\rightarrow \infty$.
The unique solution to this problem is given by
\begin{align}
    \Psi_\infty^i(\zeta)&=\widetilde{\Psi}_\infty^i(\zeta) \begin{bmatrix} \big(q^{1+\theta_t}/\zeta;q\big)_\infty & 0\\
    0 & \big(q^{1-\theta_t}/\zeta;q\big)_\infty
    \end{bmatrix},\label{eq:internalpsiinf}\\    \Psi_0^i(\zeta)&=\widetilde{\Psi}_0^i(\zeta)\begin{bmatrix} \big(q^{-\theta_t}\zeta;q\big)_\infty^{-1} & 0\\
    0 & \big(q^{+\theta_t}\zeta;q\big)_\infty^{-1}
    \end{bmatrix},\nonumber
\end{align}
where the matrix elements of $\widetilde{\Psi}_{\infty}^i(\zeta)$ are
\begin{equation}
\begin{split}
\widetilde{\Psi}_{\infty,{11}}^i(\zeta)&=F_q\left(-\theta_t+\sigma_{0t}+\theta_0,-\theta_t+\sigma_{0t}-\theta_0;2 \sigma_{0t};
\frac{q^{1+\theta_t}}{\zeta}\right),\\
\widetilde{\Psi}_{\infty,{12}}^i(\zeta)&=\frac{r_1^i}{\zeta}F_q\left(1+\theta_t-\sigma_{0t}+\theta_0,1+\theta_t-\sigma_{0t}-\theta_0;2-2\sigma_{0t};
\frac{q^{1-\theta_t}}{\zeta}\right),\\
\widetilde{\Psi}_{\infty,{21}}^i(\zeta)&=\frac{r_2^i}{\zeta}F_q\left(1-\theta_t+\sigma_{0t}+\theta_0,1-\theta_t+\sigma_{0t}-\theta_0;2+2\sigma_{0t};
\frac{q^{1+\theta_t}}{\zeta}\right),\\
\widetilde{\Psi}_{\infty,{22}}^i(\zeta)&=F_q\left(\theta_t-\sigma_{0t}+\theta_0,\theta_t-\sigma_{0t}-\theta_0;-2\sigma_{0t};
\frac{q^{1-\theta_t}}{\zeta}\right),
\end{split}\label{exp:Psihatinfi}
\end{equation}
with the parameters
\begin{equation*}
    r_1^i=\frac{q\;\beta^i}{q^{1-\sigma_{0t}}-q^{\sigma_{0t}}}, \qquad  r_2^i=\frac{q\;\gamma^i}{q^{1+\sigma_{0t}}-q^{-\sigma_{0t}}}.
\end{equation*}
Similarly, the elements of the matrix $\widetilde{\Psi}_{0}^i(\zeta)$ are
\begin{align*}
\widetilde{\Psi}_{0,{11}}^i(\zeta)&= h_{11}^iF_q\left(-\theta_t+\theta_0+\sigma_{0t},1-\theta_t+\theta_0-\sigma_{0t};1+2\theta_0;
q^{+\theta_t}\zeta\right),\\
\widetilde{\Psi}_{0,{12}}^i(\zeta)&= h_{12}^iF_q\left(+\theta_t-\theta_0+\sigma_{0t},1+\theta_t-\theta_0-\sigma_{0t};1-2\theta_0;
q^{-\theta_t}\zeta\right),\\
\widetilde{\Psi}_{0,{21}}^i(\zeta)&= h_{21}^iF_q\left(-\theta_t+\theta_0-\sigma_{0t},1-\theta_t+\theta_0+\sigma_{0t};1+2\theta_0;
q^{+\theta_t}\zeta\right),\\
\widetilde{\Psi}_{0,{22}}^i(\zeta)&= h_{22}^iF_q\left(+\theta_t-\theta_0-\sigma_{0t},1+\theta_t+\sigma_{0t}-\theta_0;1-2\theta_0;
q^{-\theta_t}\zeta\right),
\end{align*}
with the parameters
\begin{equation}\label{eq:matrixhi}
h^i=\frac{1}{1-q}\begin{bmatrix}
1-q^{\theta_t-\theta_0+\sigma_{0t}} & 1-q^{\theta_t+\theta_0+\sigma_{0t}}\\
1-q^{\theta_t-\theta_0-\sigma_{0t}} & 1-q^{\theta_t+\theta_0-\sigma_{0t}}
\end{bmatrix},
\end{equation}
and
\begin{align}
\beta^i&=\frac{(q^{\theta_t+\theta_0+\sigma_{0t}}-1)(q^{\theta_t-\theta_0+\sigma_{0t}}-1)}{q^{\theta_t}(q^{-\sigma_{0t}}-q^{+\sigma_{0t}})},\label{eq:betai}\\
\gamma^i&=\frac{(q^{\theta_t+\theta_0-\sigma_{0t}}-1)(q^{\theta_t-\theta_0-\sigma_{0t}}-1)}{q^{\theta_t}(q^{+\sigma_{0t}}-q^{-\sigma_{0t}})}.\nonumber
\end{align}

Furthermore, the linear system associated to the interior parametrix is given by
\begin{equation}\label{eq:internal_linear_system}
    Y^i(q\zeta)=A^i(\zeta)Y^i(\zeta),\quad A^i(\zeta)=A_0^i+\zeta A_1^i,
\end{equation}
where
\begin{equation*}
    A_0^i=\begin{bmatrix}
    \alpha^i & \beta^i\\
    \gamma^i& \delta^i\\
    \end{bmatrix},\quad
      A_1=-q^{-\sigma_{0t} \sigma_3},
\end{equation*}
with the constants
\begin{align*}
     \alpha^i&=\frac{q^{\theta_t}+q^{-\theta_t}-(q^{-\sigma_{0t}+\theta_0}+q^{-\sigma_{0t}-\theta_0})}{q^{+\sigma_{0t}}-q^{-\sigma_{0t}}},\\
     \delta^i&=\frac{q^{\theta_t}+q^{-\theta_t}-(q^{+\sigma_{0t}+\theta_0}+q^{+\sigma_{0t}-\theta_0})}{q^{-\sigma_{0t}}-q^{+\sigma_{0t}}},
\end{align*}
can be expressed in terms of $\Psi_\infty^i(\zeta)$ and $\Psi_0^i(\zeta)$ by the respective formulas,
\begin{align*}
    A^i(\zeta)&=\zeta\;\Psi_\infty^i(q\zeta)q^{-\sigma_{0t} \sigma_3}\Psi_\infty^i(\zeta)^{-1},\\
&=\Psi_0^i(q\zeta)q^{\theta_0\sigma_3}\Psi_0^i(\zeta)^{-1}.
\end{align*}
Finally,  we set $\zeta=z/t$, so that
\begin{equation*}
   \Psi_\infty^i(z/t)= \Psi_0^i(z/t)C^i(z/t).
\end{equation*}

\subsubsection{Global Quotient}
We now quotient the global solution $\Psi(z,t)$ of RHP \ref{rhp:factorised} by the exterior and interior parametrices as follows,
\begin{equation}\label{def:Phi}
\Phi(z,t)=\begin{cases}
    \Psi(z,t)\Psi_\infty^e(z)^{-1} & z\in D_{\text{ex}}(\gamma_e),\\
    \Psi(z,t)\Psi_0^e(z)^{-1} & z\in D_{\text{in}}(\gamma_e)\cap D_{\text{ex}}(\gamma_{0t}),\\
    \Psi(z,t)[\Psi_\infty^i(z/t)(-t)^{\sigma_{0t}\sigma_3} \operatorname{diag}(r_{0t},1)]^{-1} & z\in D_{\text{ex}}(\gamma_i)\cap D_{\text{in}}(\gamma_{0t}),\\
    \Psi(z,t)D(t)\Psi_0^i(z/t)^{-1} & z\in D_{\text{in}}(\gamma_i),
\end{cases}
\end{equation}

where $\operatorname{diag}(r_{0t},1)$ denotes the diagonal matrix with entries $r_{0t}$ and $1$.

Note that $\Phi(z,t)$ is a meromorphic function in $t\in \mathcal{T}_\delta$. With respect to the $z$-variable, $\Phi(z,t)$ has no jumps on $\gamma_e$ and $\gamma_i$ and thus extends to an analytic function in $z$ on $\mathbb{C}\setminus \gamma_{0t}$. Furthermore, the limiting values of $\Phi(z,t)$ as $z$ approaches the circle $\gamma_{0t}$ from the minus and plus sides, according to the orientation in Figure \ref{fig:single_I}, are given by
\begin{align*}
    \Phi_{-}(z,t) &= \Psi(z,t) \Psi_0^e(z)^{-1}, \\
    \Phi_{+}(z,t) &= \Psi(z,t) [\Psi_\infty^i(z/t)(-t)^{\sigma_{0t}\sigma_3} \operatorname{diag}(r_{0t},1)]^{-1},
\end{align*}
and therefore

\begin{align}
        J(z,t):&=\Phi_+(z,t)^{-1}\Phi_-(z,t)\label{eq:factorisationdirect}\\
        &= \Psi_\infty^i(z/t)(-t)^{\sigma_{0t}\sigma_3}\begin{bmatrix}r_{0t} & 0\\
    0 & 1\end{bmatrix}\Psi_0^e(z)^{-1}.
     \label{eq:jumpJ}
\end{align}

It follows that $\Phi(z,t)$ is the unique solution to the following Riemann-Hilbert problem.

\begin{rhprob}\label{rhp:circle} For $t\in \mathcal{T}_\delta$,
find a $2\times 2$ matrix-valued function $\Phi(z,t)$ which satisfies the following properties.
  \begin{enumerate}[label={{\rm (\roman *)}}]
  \item $\Phi(z,t)$ is analytic in $z$ on $\mathbb{C}\setminus \gamma_{0t}$, with continuous boundary values as $z$ approaches $\gamma_{0t}$ from either side according to the orientations in Figure \ref{fig:single_I}.
    \item $\Phi(z,t)$ satisfies the following jump condition on $\gamma_{0t}$,
\begin{equation}\label{eq:jumprhp}
    \Phi_-(z,t)=\Phi_+(z,t)J(z,t),
\end{equation}
where the jump matrix is given in equation \eqref{eq:jumpJ}.
 \item $\Phi(z,t)$ is normalised at infinity by
              \begin{equation*}
		\Phi(z,t)=I+\mathcal{O}(z^{-1})\qquad (z\rightarrow \infty).
              \end{equation*}
    \end{enumerate}
\end{rhprob}

\begin{figure}[ht]
	\centering
	\begin{tikzpicture}[scale=0.8]
	\draw[->] (-6,0)--(6,0) node[right]{$\Re{z}$};
	\draw[->] (0,-6)--(0,6) node[above]{$\Im{z}$};
	\tikzstyle{star}  = [circle, minimum width=3.5pt, fill, inner sep=0pt];
	\tikzstyle{starsmall}  = [circle, minimum width=3.5pt, fill, inner sep=0pt];
	\tikzstyle{dot}  = [circle, minimum width=2.5pt, fill, inner sep=0pt];

	\draw[domain=-1.3:6,smooth,variable=\x,red] plot ({exp(-\x*ln(2))*2*4/3*cos((pi/8) r)},{exp(-\x*ln(2))*2*4/3*sin((pi/8) r)});	
	\draw[domain=-1.3:6,smooth,variable=\x,red] plot ({exp(-\x*ln(2))*2*4/3*cos((pi/8) r)},{-exp(-\x*ln(2))*2*4/3*sin((pi/8) r)});

 \draw[black,thick,,dashed,decoration={markings, mark=at position 0.21 with {\arrow{>}}},
	postaction={decorate}] (0,0) ellipse (3.7cm and 3.7cm);

    \node[starsmall]     (or) at ({0},{0} ) {};
	\node     at ($(or)+(0.2,0.4)$) {$0$};

	\node[star]     (k1) at ({4.4*cos((-pi/8) r)},{-4.4*sin((-pi/8) r)} ) {};
	\node     at ($(k1)+(0.3,-0.3)$) {$q^{\theta_1}$};

	\node[star]     (km1) at ({4.4*cos((-pi/8) r)},{4.4*sin((-pi/8) r)} ) {};
	\node     at ($(km1)+(0.4,0.35)$) {$q^{-\theta_1}$};

	\node[star]     (qik1) at ({1.3*4.4*cos((-pi/8) r)},{-1.3*4.4*sin((-pi/8) r)} ) {};
	\node     at ($(qik1)+(0.3,-0.3)$) {$q^{-1+\theta_1}$};

	\node[star]     (qikm1) at ({1.3*4.4*cos((-pi/8) r)},{1.3*4.4*sin((-pi/8) r)} ) {};
	\node     at ($(qikm1)+(0.4,0.35)$) {$q^{-1-\theta_1}$};

    \draw[domain=-0.95:6,smooth,variable=\x,red] plot ({exp(-\x*ln(2))*2*5/3*cos((-3.8*pi/8+5*pi/4+\x*0) r)},{-exp(-\x*ln(2))*2*5/3*sin((-3.8*pi/8+5*pi/4+\x*0) r)});	
    \draw[domain=-0.95:6,smooth,variable=\x,red] plot ({exp(-\x*ln(2))*2*5/3*cos((-3.8*pi/8+5*pi/4+\x*0) r)},{exp(-\x*ln(2))*2*5/3*sin((-3.8*pi/8+5*pi/4+\x*0) r)});

	\node[star]     (qkt) at ({0.9*sqrt(sqrt(2))*2*4/5*cos((-3.8*pi/8+5*pi/4-1/4*0) r)},{0.9*sqrt(sqrt(2))*2*4/5*sin((-3.8*pi/8+5*pi/4-1/4*0) r)} ) {};
	\node     at ($(qkt)+(0.53,0.4)$) {$q^{1+\theta_t}t$};	
	\node[star]     (kt) at ({0.9*2*2*0.83*cos((-3.8*pi/8+5*pi/4-0) r)},{0.9*2*2*0.83*sin((-3.8*pi/8+5*pi/4-0) r)} ) {};
	\node     at ($(kt)+(0.4,0.4)$) {$q^{\theta_t}t$};

	\node[star]     (qktm) at ({0.9*sqrt(sqrt(2))*2*4/5*cos((-3.8*pi/8+5*pi/4) r)},{-0.9*sqrt(sqrt(2))*2*4/5*sin((-3.8*pi/8+5*pi/4) r)} ) {};
	\node     at ($(qktm)+(0.48,-0.38)$) {$q^{1-\theta_t}t$};
	\node[star]     (ktm) at ({0.9*2*2*0.83*cos((-3.8*pi/8+5*pi/4) r)},{-0.9*2*2*0.83*sin((-3.8*pi/8+5*pi/4) r)} ) {};
	\node     at ($(ktm)+(0.4,-0.35)$) {$q^{-\theta_t}t$};







 	\node[black]     at ({4.07*cos((0.21*2*pi) r)+0.1},{4.07*sin((0.21*2*pi) r)}) {$\boldsymbol{\gamma_{0t}}$};

 	\node[black]     at ({4.00*cos((0.3*2*pi) r)},{4.00*sin((0.3*2*pi) r)}) {${\scriptstyle\boldsymbol{-}}$};

 	\node[black]     at ({3.40*cos((0.3*2*pi) r)},{3.40*sin((0.3*2*pi) r)}) {${\scriptstyle\boldsymbol{+}}$};

	\end{tikzpicture}
	\caption{Topological representation of Jordan curve $\gamma_{0t}$ with respect to the origin and points in the discrete half $q$-lines 
 $q^{\mathbb{Z}_{\leq 0}} \cdot q^{\pm \theta_1}$ and
 $q^{\mathbb{Z}_{<0}} \cdot q^{\pm \theta_t}t$. The $t$-domain is restricted to $T_\delta$ so that $q^{\mathbb{Z}_{<0}} \cdot q^{\pm \theta_t}t$ always lies within the inside of $\gamma_{0t}$.}
	\label{fig:single_I}
\end{figure}
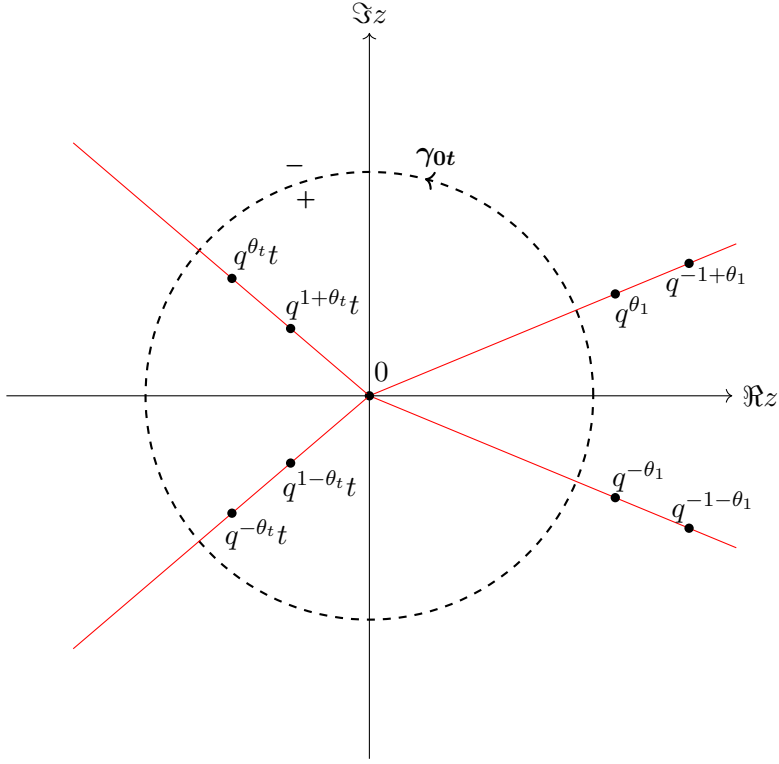

\subsection{Dual RHP and coefficients of power series}
For the purpose of studying the Widom constant associated with RHP \ref{rhp:circle}, we introduce the factorisation problem dual to RHP \ref{rhp:circle}.
\begin{rhprob}\label{rhp:circledual} For $t\in \mathcal{T}_\delta$,
find a $2\times 2$ matrix-valued function $\widehat{\Phi}(z,t)$ which satisfies the following properties.
  \begin{enumerate}[label={{\rm (\roman *)}}]
  \item $\widehat{\Phi}(z,t)$ is analytic in $z$ on $\mathbb{C}\setminus \gamma_{0t}$, with continuous boundary values as $z$ approaches $\gamma_{0t}$ from either side according to the orientations in Figure \ref{fig:single_I}.
    \item $\widehat{\Phi}(z,t)$ satisfies the following jump condition on $\gamma_{0t}$,
\begin{equation}\label{eq:jumprhp}
    \widehat{\Phi}_-(z,t)=J(z,t)\widehat{\Phi}_+(z,t),
\end{equation}
where the jump matrix is given in equation \eqref{eq:jumpJ}.
 \item $\widehat{\Phi}(z,t)$ is normalised at infinity by
              \begin{equation*}
		\widehat{\Phi}(z,t)=I+\mathcal{O}(z^{-1})\qquad (z\rightarrow \infty).
              \end{equation*}
    \end{enumerate}
\end{rhprob}

Whereas the solution of RHP \ref{rhp:circle} involves $q\Psix$ transcendents, the solution of the dual RHP is explicit in terms of the exterior and interior parametrices introduced Sections \ref{subsub:exterior-param} and \ref{subsub:interior-param} respectively, and is given by
\begin{equation}\label{def:PsiHat-pm}
\widehat{\Phi}(z,t)=\begin{cases}
\Psi_\infty^i(z/t)(-t)^{\sigma_{0t}\sigma_3}\begin{bmatrix}r_{0t} & 0\\
    0 & 1\end{bmatrix}   & z\in D_{\text{ex}}(\gamma_{0t}),\\
\Psi_0^e(z) & z\in D_{\text{in}}(\gamma_{0t}).
\end{cases}
\end{equation}

In particular, dual to \eqref{eq:factorisationdirect}, we have the factorisation of the jump matrix,
\begin{align}
        J(z,t)= \widehat{\Phi}_{-}(z,t) \widehat{\Phi}_{+}(z,t)^{-1}.\label{eq:Jumpfacder}
\end{align}

We finish the section with some explicit expansions of the solutions to RHPs \ref{rhp:circle}
 and \ref{rhp:circledual} that we require in subsequent sections.

As $z\to 0$,
\begin{equation}\label{asymp:psihatplus}
    \widehat{\Phi}(z,t)= \widehat{\Phi}_{+}^{(0)} + z\, \widehat{\Phi}_{+}^{(1)}+ \mathcal{O}(z^2),  
\end{equation}
where $\Phi_{+}^{(0)}=h^e$ is given in equation \eqref{eq:he} and
\begin{equation}\label{asymp:psihatplus0} 
 \widehat{\Phi}_{+}^{(1)}=h^e\left(\frac{1}{1-q}q^{-\theta_1 \sigma_3}+
 \begin{bmatrix}
\frac{q^{\theta _i}+q^{-\theta _i}-q^{-\theta _1-\sigma_{0t} }-q^{\theta _1+\sigma_{0t} }}{(q-1) \left(q^{\sigma_{0t} }-q^{-\sigma_{0t} }\right)} &
q^{-\sigma_{0t} } \frac{ q^{\theta _i}+q^{-\theta _i}-q^{\theta _1-\sigma_{0t} }-q^{\sigma_{0t} -\theta _1}}{\left(q^{\sigma_{0t} }-q^{-\sigma_{0t} }\right) \left(q^{\sigma_{0t} +1}-q^{-\sigma_{0t} }\right)}\\
q^{\sigma_{0t} }\frac{ q^{\theta _i}+q^{-\theta _i}-q^{\theta _1+\sigma_{0t} }-q^{-\theta _1-\sigma_{0t} }}{\left(q^{\sigma_{0t} }-q^{1-\sigma_{0t} }\right) \left(q^{\sigma_{0t} }-q^{-\sigma_{0t} }\right)} & 
\frac{q^{\theta _i}+q^{-\theta _i}-q^{\theta _1+\sigma_{0t} }-q^{-\theta _1-\sigma_{0t} }}{(q-1) \left(q^{-\sigma_{0t} }-q^{\sigma_{0t} }\right)}
 \end{bmatrix}
 \right).
\end{equation}
Also, as $z\rightarrow \infty$,
\begin{equation*}
       \widehat{\Phi}(z,t)= (-t)^{\sigma_{0t}\sigma_3}\begin{bmatrix}r_{0t} & 0\\
    0 & 1\end{bmatrix}+ \mathcal{O}(z^{-1}). \label{asymp:Psimin}  
\end{equation*}

At $z=0$, the value of $\Phi(z,t)$ is given by
\begin{equation}\label{asymp:Psiplus}
    \Phi(0,t)=H(t)D(t)(h^i)^{-1}.
\end{equation}
The matrix $h^i$ is given in \eqref{eq:matrixhi}. The matrix $H(t)=\Psi(0,t)$ diagonalises $A_0(t)$, see equation \eqref{eq:Amats}, and its determinant equals
\begin{equation*}
    |H(t)|=\frac{1}{r_{0t}a^i a^e |D(t)|},
\end{equation*}
where $a^i$ and $a^e$ are given in equation \eqref{eq:cdets} and $D(t)$ is defined in equation \eqref{eq:defi_D}. This means that $H(t)$ takes the form 
\begin{equation*}
 H(t)= q\kappa_0\kappa_t\kappa_1  \begin{bmatrix}
        \kappa_\infty-\kappa_\infty^{-1} &   (\kappa_\infty-\kappa_\infty^{-1})\sff \sfw & \\
       \sfw^{-1}u_1  & u_2
    \end{bmatrix}\cdot \begin{bmatrix}
        \epsilon_1(t) & 0\\
        0 & \epsilon_2(t)
    \end{bmatrix}\cdot D(t)^{-1},
\end{equation*}
where, recalling the $\kappa$-notation introduced in \eqref{eq:kappa_def},
\begin{align*}
    u_1=&-\frac{\left(\sff-\kappa _t t\right)\left(\sff-\kappa_t^{-1}t\right) }{q \sff \sfg }-\frac{ q \left(\sff-\kappa _1\right) \left(\sff-\kappa _1^{-1}\right) \kappa _{\infty \sfg}^2}{\sff }+\frac{ \kappa _{\infty }^2 t}{\kappa _0 \sff }+2\kappa _{\infty } \sff +\frac{\kappa _0 t}{\sff }\\
    &-\kappa_\infty\left( \kappa _t t+\frac{t}{ \kappa _t}+\kappa _1 +\frac{1}{\kappa _1 }\right),\\
    u_2=&-q \kappa _{\infty }^2\sfg \left(\sff-\kappa _1\right)\left(\sff-\kappa_1^{-1}\right)  -\frac{\left(\sff-\kappa _t t\right)\left(\sff-\kappa_t^{-1}t\right)}{q\sfg}+\kappa _{\infty }\sff^2 \\
    &-\kappa _{\infty }\sff  \left( \kappa _t t+\frac{t}{ \kappa _t}+\kappa _1 +\frac{1}{\kappa _1 }\right)+\kappa _0 \kappa _{\infty }^2 t +\frac{t}{\kappa _0},
\end{align*}
for some meromorphic functions $\epsilon_1(t)$ and $\epsilon_2(t)$ in $t\in \mathcal{T}$ that satisfy
\begin{equation}\label{id:epsilon12}
    \epsilon_1(t)\epsilon_2(t)=\frac{  \left(\kappa _{0t}-\kappa _{0t}^{-1}\right)^2}{q^2(q-1)^4\kappa_0^2\kappa_t\kappa_1  \left(\kappa _{\infty }^2-1\right)r_{0t} t}.
\end{equation}
In particular, from the above expressions we note that $\Phi(z,t)$ has the following behaviour near $z=0$,
\begin{align}\label{asymp:Psiplus}
    \Phi(z,t)&=  \Phi^{(0)} + z \Phi^{(1)} + \mathcal{O}(z^2),  \\
\end{align}
where
\begin{equation}\label{asymp:Phiplus0}
\begin{split}
    \Phi^{(0)}_{11} &= \frac{\kappa_0\kappa_1\left(\kappa_{\infty}^2-1\right) (q-1) q (\kappa_0 (\kappa_0\kappa_{t}-\kappa_{0t})\epsilon_1(t)+ (\kappa_0\kappa_{0t}-\kappa_{t})\sff(t) \sfw(t) \epsilon_2(t))}{\kappa_{\infty}\left(\kappa_{0}^2-1\right) \left(\kappa_{0t}^2-1\right) },\\
    \Phi^{(0)}_{12} &= -\frac{\kappa_0\kappa_{0t}\kappa_1\left(\kappa_{\infty}^2-1\right) (q-1) q (\kappa_0 (\kappa_0\kappa_{0t}\kappa_{t}-1)\epsilon_1(t)+\sff(t) \sfw(t) \epsilon_2(t) (\kappa_{0}-\kappa_{0t}\kappa_{t}))}{\kappa_{\infty}\left(\kappa_{0}^2-1\right) \left(\kappa_{0t}^2-1\right) },\\
    \det \Phi_{+}^{(0)} &= \frac{\left(\kappa_{0t}^2-1\right) \kappa_{1} \left(\kappa_{\infty}^2-1\right)}{\kappa_{0t} \kappa_{\infty} (q-1)^2 r_{0t}}.
    \end{split}
\end{equation}

We will also make use of the leading order asymptotic expansion of the top-right entry of $\Phi(z,t)$ as $z\rightarrow \infty$, provided in the following proposition.
\begin{proposition}\label{Prop:asymp_Phi}
The $(1,2)$ entry of $\Phi(z,t)$ has the asymptotic expansion
\begin{equation}\label{elementPhi12}
    \Phi_{12}(z,t)=\frac{1}{\kappa_{\infty}^2-q^{-1}} \left(\frac{(\kappa_{0t} \kappa_{1}-\kappa_{\infty}) (\kappa_{0t} \kappa_{\infty}-\kappa_{1})}{\kappa_{0t} \kappa_{1} \left(\kappa_{\infty}^2-1\right)}-\sfw\right)z^{-1}+\mathcal{O}{(z^{-2})}
\end{equation}
as $z\rightarrow \infty$.
\end{proposition}

\begin{proof}
  By the definition of $\Phi(z,t)$ in equation \eqref{def:Phi}, for large enough $z$,
    \begin{align}
        \Phi(z,t) = \Psi_\infty(z,t) \Psi_\infty^e(z)^{-1}.
    \end{align}
Since both $\Psi_\infty(z,t)=\mathbb{1}+\mathcal{O}(z^{-1})$ and $\Psi_\infty^e(z)=\mathbb{1}+\mathcal{O}(z^{-1})$ as $z\rightarrow \infty$, we have
\begin{equation}\label{eq:combi1}
      \Phi_{12}(z,t) =(\Psi_\infty)_{12}(z,t)-(\Psi_\infty^e)_{12}(z)+\mathcal{O}(z^{-2}).
\end{equation}
Now, by equations \eqref{eq:psiinfedefi} and \eqref{eq:defir12e},
\begin{equation}\label{eq:combi2}
    (\Psi_\infty^e)_{12}(z)=r_1^e\, z^{-1}+\mathcal{O}(z^{-2})=-\frac{\kappa_\infty \beta^e}{\kappa_{\infty}^2-q^{-1}}\, z^{-1}+\mathcal{O}(z^{-2}),
\end{equation}
where $\beta^e$ is given in equation \eqref{eq:betae}.
To extract the leading order term of $\Psi_{12}(z,t)$ as $z\rightarrow \infty$, we write $\Psi_\infty(z,t)=\mathbb{1}+z^{-1}U(t)+\mathcal{O}(z^{-2})$, so that equation \eqref{eq:Amatrixdef} yields

    \begin{align}
        q^{-1} U(t)\kappa_{\infty}^{\sigma_3} - \kappa_{\infty}^{\sigma_3}U(t) = q A_1(t). 
    \end{align}
  Since $(A_1)_{12}(t)=\kappa_{\infty}^{-1} \sfw$, we thus get
\begin{equation}\label{eq:combi3}
    (\Psi_\infty)_{12}(z,t)=U_{12}(t)z^{-1}+\mathcal{O}(z^{-2})=-\frac{\sfw}{\kappa_{\infty}^2-q^{-1}}z^{-1}+\mathcal{O}(z^{-2}).
\end{equation}

Combining equations \eqref{eq:combi1}, \eqref{eq:combi2} and \eqref{eq:combi3}, we obtain
    \begin{equation*}
    \Phi_{12}(z,t)=\frac{1}{\kappa_{\infty}^2-q^{-1}} \left(\kappa_\infty \beta^e-\sfw\right)z^{-1}+\mathcal{O}{(z^{-2})}
\end{equation*}
as $z\rightarrow \infty$, which is the statement of the proposition.
\end{proof}

\section{Widom constant and the tau-function}\label{sec:widom}

Let us recall from equations \eqref{eq:factorisationdirect} and \eqref{eq:Jumpfacder} that the jump matrix on the circle $\mathcal{C}:=\gamma_{0t}$ of RHP \ref{rhp:circle} has the following factorisations,
\begin{align}
    J =  \Phi_{+}(z,t)^{-1} \Phi_{-}(z,t)=\widehat{\Phi}_{-}(z,t) \widehat{\Phi}_{+}(z,t)^{-1} .\label{eq:jump_fac}
\end{align}
The first factorisation involves the solution $\Phi$ of RHP \ref{rhp:circle}, which is at least as transcendental as the corresponding solution of $q\Psix$. on the other hand, the second factorisation, in terms of the solution $\widehat{\Phi}_{\pm}$ of the dual  RHP \ref{rhp:circledual}, is completely explicit in terms of hypergeometric functions.

\subsection{Toeplitz operators} 

Let $L^2(\mathcal{C}, \mathbb{C}^2)$ denote the Hilbert space of square-integrable functions with values in $\mathbb{C}^2$. It admits the orthogonal decomposition
    \begin{align}\label{eq:orthogonal_decomp}
L^2(\mathcal{C}, \mathbb{C}^2)\equiv \mathcal{H} = \mathcal{H}_{+} \oplus \mathcal{H}_{-} ,
    \end{align}
    where $\mathcal{H}_{\pm}$ are closed subspaces, called the Hardy spaces, consisting of functions that are analytic in the interior and exterior of $\mathcal{C}$ respectively.
    
We denote by $\Pi_{\pm}: L^2(\mathcal{C}, \mathbb{C}^2) \rightarrow \mathcal{H}_{\pm}$ the projection operators to $\mathcal{H}_{\pm}$.

By definition, we have
    \begin{align*}
        \Pi_{\pm}^2 = \Pi_{\pm},\quad \Pi_{+} + \Pi_{-} = \mathbb{1}.
    \end{align*}
See also \cite{CGL2017,desiraju2019tau,gavrylenko2025riemannhilbertproblemsfredholmdeterminants} for a treatment of Toeplitz operators and Widom constants defined starting from the jump conditions of RHP associated to other Painlev\'e and $q$-Painlev\'e equations.

\begin{definition} With the operators $\Pi_{\pm}$, the Toeplitz operator $T_{J}: L^2(\mathcal{C},\mathbb{C}^2) \to \mathcal{H}_{+}$ for the jump matrix $J$ given in \eqref{eq:jump_fac}, is defined as
    \begin{align}
        T_{J}f(z):= \Pi_{+} (J f)(z)=\int_{\mathcal{C}} \frac{J(w) f(w)}{w-z} \frac{dw}{2\pi i},  \label{def:Toeplitz-op}
    \end{align}
    where the above expression denotes multiplication by the jump matrix $J$, followed by projection. 
\end{definition}

The Widom constant is defined, in Definition \ref{def:Widom}, in terms of the Toeplitz operator above and its inverse. Furthermore, the invertibility of a Toeplitz operator is subject to its symbol, which in this case is the jump matrix, having zero winding number (see \cite{basor2022harold} and references therein). In the next section, we show that the jump matrix \eqref{eq:jump_fac} has zero winding number and consequently, the associated Toeplitz operator \eqref{def:Toeplitz-op} is invertible.

\subsection{Winding number of the determinant of the jump matrix}
The determinant of the jump matrix \eqref{eq:jump_fac} is given by
\begin{align}
   |J(z,t)|& \mathop{=}^{\eqref{eq:jumpJ}}  r_{0t}\vert\Psi_{\infty}^{i}(z/t)\vert \vert \left(\Psi_{0}^{e}(z)\right)^{-1} \vert\\
   &\mathop{=}^{\eqref{eq:cdets},\eqref{explicit:psi0e}, \eqref{eq:internalpsiinf}}r_{0t}a^e\times\left(\frac{z}{q^{+\theta_1}},\frac{z}{q^{-\theta_1}};q\right)_\infty
   \left(\frac{q^{1+\theta_t}t}{z},\frac{q^{1-\theta_t}t}{z};q\right)_\infty.
\end{align}
Recalling the definition of the $q$-Pochammer symbol, we may rewrite this as the convergent infinite product
\begin{equation*}
   |J(z,t)|=r_{0t}a^e\times\prod_{k=0}^\infty 
   \left(1-q^k \frac{z}{q^{+\theta_1}}\right)
   \left(1-q^k \frac{z}{q^{-\theta_1}}\right)
   \left(1-q^k \frac{q^{1+\theta_t}t}{z}\right)
   \left(1-q^k \frac{q^{1-\theta_t}t}{z}\right).
\end{equation*}
To compute the winding number of $|J(z,t)|$ along the contour $\gamma_{0t}$, see Figure \ref{fig:single_I}, we take logarithms,
\begin{align}
   \log|J(z,t)|-\log(r_{0t}a^e)=&\sum_{k=0}^\infty 
   \log\left(1-q^k \frac{z}{q^{+\theta_1}}\right)+   
   \log\left(1-q^k \frac{z}{q^{-\theta_1}}\right)+\label{eq:det_J1}\\
   &\sum_{k=0}^\infty
   \log\left(1-q^k \frac{q^{1+\theta_t}t}{z}\right)+
   \log\left(1-q^k \frac{q^{1-\theta_t}t}{z}\right).\label{eq:det_J2}
\end{align}
The series on the right-hand side converge locally uniformly away from branch-points of the individual terms. Now note, that for each individual term on the right-hand side of \eqref{eq:det_J1}, the two branch-points lie within the exterior of $\gamma_{0t}$. They are thus single-valued functions on $\gamma_{0t}$. Similarly, for each term in the second line \eqref{eq:det_J2}, its two branch-points lie in the interior of $\gamma_{0t}$, so they are also single-valued on $\gamma_{0t}$. It follows that $\log|J(z,t)|$ has no monodromy around $\gamma_{0t}$ and thus the winding number of $|J(z,t)|$ around $\gamma_{0t}$ is zero.

\subsection{Widom constant and the tau-function}
With the Toeplitz operators defined as above, we can now define the Widom constant associated to the jump of the RHP on the circle \eqref{eq:Jumpfacder}. Recall that composition $T_{J}\circ T_{J^{-1}}$ is a trace-class perturbation of the identity \cite{CGL2017}, \cite{basor2012toeplitz}, \cite{widom1974asymptotic}.

\begin{definition}\label{def:widomconstant}
  The Widom constant for the jump matrix, or symbol, $J(z,t)$ is defined as
    \begin{align}\label{def:Widom}
        \tau_{W}=\tau_{W}(t; \theta_0, \theta_t, \theta_{1}, \theta_{\infty}, \sigma_{0t}, s_{0t},):= \det_{\mathcal{H}_{+}} \left[T_{J}\circ T_{J^{-1}} \right].
    \end{align}
\end{definition}
    It is an analytic function of $t\in \mathcal{T}_\delta$, where we recall the definition of the time domain $\mathcal{T}_\delta$ in equation \eqref{eq:Tdelta}, but has analytic continuation to the whole of $\mathcal{T}$, defined in equation \eqref{eq:tdomain},

\begin{lemma}\label{prop:tau_analytic_continuation}
The Widom constant $\tau_W=\tau_W(t)$ has unique analytic continuation to the domain $\mathcal{T}$, defined in equation \eqref{eq:tdomain}.
\end{lemma}
\begin{proof}
We restricted $t$ to the domain $\mathcal{T}_\delta$, so that we could take the curve $\mathcal{C}:=\gamma_{0t}$, separating the discrete half-spirals $q^{\pm \theta_t}q^{\mathbb{Z}_{>0}}t$, from the discrete half-spirals $q^{\pm \theta_1}q^{\mathbb{Z}_{\leq 0}}$, to be a fixed circle. However, the decomposition \eqref{eq:orthogonal_decomp}, corresponding projection operators and definition of the Widom constant, remain valid when the circle is replaced by any other Jordan curve that separates these half-spirals. The Widom constant is, furthermore, independent of the choice of such curve. This means that we can analytically continue $\tau_W(t)$ as long as we stay away from points where the half-spirals above have non-trivial intersection, without acquiring monodromy going around any such points. Since such points lie in the complement of $\mathcal{T}$, $\tau_W(t)$ can be analytically continued to the whole of $\mathcal{T}$.
\end{proof}

A key feature of $\tau_{W}$ for the jump with factorisation as in \eqref{eq:Jumpfacder} is that it can be written explicitly in terms of the functions $\widehat{\Phi}_{\pm}$ that are described by certain hypergeometric functions.

\begin{proposition}
The Widom constant defined in \eqref{def:Widom} can be written explicitly in terms of the local parametrices $\widehat{\Phi}_{\pm}$ defined in \eqref{def:PsiHat-pm} as
\begin{align}
    \tau_{W} = \det_{\mathcal{H}}\left[ \mathbb{1} + U\right], && U = \begin{bmatrix}
        0 & a \\
        b & 0
    \end{bmatrix},\label{prop:Widom}
\end{align}
where the operators
\begin{align}
    a:= \Pi_{+}- \widehat{\Phi}_{+}^{-1} \Pi_{+} \widehat{\Phi}_{+}: \mathcal{H}_{-} \rightarrow \mathcal{H}_{+}, && b:= \widehat{\Phi}_{-}^{-1} \Pi_{-} \widehat{\Phi}_{-} - \Pi_{-} : \mathcal{H}_{+} \rightarrow \mathcal{H}_{-}.
\end{align}
The kernels of the above operators are \cite{CGL2017}
\begin{align}\label{exp:kernels}
   a(z,w) = \frac{ \widehat{\Phi}_{+}^{-1}(z)\widehat{\Phi}_{+}(w) -\mathbb{1} }{z-w}, && b(z,w) = \frac{ \mathbb{1}-\widehat{\Phi}_{-}^{-1}(z)\widehat{\Phi}_{-}(w) }{z-w}.
\end{align}
\end{proposition}
\begin{proof}
Starting with the definition of the Widom constant and using the factorization of the jump matrix \eqref{eq:Jumpfacder}, we can write the determinant \eqref{def:Widom} as
\begin{align*}
        \tau_W &= \det_{\mathcal{H}_{+}} \left[T_{J}\circ T_{J^{-1}} \right] \mathop{=}^{}   \det_{\mathcal{H}_{+}} \left[\Pi_{+} J \, \Pi_{+}{J}^{-1} \right] \\
        &\mathop{=} \det_{\mathcal{H}_{+}} \left[\Pi_{+} \widehat{\Phi}_{-}\widehat{\Phi}_{+}^{-1}  \, \Pi_{+} \widehat{\Phi}_{+} \widehat{\Phi}_{-}^{-1}\right] = \det_{\mathcal{H}_{+}} \left[\widehat{\Phi}_{-}^{-1}\Pi_{+} \widehat{\Phi}_{-}\widehat{\Phi}_{+}^{-1}  \, \Pi_{+} \widehat{\Phi}_{+}  \right] \\
        &= \det_{\mathcal{H}_{+}} \left[\widehat{\Phi}_{-}^{-1}\left( \mathbb{1} - \Pi_{-}\right) \widehat{\Phi}_{-}\widehat{\Phi}_{+}^{-1}  \, \Pi_{+} \widehat{\Phi}_{+}  \right]  = \det_{\mathcal{H}_{+}} \left[\mathbb{1} + \left( \widehat{\Phi}_{-}^{-1} \Pi_{-} \widehat{\Phi}_{-} - \mathbb{1}\right)\left(\mathbb{1} - \widehat{\Phi}_{+}^{-1} \Pi_{+} \widehat{\Phi}_{+}\right)  \right] \\
        & = \det_{\mathcal{H}}\left[\mathbb{1}+ U \right],
    \end{align*}
    where $U$ has the form \eqref{prop:Widom} with the operators $a$, $b$ and their kernels specified in \eqref{exp:kernels}.
\end{proof}

The operator $T_{J^{-1}}$ is invertible on $\mathcal{H}_{+}$ iff the RHP with the jump condition $\Phi_{+} = \Phi_{-} J^{-1}$ is solvable, and conversely,  $T_{J}$ is invertible on $\mathcal{H}_{+}$ iff the RHP with the jump condition $\widehat{\Phi}_{+} = J^{-1} \widehat{\Phi}_{-}$ is solvable \cite{WIDOM19761,CGL2017}. Therefore, we express the operators and their inverses as
\begin{align}
  \Pi_{+} J^{-1} = \Pi_{+} \Phi_{-}^{-1} \Phi_{+}, &&  \left(\Pi_{+} J^{-1} \right)^{-1} = \Phi_{+}^{-1}\Pi_{+} \Phi_{-}, \label{eq:TopInv}
\end{align}
\begin{align}
  \Pi_{+} J = \widehat{\Phi}_{-} \widehat{\Phi}_{+}^{-1} , &&  \left(\Pi_{+} J \right)^{-1} = \widehat{\Phi}_{+} \Pi_{+} \widehat{\Phi}_{-}^{-1}. \label{eq:BotInv}
\end{align}

We have the following fundamental relation between the Widom constant and pointwise solvability of the RHP \ref{rhp:circle}.
\begin{proposition}\label{prop:solvabilitytau}
    For $t_*\in \mathcal{T}$, the Widom constant $\tau_W=\tau_W(t)$ is zero at $t=t_*$ if and only if RHP \ref{rhp:circle} is not solvable at $t=t_*$.
\end{proposition}
The proof of this proposition follows from \cite[Theorem 2.2]{bertola2017malgrange}. Due to the above proposition, we may think of $\tau_W$ as a $\tau$-function for $q\Psix$. In the next section, we will further show that the  solution of $q$PVI can be written explicitly in terms of $\tau_W$ and some copies of it with different parameter shifts. Finally, we can obtain the minor expansion of the determinant and obtain the asymptotics in the limit $t\to 0$.

\section{$q$-painlev\'e VI solutions and asymptotics}\label{sec:sol-asymp}
In this section, we first relate the geometry of the initial value space of $q\Psix$ to the eight tau functions in \eqref{def:tau-shifts}. We then identify the tau function to be the Widom constant in \eqref{prop:Widom}:
\begin{align}\label{eq:all-taus-are-the-same}
 \tau(t;\theta_0, \theta_t, \theta_{1}, \theta_{\infty}, \sigma_{0t}, s_{0t} ) \equiv \tau_1(t) = \tau_W(t),
\end{align}
and use the Fredholm determinant representation to derive relations between the eight tau functions in \eqref{def:tau-shifts}, and consequently prove the relations in \cite[Theorem 3.3]{JNS2017} for the $q$PVI transcendents in terms of the tau functions.

\subsection{Origins of the different tau functions}\label{sec:origins}
We start by explaining where the specific parameter shifts in equation \eqref{def:tau-shifts} come from through the analytic theory of linear $q$-difference equations. For this purpose, we consider the Heine hypergeometric system
\begin{align}
    Y(qz)&=A(z)Y(z),\label{eq:hypsystem}\\
    A(z)&=\begin{bmatrix}
    \alpha & \beta\, w\\
    \gamma\, w^{-1} & \delta\\
    \end{bmatrix}+z \begin{bmatrix}
    1/\mu_1 & 0\\
    0 & 1/\mu_2\\
    \end{bmatrix},\nonumber
\end{align}
where
\begin{align*}
    \alpha=&\,\frac{x_1+x_2+\mu_2(\sigma_1+\sigma_2)}{\mu_2-\mu_1}, &    \beta=&\,\frac{(x_1+\sigma_1\mu_2)(x_2+\sigma_1\mu_1)}{\sigma_1\mu_2(\mu_2-\mu_1)},\\
    \delta=&\,\frac{x_1+x_2+\mu_1(\sigma_1+\sigma_2)}{\mu_1-\mu_2}, &
    \gamma=&\,\frac{(x_1+\sigma_1\mu_1)(x_2+\sigma_1\mu_2)}{\sigma_1\mu_1(\mu_1-\mu_2)},
\end{align*}
with the following relation among the parameters $\{\sigma_1,\sigma_2,\mu_1,\mu_2,x_1,x_2\}$,
\begin{equation*}
    \sigma_1\sigma_2\mu_1\mu_2=x_1x_2.
\end{equation*}
This condition ensures that the determinant of the coefficient matrix factorises as $|A(z)|=\mu_1^{-1}\mu_2^{-1}(z-x_1)(z-x_2)$. Thus the parameters $x_{1,2}$ encode the locations of the zeros of the determinant of $A(z)$, whereas the parameters $\sigma_{1,2}$ and $\mu_{1,2}$ are the critical exponents of the linear system around $z=0$ and $z=\infty$ respectively, and $w$ parametrises the gauge freedom by diagonals, see \cite[\S 4.2]{roffelsen2024segre} for details.

The interior and exterior parametrices, defined in Section \ref{sec:rhpcircle}, are solutions to the the Heine hypergeometric systems $Y^i(qz)=A^i(z)Y^i(z)$ and $Y^e(qz)=A^e(z)Y^e(z)$
with respective interior and exterior parameter values,
    \begin{align*}
    \sigma_1^i&=q^{+\theta_0}, &  x_1^i&=q^{+\theta_t}, &   \mu_1^i&=-q^{+\sigma_{0t}}, & w^i=&\,1,\\
    \sigma_2^i&=q^{-\theta_0}, &  x_2^i&=q^{-\theta_t}, & \mu_2^i&=-q^{-\sigma_{0t}}, & &
    \end{align*}
    and
    \begin{align*}
    \sigma_1^e&=-q^{-\sigma_{0t}}, &  x_1^e&=q^{+\theta_1}, &   \mu_1^e&=q^{+\theta_\infty}, & w^e=&\,1,\\
    \sigma_2^e&=-q^{+\sigma_{0t}}, &  x_2^e&=q^{-\theta_1}, & \mu_2^e&=q^{-\theta_\infty}. & &
    \end{align*}

Now, the Heine hypergeometric system has a large number of symmetries induced by gauge transformations. When such (local) symmetries, acting on the individual interior and exterior hypergeometric systems,  act consistently on $\{\mu_1^i,\mu_2^i\}=\{\sigma_2^e,\sigma_1^e\}$, they combine to a global symmetry of the jump matrix $J(z,t)$ of RHP \ref{rhp:circle}, under a suitable transformation of the twist parameter $s_{0t}$, and correspondingly induce a B\"acklund transformation of $q\Psix$.

We will illustrate this in detail for the global symmetry that gives rise to $\tau_3$ in \eqref{def:tau-shifts}. This symmetry arises by combining the local symmetries
\begin{align*}
    &Y^i(z)\mapsto \widetilde{Y}^i(z)=z^{+\frac{1}{2}}B^i(z)Y^i(z), &
    &A^i(z)\mapsto \widetilde{A}^i(z)=q^{+\frac{1}{2}}B^i(qz)A^i(z)B^i(z)^{-1},\\
    &Y^e(z)\mapsto \widetilde{Y}^e(z)=z^{-\frac{1}{2}}B^e(z)Y^e(z), &
    &A^e(z)\mapsto \widetilde{A}^e(z)=q^{-\frac{1}{2}}B^e(qz)A^e(z)B^e(z)^{-1},
\end{align*}
where
\begin{equation*}
    B^i(z)=z^{-1}\begin{bmatrix}
        0 & 0\\
        1 & 0
    \end{bmatrix}+B_0^i,\qquad
      B^e(z)=z\begin{bmatrix}
        1 & 0\\
        0 & 0
    \end{bmatrix}+B_0^e,\\  
\end{equation*}
with the constant matrices
\begin{align*}
    B_0^i&=\frac{q}{\sigma _1^i\left(\mu _1^i-\mu _2^i\right)  \left(\mu _1^i-q\mu _2^i\right)}\begin{bmatrix}
 \mu _2^i  \left(\sigma _1^i\mu _1^i +x_1^i\right) \left(q\sigma _1^i\mu _2^i  +x_2^i\right) & -\mu _1^i  \left(\sigma _1^i\mu _2^i +x_1^i\right) \left(q\sigma _1^i\mu _2^i  +x_2^i\right) \\
 \mu _2^i  \left(\sigma _1^i\mu _1^i +x_1^i\right) \left( \sigma _1^i\mu _1^i+x_2^i\right) & -\mu _1^i  \left(\sigma _1^i\mu _1^i +x_2^i\right) \left(\sigma _1^i\mu _2^i +x_1^i\right) \\
\end{bmatrix},\\
    B_0^e&=\frac{q}{\sigma _1^e\left(\mu _1^e-\mu _2^e\right)  \left(\mu _1^e-q\mu _2^e\right)}\begin{bmatrix}
 \mu _2^e \left(\sigma _1^e\mu _1^e +x_1^e\right) \left(\sigma _1^e\mu _1^e +x_2^e\right) & -\mu _1^e \left(\sigma _1^e\mu _1^e +x_2^e\right) \left(\sigma _1^e\mu _2^e +x_1^e\right) \\
 \mu _2 \left(\sigma _1^e\mu _1^e +x_1^e\right) \left(q \sigma _1^e\mu _2^e +x_2^e\right) & -\mu _1^e \left(\sigma _1^e\mu _2^e +x_1^e\right) \left(q \sigma _1^e\mu _2^e +x_2^e\right) 
\end{bmatrix}.
\end{align*}
These two gauge transformations act on the interior and exterior hypergeometric systems, \eqref{eq:internal_linear_system} and \eqref{ext-hyp} respectively, by the following parameter shifts:
\begin{align*}
     &\sigma_1^i\mapsto \widetilde{\sigma}_1^i=q^{+\frac{1}{2}}\sigma_1^i, & &\mu_1^i\mapsto \widetilde{\mu}_1^i=q^{+\frac{1}{2}}\mu_2^i, & &\sigma_1^e\mapsto \widetilde{\sigma}_1^e=q^{-\frac{1}{2}}\sigma_2^e, & &\mu_1^e\mapsto \widetilde{\mu}_1^e=q^{-\frac{1}{2}}\mu_1^e,\\
  &\sigma_2^i\mapsto \widetilde{\sigma}_2^i=q^{-\frac{1}{2}}\sigma_2^i, & &\mu_2^i\mapsto \widetilde{\mu}_2^i=q^{-\frac{1}{2}}\mu_1^i, & &\sigma_2^e\mapsto \widetilde{\sigma}_2^e=q^{+\frac{1}{2}}\sigma_1^e, & &\mu_2^e\mapsto \widetilde{\mu}_2^e=q^{+\frac{1}{2}}\mu_2^e,
\end{align*}
which is equivalent to
\begin{equation}\label{eq:thetashift}
   \theta_0\mapsto \widetilde{\theta}_0=\theta_0+\tfrac{1}{2},\quad
   \theta_t\mapsto \widetilde{\theta}_t=\theta_t,\quad
   \theta_1\mapsto \widetilde{\theta}_1=\theta_1,\quad
   \theta_\infty\mapsto \widetilde{\theta}_\infty=\theta_\infty-\tfrac{1}{2},\quad \sigma_{0t}\mapsto \widetilde{\sigma}_{0t}=\tfrac{1}{2}-\sigma_{0t}.
\end{equation}

Next, recall that \eqref{rcs}: $r_{0t}=c_{0t}s_{0t}$, where $s_{0t}$ the twist parameter. To obtain an ansatz for the transformation of the twist parameter under \eqref{eq:thetashift}, we note that the scalar $c_{0t}$ transforms as
\begin{equation*}
    c_{0t}\mapsto\widetilde{c}_{0t}=\frac{m_c}{c_{0t}},\quad
    m_c=\frac{q \kappa_{0t}^2 \left(\kappa_{0t}-\frac{\kappa _{\infty }}{\kappa _1}\right) \left(\kappa_{0t}-\frac{\kappa _0}{\kappa _t}\right) \left(\kappa_{0t}-q\kappa _1 \kappa _{\infty }\right) \left(\kappa_{0t}-q\kappa _0 \kappa _t\right)}{\kappa _0 \kappa _{\infty } \left(\kappa_{0t}^2-1\right){}^2 \left(\kappa_{0t}^2-q\right){}^2}.
\end{equation*}
This motivates the ansatz
\begin{equation}\label{eq:s0tshift}
    s_{0t}\mapsto\widetilde{s}_{0t}=\frac{m_s}{s_{0t}},
\end{equation}
for a nonzero scalar $m_s$ yet to be determined. This ansatz will be justified in a moment, with the corresponding value $m_s=-1$ for the constant in \eqref{eq:s0tshift}.

We recall the definition of the jump matrix $J(z,t)$, which is given as the product of four components in equation \eqref{eq:jumpJ}, namely
$$\Psi_\infty^i(z/t),\quad (-t)^{\sigma_{0t}\sigma_3},\quad\begin{bmatrix}r_{0t} & 0\\
    0 & 1\end{bmatrix},\quad\Psi_0^e(z).$$
These four components transform as follows under the parameter shift \eqref{eq:thetashift},
\begin{align*}
\widetilde{\Psi}_\infty^i(z/t)&=B^i(z) \Psi_\infty^i(z/t)\begin{bmatrix}
        1 & 0\\
        0 & z/t
    \end{bmatrix}\begin{bmatrix}
        u_1^i & 0\\
        0 & u_2^i
    \end{bmatrix}\sigma_1,\\
    (-t)^{\widetilde{\sigma}_{0t}\sigma_3}&=(-t)^\frac{1}{2} \sigma_1 (-t)^{\sigma_{0t}\sigma_3} \sigma_1,\\
    \quad\begin{bmatrix}\widetilde{r}_{0t} & 0\\
    0 & 1\end{bmatrix}&=r_{0t}^{-1}\sigma_1\begin{bmatrix}r_{0t} & 0\\
    0 & 1\end{bmatrix}\begin{bmatrix}1 & 0\\
    0 & m_cm_s\end{bmatrix}\sigma_1,\\
    \widetilde{\Psi}_0^e(z)&=
    B^e(z) \Psi_\infty^e(z)\begin{bmatrix}
        1/z & 0\\
        0 & 1
    \end{bmatrix}\begin{bmatrix}
        u_1^e & 0\\
        0 & u_2^e
    \end{bmatrix}\sigma_1,
\end{align*}
where
\begin{align*}
    u_1^i&=1, & u_1^e&=\frac{\mu _2^e \sigma _1^e \left(q\mu _2^e-\mu _1^e\right) \left(q \sigma _1^e-\sigma _2^e\right)}{q \left(\mu _2^e \sigma _1^e+x_1^e\right) \left(q\mu _2^e  \sigma _1^e+x_2^e\right)},\\
    u_2^i&=\frac{\left(\mu _1^i-\mu _2^i\right) \sigma _1^i \left(q\mu _2^i-\mu _1^i\right)}{q\mu _1^i \left(\mu _2^i \sigma _1^i+x_1^i\right) \left(q\mu _2^i\sigma _1^i+x_2^i\right)}, & 
    u_2^e&=\frac{\mu _1^e-q\mu _2^e}{q\mu _1^e \mu _2^e  \left(\sigma _1^e-\sigma _2^e\right)}.
\end{align*}
Putting all of this together, we find that the transformed jump matrix is
\begin{equation}\label{eq:jumptransform}
    \widetilde{J}(z,t)=(-t)^{\frac{1}{2}}z \frac{u_1^i} {r_{0t}u_1^e}{}B^i(z)\Psi_\infty^i(z/t) (-t)^{\sigma_{0t}\sigma_3}\begin{bmatrix}r_{0t} & 0\\
    0 & 1\end{bmatrix}\begin{bmatrix}
        1 & 0\\
        0 & M
    \end{bmatrix}\Psi_0^e(z)^{-1}B^e(z)^{-1},
\end{equation}
with
\begin{equation*}
    M=m_c m_s\frac{u_1^e u_2^i}{u_1^i u_2^e}.
\end{equation*}

By direct calculation, we have $M=1$ if and only if $m_s=-1$, in which case \eqref{eq:jumptransform} collapses to
\begin{equation*}
    \widetilde{J}(z,t)=(-t)^{\frac{1}{2}}z \frac{u_1^i} {r_{0t}u_1^e}{}B^i(z) J(z,t)B^e(z)^{-1}.
\end{equation*}
Correspondingly, it is precisely under the transformation of the twist parameter $s_{0t}$ as in \eqref{eq:s0tshift} with $m_s=-1$ that the connection matrix $C(z,t)$ defined in equation \eqref{eq:connectionfactorisation} transforms as
\begin{equation*}
    \widetilde{C}(z,t)= D_2^{-1} \begin{bmatrix}
        z & 0\\
        0 & 1
    \end{bmatrix}C(z,t)\begin{bmatrix}
        z^{-1} & 0\\
        0 & 1
    \end{bmatrix} D_1,
\end{equation*}
for some diagonal matrices $D_{1}$ and $D_2$ which are independent of $z$ and $t$, where
\begin{equation*}
    D_1=\begin{bmatrix}
        1 & 0\\
        0 & d
    \end{bmatrix},\quad d=\frac{\left(q^{\theta _{\infty }-1}-q^{-\theta _{\infty }}\right) \left(q^{\theta _{\infty }}-q^{-\theta _{\infty }}\right)}{\left(q^{\sigma_{0t} }-q^{\theta _1+1-\theta_\infty}\right) \left(q^{\theta _{\infty }}-q^{-\theta _1-\sigma_{0t} }\right)}.
\end{equation*}
In the analytic theory of $q$-difference equations, the connection matrix is only rigidly defined up to arbitrary left multiplication by diagonal matrices, and so we may assume that $D_1=D_2\equiv D$. Then such a transformation of the connection matrix induces a Schlesinger transformation of the linear $q$-difference system \eqref{eq:laxpairA}, or monodromy preserving deformation \cite{ormerod2011}, justifying the transformation $\widetilde{s}_{0t}=-s_{0t}^{-1}$.
Indeed, under the above transformation, the solution of RHP \ref{rhp:standard} transforms as
\begin{align}
\widetilde{\Psi}_\infty(z,t)&=G(z,t)\Psi_\infty(z,t)D\begin{bmatrix}
        z^{-1} & 0\\
        0 & 1
    \end{bmatrix},\\
\widetilde{\Psi}_0(z,t)&=G(z,t)\Psi_0(z,t)D\begin{bmatrix}
        z^{-1} & 0\\
        0 & 1
    \end{bmatrix},
\end{align}
for a unique matrix polynomial $G(z,t)$ of the form
\begin{equation*}
    G(z,t)=z\begin{bmatrix}
        1 & 0\\
        0 & 0
    \end{bmatrix}+\begin{bmatrix}
        g_{11} & g_{12}\\
        g_{21} & g_{22}
    \end{bmatrix}.
\end{equation*}
The entries of the constant coefficient can be determined in the standard way, see e.g. \cite{muganfokas,mugansakka}, giving
\begin{equation*}
    g_{12}=(q^{-1}-\kappa_\infty^2)^{-1} w,\quad 
    g_{21}=(q^{-1}-\kappa_\infty^2) w^{-1},\quad
    g_{22}=0,
\end{equation*}
and
\begin{align*}
    g_{11}=\,& f^{-1}\left(\frac{t^2 \left(q \kappa _{\infty }^2+1\right)}{g \kappa _{\infty }}-\frac{\left(\kappa _0^2+1\right) q t \left(q \kappa _{\infty }^2+1\right)}{\kappa _0 \kappa _{\infty }}+g \left(q^3 \kappa _{\infty }+\frac{q^2}{\kappa _{\infty }}\right)\right)\\
    &-\frac{t \left(\kappa _t^2+1\right) \left(q \kappa _{\infty }^2+1\right)}{g \kappa _{\infty } \kappa _t}+\frac{(q+1) q \left(\kappa _t+\kappa _1 \left(\kappa _t \left(\kappa _1+t \kappa _t\right)+t\right)\right)}{\kappa _1 \kappa _t}-\frac{g \left(\kappa _1^2+1\right) q^2 \left(q \kappa _{\infty }^2+1\right)}{\kappa _1 \kappa _{\infty }}\\
    &+f\left(\frac{\frac{1}{\kappa _{\infty }}+q \kappa _{\infty }}{g}-2 q (q+1)+q^2g \left(q \kappa _{\infty }+\frac{1}{\kappa _{\infty }}\right)\right).
\end{align*}

Then, the solutions $Y_0(z,t)$ and $Y_\infty(z,t)$ of \eqref{eq:laxpair}, defined in Proposition \ref{prop:lax}, both transform as
\begin{equation*}
    Y(z)\mapsto \widetilde{Y}(z)=z^{-\frac{1}{2}}G(z,t)Y(z,t)D,
\end{equation*}
and correspondingly the coefficient matrix $A(z,t)$ of the linear $q$-difference equation \eqref{eq:laxpairA} transforms as
\begin{equation*}
    \widetilde{A}(z)=q^{-\frac{1}{2}}G(qz,t)A(z,t)G(z,t)^{-1}.
\end{equation*}
A direct computation finally gives the following B\"acklund transformation of $q\Psix$,
\begin{align}\nonumber
\widetilde{{\sf f}}&=\frac{t(\sff^2  (\sfg-\kappa _{\infty }) (q^2 \kappa _0^2  \kappa _{\infty }\sfg-1)+
 \kappa _{\infty }\sff (t-q\kappa _0 \sfg ) ( q\kappa _0 (\kappa _1+\kappa _1^{-1}) \sfg - (\kappa _t+\kappa _t^{-1}))
+ \kappa _{\infty }  (t- q\kappa _0 \sfg)^2)}
{\sff( q\kappa _0  \sff^2 (\sfg-\kappa _{\infty })^2
-q\kappa _0  \sff (\kappa _{\infty }-\sfg) ( \kappa _{\infty } (\kappa _t+\kappa_t^{-1})t-(\kappa _1+\kappa_1^{-1})\sfg  )
+ (t-q \kappa _0 \sfg) (q\kappa _0  \kappa _{\infty}^2t-\sfg))},\\
    \widetilde{\sfg}&=q^{-\frac{1}{2}}\frac{\kappa _0 \kappa _1 t (\sfg ( \kappa _t \sff-q\kappa _0  \kappa _{\infty })-\kappa _{\infty } ( \kappa _t \sff-t)) ( (\sff-q\kappa _0 \kappa _t \kappa _{\infty } )\sfg-\kappa _{\infty } (\sff-t \kappa _t))}{\sfg   \kappa _t (q \kappa _0 \kappa _{\infty }  (\sff-\kappa _1)\sfg-(q \kappa _0  \kappa _{\infty }\sfg-\kappa _1 t)) (q \kappa _0  (\kappa _1\sff-1)\sfg-(q\kappa _0 \kappa _1  \kappa _{\infty }\sff-t))}, \label{eq:BTtau1}
\end{align}
corresponding to shift in parameters related to $\tau_3$.
The parameter transformations, underlying the tau functions $\tau_j$, $j=2,4$, correspond to B\"acklund transformations in a similar way.

\subsection{Tau functions and the initial value space}\label{subsec:tauinitial}
In this section, we relate the geometry of the initial value space of $q\Psix$ to the four tau functions.

The initial value space of $q\Psix$ is obtained as follows. We start by blowing up $\{({\sf f}, {\sf g})\in\mathbb P^1\times \mathbb P^1\}$ at the eight base points,
\begin{equation}
    \begin{aligned}
   &b_1=(\infty,q^{-\theta_\infty}), & &b_3=(0,q^{-1+\theta_0}t), & &b_5=(q^{+\theta_1},\infty), & &b_7=(q^{+\theta_t}t,0), \\
  &b_2=(\infty,q^{-1+\theta_\infty}), &  &b_4=(0,q^{-1-\theta_0}t), & &b_6=(q^{-\theta_1},\infty). & &b_8=(q^{-\theta_t}t,0). 
\end{aligned}\label{eq:basepoints}
\end{equation}
Denoting the resulting space by $\overline{\mathcal{X}}_t$, the dynamical system \eqref{eq:qpvi} lifts to an isomorphism from $\overline{\mathcal{X}}_t$ to $\overline{\mathcal{X}}_{qt}$. 
The projective surface $\overline{\mathcal{X}}_t$ has a unique, effective, anti-canonical divisor  \cite{sakai2001}, whose support $D_t$ is given by the union of the strict transforms the curves ${\sf f}=0$, ${\sf f}=\infty$, ${\sf g}=0$ and ${\sf g}=\infty$.
\begin{definition}\label{def:initial_value_space}
  The initial value space of $q\Psix$ is defined as the open surface
\begin{equation*}
    \mathcal{X}_t:=\overline{\mathcal{X}}_t\setminus D_t.
\end{equation*}  
\end{definition}

Denote by $E_k$ the exceptional line corresponding to $b_k$ in $\overline{\mathcal{X}}_t$, $1\leq k\leq 8$. 
The parts of the exceptional lines away from $D_t$ are explicitly parametrised by
\begin{equation}\label{eq:exceptional_para}
    E_k=\{v_k\in\mathbb{C}:u_k=0\}\quad (1\leq k\leq 8),
\end{equation}
where each of the pairs of coordinates $(u_k,v_k)$, $1\leq k\leq 8$, comes from a bi-rational change of variables,

\begin{align*}
    &\begin{cases}
        {\sf f}=u_1^{-1},&\\
        {\sf g}=q^{-\theta_\infty}+u_1\,v_1, &
    \end{cases} &
          &\begin{cases}
        {\sf f}=u_2^{-1}, &\\
         {\sf g}=q^{-1+\theta_\infty}+u_2\, v_2, &
    \end{cases}\\
   &\begin{cases}
        {\sf f}=u_3, &\\
         {\sf g}=q^{-1+\theta_0}t+u_3\,v_3, &
    \end{cases} &
    &\begin{cases}
        {\sf f}=u_4, &\\
         {\sf g}=q^{-1-\theta_0}t+u_4\,v_4, &
    \end{cases}\\
    &\begin{cases}
        {\sf f}=q^{+\theta_1}+u_5\,v_5, &\\
         {\sf g}=u_5^{-1}, &
    \end{cases} &
    &\begin{cases}
        {\sf f}=q^{-\theta_1}+u_6\,v_6, &\\
         {\sf g}=u_6^{-1}, &
    \end{cases}\\
  &\begin{cases}
        {\sf f}=q^{+\theta_t}t+u_7\,v_7, &\\
         {\sf g}=u_7, &
    \end{cases} &
    &\begin{cases}
        {\sf f}=q^{-\theta_t}t+u_8\,v_8, &\\
         {\sf g}=u_8. &
    \end{cases}
\end{align*}

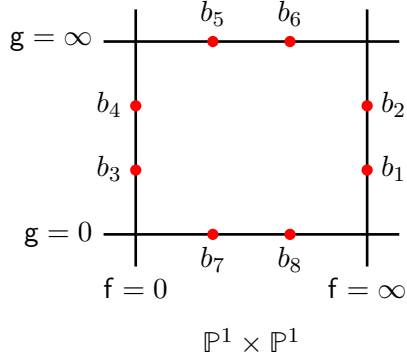
\begin{figure}[t]
\centering
	\begin{tikzpicture}[scale=.85,>=stealth,basept/.style={circle, draw=red!100, fill=red!100, thick, inner sep=0pt,minimum size=1.2mm}]
			\draw [black, line width = 1pt] 	(4.1,2.5) 	-- (-0.5,2.5)	node [left]  {$ {\sf g}=\infty$} node[pos=0, right] {};
			\draw [black, line width = 1pt] 	(0,3) -- (0,-1)			node [below] {$ {\sf f}=0$}  node[pos=0, above, xshift=-7pt] {} ;
			\draw [black, line width = 1pt] 	(3.6,3) -- (3.6,-1)		node [below]  {$ {\sf f}=\infty$} node[pos=0, above, xshift=7pt] {};
			\draw [black, line width = 1pt] 	(4.1,-.5) 	-- (-0.5,-0.5)	node [left]  {$ {\sf g}=0$} node[pos=0, right] {};

			\node (p1) at (0,0.5) [basept,label={[xshift=-10pt, yshift = -10 pt] $b_{3}$}] {};
			\node (p2) at (0,1.5) [basept,label={[xshift=-10pt, yshift = -10 pt] $b_{4}$}] {};
			\node (p3) at (1.2,-0.5) [basept,label={[yshift=-20pt, xshift=0pt] $b_{7}$}] {};
			\node (p4) at (2.4,-0.5) [basept,label={[xshift=0pt, yshift = -20 pt] $b_{8}$}] {};
			\node (p5) at (1.2,2.5) [basept,label={[xshift=0pt,yshift=0pt] $b_{5}$}] {};
			\node (p6) at (2.4,2.5) [basept,label={[xshift=0pt, yshift = 0 pt] $b_{6}$}] {};
			\node (p7) at (3.6,1.5) [basept,label={[xshift=10pt, yshift = -10 pt] $b_{2}$}] {};
			\node (p8) at (3.6,0.5) [basept,label={[xshift=10pt, yshift = -10 pt] $b_{1}$}] {};

			\node (P1P1) at (1.8, -1.8) [below] {$\mathbb{P}^1 \times \mathbb{P}^1$};
    \end{tikzpicture}

	\caption{Configuration of base points}
	\label{fig:initial}
\end{figure}

We have the following fundamental relation between the initial value space and solvability of the Riemann-Hilbert problems introduced.
\begin{proposition}\label{prop:solvabilityE8}
    For any $t_*\in \mathcal{T}$, the following are equivalent.
    \begin{enumerate}
        \item The point $(\sff(t_*),\sfg(t_*))\in \mathcal{X}_t$ lies on the exceptional line $E_1$,
        \item RHP \ref{rhp:standard} does not have a solution at $t=t_*$,
        \item RHP \ref{rhp:factorised} does not have a solution at $t=t_*$,
        \item RHP \ref{rhp:circle} does not have a solution at $t=t_*$,
        \item The solution $\Phi(z,t)$ of  RHP \ref{rhp:circle}, which is a meromorphic function in $(z,t)$ and by definition piece-wise analytic in $z$, has a pole at $t=t_*$.
    \end{enumerate}
\end{proposition}
\begin{proof}
    The equivalence between (1) and (2) was shown in \cite[Theorem 2.12]{joshi_rof_qpvi}. The equivalences between (2), (3) and (4) are a direct consequence of the way the solutions of the corresponding Riemann-Hilbert problems are related. Finally, the solution $\Phi(z,t)$ of RHP \ref{rhp:circle} is piece-wise analytic in $z\in\mathbb{P}^1$ and, as a consequence of Proposition \ref{prop:rhpsolvability}, meromorphic in $t\in \mathcal{T}$. Its poles with respect to $t$ coincide with values of $t$ where RHP \ref{rhp:circle} is not solvable, since the analogous statement is true for the solution $\Psi(z,t)$ of RHP \ref{rhp:standard}, as shown in Proposition \ref{prop:rhpsolvability}. This provides the equivalence between (4) and (5) so that the proposition follows.
\end{proof}

\begin{corollary}\label{cor:RHPSOLVE}
    For any $t_*\in \mathcal{T}$, the following are equivalent.
    \begin{enumerate}
        \item The point $(\sff(t_*),\sfg(t_*))\in \mathcal{X}_t$ lies on the exceptional line $E_1$,
        \item RHP \ref{rhp:circle} does not have a solution at $t=t_*$,
        \item The tau function $\tau_1(t)$ vanishes at $t=t_*$.
    \end{enumerate}
\end{corollary}
\begin{proof}
This is a direct consequence of Propositions \ref{prop:solvabilitytau} and \ref{prop:solvabilityE8}.
\end{proof}

We now have all the ingredients to prove Theorem \ref{thm:fundamentalrelationtauRHP}.

\begin{proof}[Proof of Theorem~\ref{thm:fundamentalrelationtauRHP}] The statement of the theorem for $j=1$ is the equivalence between (1) and (3) in Corollary \ref{cor:RHPSOLVE}. We will show how this equivalence can be used to deduce it for the other values of $j$. We discuss the case $j=3$ in detail, the others are done similarly.

For any object $x=x(t;\sigma_{0t},s_{0t},\theta_0,\theta_t,\theta_1,\theta_\infty)$, we denote, in accordance with the definition of $\tau_3$ in \eqref{def:tau-shifts},
\begin{equation*}
    \widetilde{x}=\widetilde{x}(t; \tfrac{1}{2}-\sigma_{0t}, -s_{0t}^{-1}, \theta_0+\tfrac{1}{2}, \theta_t, \theta_{1}, \theta_{\infty}-\tfrac{1}{2}).
\end{equation*}
Then we have, by definition, $\widetilde{\tau}_1=\tau_3$, so that the following are equivalent for $t_*\in \mathcal{T}$,
\begin{enumerate}
    \item $\tau_3(t)$ vanishes at $t=t_*$,
    \item $(\widetilde{ {\sf f}}(t),\widetilde{ {\sf g}}(t))\in\widetilde{E}_1$ at $t=t_*$.
\end{enumerate}
To prove the statement in the theorem, it remains to be shown that $(\widetilde{ {\sf f}}(t),\widetilde{ {\sf g}}(t))\in\widetilde{E}_1$ at $t=t_*$ if and only if $( {\sf f}(t), {\sf g}(t))\in E_3$ at $t=t_*$. This last statement can be obtained directly from the explicit formulas for the corresponding B\"acklund transformation \eqref{eq:BTtau1}. A direct computation yields that this transformation maps to two affine lines to one and another, and indeed,
using local coordinates $(u_1,v_1)$ around $E_3$ and local coordinates $(\widetilde{u}_1,\widetilde{v}_1)$ around $\widetilde{E}_1$, we have $u_3=0$ if and only if $\widetilde{u}_1=0$, in which case
\begin{equation*}
    v_3=\frac{\kappa _0^2 \left(q^2 \kappa _{\infty }^2-1\right)}{\left(\kappa _0^2-1\right) q^{5/2} \kappa _{\infty }^2}\widetilde{v}_1+\frac{\kappa _0 \left(\kappa _0 t-q \kappa _{\infty }\right) \left(-\kappa _1+\kappa _0 \left(\kappa _1^2+1\right) q \kappa _{\infty } \kappa _t-\kappa _1 \kappa _t^2\right)}{\left(\kappa _0^2-1\right) \kappa _1 q^2 \kappa _{\infty } \kappa _t}.
\end{equation*}
This last equation gives the explicit mapping between $E_1$ and $\widetilde{E}_8$ under the B\"acklund transformation \eqref{eq:BTtau1}. This proves the statement of the theorem for $j=3$. The remaining cases, $j=2,4$, are proven similarly.
\end{proof}


\subsection{Proof of Theorem~\ref{thm:g-tau}}\label{subsec:proofthmg}



We begin with the transformation of the jump matrix \eqref{eq:Jumpfacder} under the shift $\theta_{\infty} \to \theta_{\infty}-1$, which can be obtained using the explicit factorization of the jump matrix in terms of $q$-hypergeometric functions, see Subsection \ref{sec:origins},
\begin{align}\label{eq:Jshift}
  J(z, t)\vert_{\theta_{\infty} \to \theta_{\infty}-1}=:  {J}^{(2)}(z,t) = J(z, t) B^{(2)}(z,t)^{-1}, 
\end{align}
where
\begin{equation*}
    B^{(2)} = B_0^{(2)} + z\, B_{1}^{(2)}=:  \begin{bmatrix}

    b_1^{(2)} + z\, b_2^{(2)} & b_3^{(2)} \\
    b_4^{(2)} & 0
    
\end{bmatrix},
\end{equation*}
with 
\begin{align}
    B_0^{(2)} &= \begin{bmatrix}
 \frac{\left(q \kappa _{\infty }^2-1\right) \left(\kappa _1^2 \kappa_{0t}  (q+1) \kappa _{\infty }+\kappa_{0t}  (q+1) \kappa _{\infty }\right)-\kappa _1 \left(\kappa_{0t} ^2+1\right) \left(q \kappa _{\infty }^2+1\right)}{q \left(\kappa _{\infty }^2-1\right) \left(\kappa _1-\kappa_{0t}  \kappa _{\infty }\right) \left(\kappa _1 \kappa_{0t}  q \kappa _{\infty }-1\right)} & \frac{\left(\kappa _{\infty }-\kappa _1 \kappa_{0t} \right) \left(q^2 \kappa _{\infty }^2-1\right)}{q \kappa _{\infty } \left(\kappa _{\infty }^2-1\right) \left(q\kappa _1 \kappa_{0t}  \kappa _{\infty }-1\right)} \\
 \frac{\kappa _{\infty } \left(\kappa_{0t} -q\kappa _1 \kappa _{\infty }\right)}{\kappa_{0t}  \kappa _{\infty }-\kappa _1} & 0 \\
\end{bmatrix},\\
    B_1^{(2)} &= \begin{bmatrix}
       -\frac{\kappa _1 \kappa_{0t}  \left(q^2 \kappa _{\infty }^2-1\right) \left(q \kappa _{\infty }^2-1\right)}{q^2 \kappa _{\infty } \left(\kappa_{0t}  \kappa _{\infty }-\kappa _1\right) \left(q\kappa _1 \kappa_{0t}  \kappa _{\infty }-1\right)} & 0 \\
       0  & 0
    \end{bmatrix}.
\end{align}
We call $B^{(2)}=B^{(2)}(z,t)$ a shift matrix.
From here on we sometimes drop the $z,t$ dependence of the jump and shift matrices for ease of notation.

We now compute the expression for $\tau_2(t)$ defined in \eqref{def:tau-shifts},
\begin{align}
    \tau_2(t) &= \det_{\mathcal{H}_{+}} \left[ \Pi_{+} J^{(2)} \Pi_{+} (J^{(2)})^{-1} \right]\\
    &= \det_{\mathcal{H}_{+}} \left[ \Pi_{+}{J} (B^{(2)})^{-1} \Pi_{+} B^{(2)} {J}^{-1} \right].\label{eq:tautilde1}
\end{align}
Let us now analyse the term $   (B^{(2)})^{-1} \Pi_{+} B^{(2)}$. For a test function $f(z)$,
\begin{align}
    (B^{(2)})^{-1} \Pi_{+} B^{(2)} f(z) &= -\frac{1}{b_3^{(2)} b_4^{(2)}}\begin{bmatrix}
   0 & -b_3^{(2)} \\
   - b_4^{(2)} &  b_1^{(2)} + z b_2^{(2)}
\end{bmatrix} \Pi_{+} \begin{bmatrix}
   b_1^{(2)} + z b_2^{(2)} & b_3^{(2)} \\
    b_4^{(2)} &  0
\end{bmatrix} \begin{bmatrix}
     f_1 (z)  \\
     f_2(z)
\end{bmatrix} \\
&= \begin{bmatrix}
 \Pi_{+}f_1(z)\\
\Pi_{+} f_2(z)+ (b_2^{(2)}/b_3^{(2)}) \left(\Pi_{+} z f_1(z) - z  \Pi_{+}f_1(z)\right)
\end{bmatrix}\\
&= \Pi_{+} f(z) + (b_2^{(2)}/b_3^{(2)}) \begin{bmatrix}
     0 \\
      \Pi_{+} z f_1(z) - z  \Pi_{+}f_1(z)
\end{bmatrix}. 
\end{align}
We can simplify the above expression further by noting that
\begin{align}\label{id:ev}
     \Pi_{+} z f_1(z) - z  \Pi_{+}f_1(z) = \int_{C} \frac{w f_1(w)}{w-z} \frac{dw}{2\pi i} - \int_{C} \frac{z f_1(w)}{w-z} \frac{dw}{2\pi i} = \res_{z=0} f_1(z) =:{\sf{ev}_{-1}} f_1(z),
\end{align}
where ${\sf ev}_{-1}$ reads the coefficient of $z^{-1}$ in the Laurent expansion of $f(z)$\footnote{Our notation is consistent with the notation of \cite{gavrylenko2025riemannhilbertproblemsfredholmdeterminants} for identities such as \eqref{id:ev}.}. With the identity \eqref{id:ev}, we can then write 
\begin{align}
      (B^{(2)})^{-1} \Pi_{+} B^{(2)} = \Pi_{+} + \beta_7 v_7 \otimes \bar{v}_7, \label{eq:BPi1}
\end{align}
where 
\begin{align}
v_7 := \begin{bmatrix} 0\\ 1   \end{bmatrix}, &&  \bar{v}_7 f(z) := {\sf{ev}_{-1}} f_1(z), && \beta_7 := (b_2^{(2)}/b_3^{(2)}) = \frac{\kappa_{0t} \kappa_{1} \left(\kappa_{\infty}^2-1\right) \left(\kappa_{\infty}^2 q-1\right)}{q (\kappa_{0t} \kappa_{1}-\kappa_{\infty}) (\kappa_{0t} \kappa_{\infty}-\kappa_{1})}. \label{def:beta}
\end{align}

With \eqref{eq:BPi1}, the expression \eqref{eq:tautilde1} reads
\begin{align}
    \tau_2(t) = \det_{\mathcal{H}_{+}}\left[ \Pi_{+} J \left(\Pi_{+} + \beta_7 v_7 \otimes \bar{v}_7 \right)  J^{-1}\right].
\end{align}
The ratio of the tau functions $\tau_2$ and $\tau_1$ then simplifies as
\begin{align}
  \frac{\tau_2}{\tau_1}&= \det_{\mathcal{H}_{+}} \left[\left( \Pi_{+} J \Pi_{+} J^{-1}\right)^{-1}\Pi_{+} J \left(\Pi_{+} + \beta_7 v_7 \otimes \bar{v}_7 \right)  J^{-1}  \right]\\
    &=\det_{\mathcal{H}_{+}} \left[ \mathbb{1} + \beta_7 \left( \Pi_{+} J \Pi_{+} J^{-1}\right)^{-1} \Pi_{+} J\,\, v_7 \otimes \bar{v}_7  J^{-1}  \right]\\
    &\mathop{=}^{\eqref{det_id}}1 + \beta_7 \bar{v}_7  J^{-1}  \left( \Pi_{+} J^{-1}\right)^{-1} v_7 \mathop{=}^{\eqref{eq:jump_fac},\eqref{eq:TopInv}} 1 + \beta_7  \bar{v}_7 \Phi_{-}^{-1} \Pi_{+} \Phi_{-}v_7 \\
    &\mathop{=}^{\eqref{def:beta},\text{Prop}:\ref{Prop:asymp_Phi},} \sfw(t) \frac{\kappa_{0t} \kappa_{1} \left(\kappa_{\infty}^2-1\right) }{(\kappa_{0t} \kappa_{1}-\kappa_{\infty}) (\kappa_{0t} \kappa_{\infty}-\kappa_{1})}.\label{rat:t7t8}
\end{align}
In the third line of the above computation, we use the identity 
\begin{align}\label{det_id}
    \det(\mathbb{1}+ AB) =   \det(\mathbb{1} + B A),
\end{align}
with $A = \beta_7 \left(\Pi_{+} J \Pi_{+} J^{-1}\right)^{-1} \Pi_{+} J\,\, v_7$ and $B = \bar{v}_7  J^{-1}$.

As a final step, substituting \eqref{rat:t7t8} in the known relation
\begin{align}
    \frac{\sfw(t)}{\sfw(q/t)} = \frac{q \kappa_{\infty}^{-1} \sfg(t) -1}{\kappa_{\infty} \sfg(t)-1},
\end{align}
we obtain the following expression for $\sfg(t)$ in terms of the tau function
\begin{align}
    \sfg(t) = \frac{\tau_2(t) \tau_{8}(t/q) - \tau_1(t) \tau_2(t/q)}{\kappa_{\infty} \tau_2(t) \tau_1(t/q) - q \kappa_{\infty}^{-1} \tau_1(t) \tau_2(t/q)},
\end{align}
which is the same as \eqref{eqthm:g-tau} with the parameter $\kappa_{\infty} = q^{\theta_{\infty}}$.

\subsection{Proof of Theorem~\ref{thm:f-tau}}\label{subsec:prooftheoremf}

As a first step, we note that the ratio of tau functions in the RHS of the expression \eqref{eqthm:f-tau} can be written as
\begin{align}
    \frac{\tau_3 \tau_4}{\tau_2 \tau_1} = \left(\frac{\tau_3}{\tau_1} \right)\left(\frac{\tau_1}{\tau_2} \right) \left(\frac{\tau_4}{\tau_1} \right).
\end{align}
In the above expression, the term $\tau_1/\tau_2$ is simply the inverse of \eqref{rat:t7t8}. Let us now compute the ratios $\frac{\tau_3}{\tau_1}$, $\frac{\tau_4}{\tau_1}$ individually.

We first compute $\tau_3$. Recall from \eqref{def:tau-shifts} that 
\begin{align}
    \tau_3 =  \tau(t; \frac{1}{2}-\sigma_{0t}, s_{0t}, \theta_0+\frac{1}{2}, \theta_t, \theta_{1}, \theta_{\infty}-\frac{1}{2}).
\end{align}
Our first step is therefore to analyse the action of the shift 
\begin{align}\label{def:Shift-1}
   \left(\sigma_{0t}, \theta_{0}, \theta_{\infty}\right) \to \left(\frac{1}{2}-\sigma_{0t}, \theta_0+\frac{1}{2}, \theta_{\infty}-\frac{1}{2} \right)
\end{align}
on the jump matrix. Once again, using the factorisation of the jump matrix \eqref{eq:Jumpfacder} in terms of the $q$-hypergeometric functions, we obtain the expression \eqref{eq:jumptransform} which we now re-write as
\begin{align}
    J(z,t)\vert_{\eqref{def:Shift-1}}=: J^{(3)}(z,t) = (-t)^{-1/2} c^{(3)}  B^{(3)}_{\text{in}}(z,t) J(z,t) B^{(3)}_{\text{ex}}(z)^{-1}, \label{jump-shift-J1}
\end{align}
where the constant pre-factor
\begin{align}
    c^{(3)} = \frac{q \kappa_1^{-1} \left(q \kappa_1-\kappa_{\infty} \kappa_{0t}\right) \left(\kappa_1 \kappa_{\infty} \kappa_{0t}-1\right)}{\left(\kappa_{\infty}^2-q\right) \left(q-\kappa_{0t}^2\right) s_{0t} c_{0t}},
\end{align}
and the shift matrices equal
\begin{equation}\label{Bin-out-1}
\begin{split}
    B^{(3)}_{\text{in}}(z,t) &= t B^{(3)}_{\text{in}, 0} + z B^{(3)}_{\text{in},1}=: \begin{bmatrix}
       b^{(3)}_1  & b_2^{(3)} \\
        b_3^{(3)}+z & b_4^{(3)} 
    \end{bmatrix},\\
    B^{(3)}_{\text{ex}}(z) &= B^{(3)}_{\text{ex},0} + z B^{(3)}_{\text{ex},1}= \begin{bmatrix}
       b_5^{(3)}+ z  & b_6^{(3)} \\
        b_7^{(3)} & b_8^{(3)} 
    \end{bmatrix},
\end{split}
\end{equation}
with the matrices
\begin{align}
   B^{(3)}_{\text{\text{in}},1} = \begin{bmatrix}
       0  & 0 \\
        1 & 0
    \end{bmatrix}, &&  B^{(3)}_{\text{ex},1} = \begin{bmatrix}
       1  & 0 \\
        0 & 0
    \end{bmatrix},
\end{align}
\begin{equation}\label{val:B1inout0}
\begin{split}
    B^{(3)}_{\text{\text{in}},0} &= \begin{bmatrix}
 \frac{q (\kappa_{t}-\kappa_{0} \kappa_{0t}) (\kappa_{0} \kappa_{t} q-\kappa_{0t})}{\kappa_{0} \kappa_{t} \left(\kappa_{0t}^4-\kappa_{0t}^2 (q+1)+q\right)} & \frac{\kappa_{0t} q (\kappa_{0t} \kappa_{t}-\kappa_{0}) (\kappa_{0t}-\kappa_{0} \kappa_{t} q)}{\kappa_{0} \kappa_{t} \left(\kappa_{0t}^4-\kappa_{0t}^2 (q+1)+q\right)} \\
 -\frac{\kappa_{0t} q (\kappa_{0} \kappa_{0t}-\kappa_{t}) (\kappa_{0} \kappa_{0t} \kappa_{t}-1)}{\kappa_{0} \kappa_{t} \left(\kappa_{0t}^4-\kappa_{0t}^2 (q+1)+q\right)} & -\frac{\kappa_{0t}^2 q (\kappa_{0t} \kappa_{t}-\kappa_{0}) (\kappa_{0} \kappa_{0t} \kappa_{t}-1)}{\kappa_{0} \kappa_{t} \left(\kappa_{0t}^4-\kappa_{0t}^2 (q+1)+q\right)} \\
\end{bmatrix}=\begin{bmatrix}
       b^{(3)}_1  & b_2^{(3)} \\
        b_3^{(3)} & b_4^{(3)} 
    \end{bmatrix},\\
B^{(3)}_{\text{ex},0}&= \begin{bmatrix}
 \frac{\kappa_{\infty} q (\kappa_{1}-\kappa_{0t} \kappa_{\infty}) (\kappa_{0t} \kappa_{1} \kappa_{\infty}-1)}{\kappa_{0t} \kappa_{1} \left(\kappa_{\infty}^2-1\right) \left(\kappa_{\infty}^2 q-1\right)} & \frac{q (\kappa_{0t} \kappa_{1}-\kappa_{\infty}) (\kappa_{0t} \kappa_{\infty}-\kappa_{1})}{\kappa_{0t} \kappa_{1} \left(\kappa_{\infty}^2-1\right) \left(\kappa_{\infty}^2 q-1\right)} \\
 \frac{\kappa_{\infty}^2 q (\kappa_{0t} \kappa_{1} \kappa_{\infty}-1) (\kappa_{1} \kappa_{\infty} q-\kappa_{0t})}{\kappa_{0t} \kappa_{1} \left(\kappa_{\infty}^2-1\right) \left(\kappa_{\infty}^2 q-1\right)} & \frac{\kappa_{\infty} q (\kappa_{0t} \kappa_{1}-\kappa_{\infty}) (\kappa_{0t}-\kappa_{1} \kappa_{\infty} q)}{\kappa_{0t} \kappa_{1} \left(\kappa_{\infty}^2-1\right) \left(\kappa_{\infty}^2 q-1\right)}\\
\end{bmatrix}=\begin{bmatrix}
       b_5^{(3)}  & b_6^{(3)} \\
        b_7^{(3)} & b_8^{(3)} 
    \end{bmatrix} .
\end{split}
\end{equation}
Let us note that $B^{(3)}_{\text{\text{in}},0}$, $B^{(3)}_{\text{ex},0}$ are singular matrices, and therefore
\begin{align}
    \det B^{(3)}_{\text{\text{in}}}(z/t) = -z b_2^{(3)}, 
    && \det B^{(3)}_{\text{ex}}(z) = z b_8^{(3)}.
\end{align}
Due to the above structure of these shift matrices, they can be factorized as follows\footnote{dropping the $z,t$ dependence in favour of readability.}
\begin{align}
    B^{(3)}_{\text{in}} 
     &=\begin{bmatrix}
       0  & b_1^{(3)} \\
        1 & b_3^{(3)}
    \end{bmatrix} 
    \begin{bmatrix}
       z  & 0 \\
        0 & 1
    \end{bmatrix} \begin{bmatrix}
       1 & 0 \\
        1 & b_2^{(3)}/b_1^{(3)}
    \end{bmatrix}\label{B1in-fac}
\end{align}
and
\begin{align}
      B^{(3)}_{\text{ex}} 
    &=\begin{bmatrix}
       1  & b_6^{(3)} \\
        0 & b_8^{(3)}
    \end{bmatrix} \begin{bmatrix}
       z  & 0 \\
        0 & 1
    \end{bmatrix} \begin{bmatrix}
       1  & 0 \\
        b_7^{(3)}/b_8^{(3)} & 1
    \end{bmatrix}.\label{B1out-fac}
\end{align}
Now, the tau function with the shifts \eqref{def:Shift-1} is 
\begin{align}
    \tau_3(t) =\det_{\mathcal{H}_{+}} \left[\Pi_{+} J^{(3)} \Pi_{+} (J^{(3)})^{-1} \right] &= \det_{\mathcal{H}_{+}} \left[\Pi_{+} J^{(3)} \Pi_{+} (J^{(3)})^{-1} \Pi_{+} \right]\\
    &= \det_{\mathcal{H}_{+}} \left[\Pi_{+} B^{(3)}_{\text{in}} J (B^{(3)}_{\text{ex}})^{-1}\Pi_{+} B^{(3)}_{\text{ex}} J^{-1} (B^{(3)}_{\text{in}})^{-1} \Pi_{+}\right]\\
       &= \det_{\mathcal{H}_{+}} \left[\Pi_{+} \Lambda_1^{-1} \widetilde{J}_1 \Lambda_1\Pi_{+} \Lambda_1^{-1} (\widetilde{J}_1)^{-1} \Lambda_1 \Pi_{+}\right], \label{tau1simpl:Step1}
\end{align}
where
\begin{align}
 \Lambda_1 := \begin{bmatrix}
       1  & 0 \\
        0 & z
    \end{bmatrix},   && \widetilde{J}_1 := {\begin{bmatrix}
       1 & 0 \\
        1 & b_2^{(3)}/b_1^{(3)}
    \end{bmatrix}} J \begin{bmatrix}
       1  & 0 \\
        b_7^{(3)}/b_8^{(3)} & 1
    \end{bmatrix}^{-1}.\label{id:Jtil}
\end{align}
Repeating similar computation to \eqref{id:ev}, we first note that  

\begin{align}
        \Lambda_1 \Pi_{+} \Lambda_1^{-1} f(z) = \begin{bmatrix}
       1  & 0 \\
        0 & z
    \end{bmatrix} \Pi_{+} \begin{bmatrix}
       1  & 0 \\
        0 & z
    \end{bmatrix}^{-1} f(z) 
    &=   \begin{bmatrix}
     \Pi_{+} f_1(z) \\
     z \Pi_{+} f_2(z)/z
    \end{bmatrix}  . \label{eq:Lambda1conj}
    \end{align}
    Moreover,
    \begin{align}
         z \Pi_{+} f_2(z)/z = z \int_C \frac{f_2(w)}{w (w-z)} \frac{dw}{2\pi i} = \int_C f_2(w)\left(\frac{1}{w-z}-\frac{1}{w}\right)\frac{dw}{2\pi i} 
        &=: \Pi_{+} f_2(z) - {\sf ev}_0 f_2(z),
    \end{align}
where ${\sf ev}_0$ reads out the constant term in the Laurent expansion of the test function $f_2(z)$.
Therefore, with the above identity, we can write the expression \eqref{eq:Lambda1conj} as
    \begin{align}
        \Lambda_1 \Pi_{+} \Lambda_1^{-1} = \Pi_{+} - v_1 \otimes \bar{v}_1; && v_1 := \begin{bmatrix}
             0  \\
             1 
        \end{bmatrix}, && \bar{v}_1 f(z) := {\sf ev}_0 f_2 .
    \end{align}
    Let us now recall the identity \cite[eq (4.15)]{gavrylenko2025riemannhilbertproblemsfredholmdeterminants}:
    \begin{align}
        \det_{\mathcal{H}_{+}} \left( X \right) =  \det_{\mathcal{H}_{+}} \left(v_1 \otimes \bar{v}_1 + \Lambda_1 \Pi_{+} X \Pi_{+} \Lambda_1^{-1}\right),
    \end{align}
    which when applied to \eqref{tau1simpl:Step1} gives
    \begin{align}
        \tau_3(t) & =  \det_{\mathcal{H}_{+}} \left[v_1 \otimes \bar{v}_1 + \Lambda_1\Pi_{+} \Lambda_1^{-1} \widetilde{J}_1 \Lambda_1\Pi_{+} \Lambda_1^{-1} (\widetilde{J}_1)^{-1} \Lambda_1 \Pi_{+} \Lambda_1^{-1}\right]\\
        &= \det_{\mathcal{H}_{+}} \left[\Pi_{+} \widetilde{J}_1 \Pi_{+} (\widetilde{J}_1)^{-1} \right] \left(\bar{v}_1 \left(\Pi_{+} (\widetilde{J}_1)^{-1} \Pi_{+} \right)^{-1} v_1 . \bar{v}_1 \left(\Pi_{+} (\widetilde{J}_1) \Pi_{+} \right)^{-1} v_1 \right). \label{t1der-1}
    \end{align}
   The final expression in the above display is derived by repeating the argument used in \cite[(4.17)-(4.19)]{gavrylenko2025riemannhilbertproblemsfredholmdeterminants}.
Furthermore, with the definition \eqref{id:Jtil}, we observe that $$\det_{\mathcal{H}_{+}} \left[\Pi_{+} \widetilde{J}_1 \Pi_{+} \widetilde{J}_1^{-1} \right] = \det_{\mathcal{H}_{+}} \left[\Pi_{+} J \Pi_{+} {J}^{-1} \right]= \tau_1.$$ Therefore, with the identity above, we can simplify \eqref{t1der-1} as
\begin{align}
    \frac{\tau_3}{\tau_1} = \bar{v}_1 \left(\Pi_{+} \widetilde{J}_1^{-1} \Pi_{+} \right)^{-1} v_1 . \bar{v}_1 \left(\Pi_{+} \widetilde{J}_1 \Pi_{+} \right)^{-1} v_1.\label{tau1comp:frac}
\end{align}

Let us now compute the terms in the RHS of the above expression. We begin by computing the element
\begin{align}
    \bar{v}_1 \left(\Pi_{+} \widetilde{J}_1^{-1} \Pi_{+} \right)^{-1} v_1 &= \bar{v}_1 \left(\Pi_{+}\begin{bmatrix}
       1  & 0 \\
        b_7^{(3)}/b_8^{(3)} & 1
    \end{bmatrix} J^{-1} { \begin{bmatrix}
       1 & 0 \\
        1 & b_2^{(3)}/b_1^{(3)}
    \end{bmatrix}}^{-1}\Pi_{+} \right)^{-1} v_1 \\
    & = \bar{v}_1 { \begin{bmatrix}
       1 & 0 \\
        1 & b_2^{(3)}/b_1^{(3)}
    \end{bmatrix}} \left(\Pi_{+}J^{-1}\Pi_{+} \right)^{-1} \begin{bmatrix}
       1  & 0 \\
        b_7^{(3)}/b_8^{(3)} & 1
    \end{bmatrix}^{-1}  v_1 \\
    & = \bar{v}_1 { \begin{bmatrix}
       1 & 0 \\
        1 & b_2^{(3)}/b_1^{(3)}
    \end{bmatrix}} \Phi_{+}^{-1} \Pi_{+} \Phi_{-}  \left( \begin{array}{c}
             0  \\
             1 
        \end{array}\right) \\
    &\mathop{=}^{\eqref{asymp:Psiplus},\eqref{asymp:Phimin}} \frac{{ - (\Phi_{+}^{(0)})_{12} + (b_2^{(3)}/b_1^{(2)}) (\Phi_{+}^{(0)})_{11}}}{\det (\Phi_{+}^{(0)})},\label{tau1comp:term1}
\end{align}
where, to obtain the final line, we used the fact that for $z\to \infty$
\begin{align}
    \Phi_{-}(z,t)&=  \mathbb{1}  + \mathcal{O}(z^{-1}). \label{asymp:Phimin} \\
\end{align}
Similarly,
\begin{align}
    \bar{v}_1 \left(\Pi_{+} \widetilde{J}_1 \Pi_{+} \right)^{-1} v_1 &=  \bar{v}_1\begin{bmatrix}
       1  & 0 \\
        b_7^{(3)}/b_8^{(3)} & 1
    \end{bmatrix}  \left(\Pi_{+} \widetilde{J}  \Pi_{+} \right)^{-1} { \begin{bmatrix}
       1 & 0 \\
        1 & b_2^{(3)}/b_1^{(3)}
    \end{bmatrix}}^{-1} v_1 \\
    &=  \bar{v}_1\begin{bmatrix}
       1  & 0 \\
        b_7^{(3)}/b_8^{(3)} & 1
    \end{bmatrix}  \widehat{\Phi}_{+} \Pi_{+} \widehat{\Phi}_{-}^{-1} { \begin{bmatrix}
       1 & 0 \\
        1 & b_2^{(3)}/b_1^{(3)}
    \end{bmatrix}}^{-1} v_1 \\
    & \mathop{=}^{\eqref{asymp:psihatplus},\eqref{asymp:Psimin}}(-t)^{\sigma_{0t}} \left(\frac{b_1^{(3)}}{b_2^{(3)}} \right) \left(\frac{b_7^{(3)}}{b_8^{(3)}} (\widehat{\Phi}_{+}^{(0)})_{12} + (\widehat{\Phi}_{+}^{(0)})_{22}\right).\label{tau1comp:term2}
\end{align}
Substituting \eqref{tau1comp:term1}, \eqref{tau1comp:term2} back in \eqref{tau1comp:frac}, we get 
\begin{align}
    \frac{\tau_3}{\tau_1} &= (-t)^{\sigma_{0t}} \left(\frac{{ - (b_1^{(3)}/b_2^{(3)})(\Phi_{+}^{(0)})_{12} +  (\Phi_{+}^{(0)})_{11}}}{\det (\Phi_{+}^{(0)})} \right) \left( \frac{b_7^{(3)}}{b_8^{(3)}} (\widehat{\Phi}_{+}^{(0)})_{12} + (\widehat{\Phi}_{+}^{(0)})_{22}\right)\label{tau18:exp1}\\
&\mathop{=}^{\eqref{asymp:Phiplus0},\eqref{asymp:psihatplus0},\eqref{val:B1inout0}} \frac{\left(\kappa_{0t}^2-1\right)^2  (-t)^{\sigma_{0t}} \sff(t)\sfw(t)}{\kappa_{0} \kappa_{0t}^2 (q-1)^2 q t \kappa_{1}(t) (\kappa_{0} \kappa_{0t} \kappa_{t}-1) (\kappa_{0t} \kappa_{1}-\kappa_{\infty})\epsilon_1(t)}.\label{tau18:final}
\end{align}

Let us now repeat a similar computation as above for the term $\tau_4/\tau_1$.
Recall that 
\begin{align}
    \tau_4 =  \tau(t; \frac{1}{2}-\sigma_{0t}, s_{0t}, \theta_0-\frac{1}{2}, \theta_t, \theta_{1}, \theta_{\infty}-\frac{1}{2}).
\end{align}
To compute $\tau_4$, we need to analyse the action of the shift 
\begin{align}\label{def:Shift-2}
  \left(\sigma_{0t}, \theta_{0}, \theta_{\infty}\right) \to \left(\frac{1}{2}-\sigma_{0t}, \theta_0-\frac{1}{2}, \theta_{\infty}-\frac{1}{2} \right)
\end{align}
on the jump. From an analogous computation to \eqref{eq:jumptransform} for the shifts of the parameters above, we obtain the expression
\begin{align}
   J(z,t)\vert_{\eqref{def:Shift-2}}=: J^{(4)}(z,t) = (-t)^{-1/2} C^{(4)}  B_{\text{in}}^{(4)}(z,t) J(z,t) (B_{\text{ex}}(z)^{(4)})^{-1},
\end{align}
where the constant multiplier
\begin{align}
    C^{(4)} 
    =\frac{q \kappa_1^{-1} \left(q\kappa_1-\kappa_{\infty}\kappa_{0t}\right) \left(\kappa_1\kappa_{\infty} \kappa_{0t}-1\right)}{\left(\kappa_{\infty}^2-q\right) \left(\kappa_{0t}^2-q\right) c_{0t} s_{0t}},
\end{align}
and the shift matrices equal
\begin{equation}\label{Bin-out-2}
\begin{split}
    B^{(4)}_{\text{in}}(z,t) &= t B^{(4)}_{\text{in}, 0} + z B^{(4)}_{\text{in},1}=: \begin{bmatrix}
       b^{(4)}_1  & b_2^{(4)} \\
        b_3^{(4)}+z & b_4^{(4)} 
    \end{bmatrix},\\
    B^{(4)}_{\text{ex}}(z) &= B^{(4)}_{\text{ex},0} + z B^{(4)}_{\text{ex},1}= \begin{bmatrix}
       b_5^{(4)}+ z  & b_6^{(4)} \\
        b_7^{(4)} & b_8^{(4)} 
    \end{bmatrix},
\end{split}
\end{equation}
with the constant coefficient matrices
\begin{align}
   B^{(4)}_{\text{in},1} = \begin{bmatrix}
       0  & 0 \\
        1 & 0
    \end{bmatrix}, &&  B^{(4)}_{\text{ex},1} = \begin{bmatrix}
       1  & 0 \\
        0 & 0
    \end{bmatrix},
\end{align}
\begin{equation}\label{val:B2inout0}
\begin{split}
    B_{\text{in},0}^{(4)} &= \begin{bmatrix}
 \frac{q (\kappa_{0t}-\kappa_{0} \kappa_{t}) (\kappa_{0} \kappa_{0t}-\kappa_{t} q)}{\kappa_{0} \kappa_{t} \left(\kappa_{0t}^4-\kappa_{0t}^2 (q+1)+q\right)} & \frac{\kappa_{0t} q (\kappa_{0} \kappa_{0t} \kappa_{t}-1) (\kappa_{0} \kappa_{0t}-\kappa_{t} q)}{\kappa_{0} \kappa_{t} \left(\kappa_{0t}^4-\kappa_{0t}^2 (q+1)+q\right)} \\
 \frac{\kappa_{0t} q (\kappa_{0t} (\kappa_{t}+\kappa_{t}^{-1})-\kappa_{0}^{-1}\kappa_{0t}^2-\kappa_{0})}{\kappa_{0t}^4-\kappa_{0t}^2 (q+1)+q} & -\frac{\kappa_{0t}^2 q (\kappa_{0t} \kappa_{t}-\kappa_{0}) (\kappa_{0} \kappa_{0t} \kappa_{t}-1)}{\kappa_{0} \kappa_{t} \left(\kappa_{0t}^4-\kappa_{0t}^2 (q+1)+q\right)} \\
\end{bmatrix} = \begin{bmatrix}
       b^{(4)}_1  & b_2^{(4)} \\
        b_3^{(4)} & b_4^{(4)} 
    \end{bmatrix},\\
     B_{\text{ex},0}^{(4)}&= \begin{bmatrix}
 \frac{\kappa_{\infty} q (\kappa_{1}-\kappa_{0t} \kappa_{\infty}) (\kappa_{0t} \kappa_{1} \kappa_{\infty}-1)}{\kappa_{0t} \kappa_{1} \left(\kappa_{\infty}^2-1\right) \left(\kappa_{\infty}^2 q-1\right)} & \frac{q (\kappa_{0t} \kappa_{1}-\kappa_{\infty}) (\kappa_{0t} \kappa_{\infty}-\kappa_{1})}{\kappa_{0t} \kappa_{1} \left(\kappa_{\infty}^2-1\right) \left(\kappa_{\infty}^2 q-1\right)} \\
 \frac{\kappa_{\infty}^2 q (\kappa_{0t} \kappa_{1} \kappa_{\infty}-1) (\kappa_{1} \kappa_{\infty} q-\kappa_{0t})}{\kappa_{0t} \kappa_{1} \left(\kappa_{\infty}^2-1\right) \left(\kappa_{\infty}^2 q-1\right)} & \frac{\kappa_{\infty} q (\kappa_{0t} \kappa_{1}-\kappa_{\infty}) (\kappa_{0t}-\kappa_{1} \kappa_{\infty} q)}{\kappa_{0t} \kappa_{1} \left(\kappa_{\infty}^2-1\right) \left(\kappa_{\infty}^2 q-1\right)} \\
\end{bmatrix} = \begin{bmatrix}
       b_5^{(4)}  & b_6^{(4)} \\
        b_7^{(4)} & b_8^{(4)} 
    \end{bmatrix}.
\end{split}
\end{equation}

As before, we find that $B_{\text{in},0}^{(4)}$, $B_{\text{ex},0}^{(4)}$ are singular matrices, and so,
\begin{align}
    \det   B_{\text{in}}^{(4)} = -z b_2^{(4)}, 
    && \det   B_{\text{ex},0}^{(4)} = z b_8^{(4)}.
\end{align}
Note that the analytic structure of the shift matrices \eqref{Bin-out-1} and \eqref{Bin-out-2} is the same with the only difference being the values of the constant coefficient matrices \eqref{val:B1inout0} and \eqref{val:B2inout0}. Therefore, starting from the expression
\begin{align}
    \tau_4(t) &= \det_{\mathcal{H}_{+}} \left[\Pi_{+} J^{(4)} \Pi_{+} (J^{(4)})^{-1} \right],
\end{align}
and repeating the arguments for the ratio $\tau_3/\tau_1$ in \eqref{tau18:exp1}, we can conclude that
\begin{align}
    \frac{\tau_4(t)}{\tau_1(t)} &= (-t)^{\sigma_{0t}} \left(\frac{{ - (b_1^{(4)}/b_2^{(4)})(\Phi_{+}^{(0)})_{12} +  (\Phi_{+}^{(0)})_{11}}}{\det (\Phi_{+}^{(0)})} \right) \left( \frac{b_7^{(4)}}{b_8^{(4)}} (\widehat{\Phi}_{+}^{(0)})_{12} + (\widehat{\Phi}_{+}^{(0)})_{22}\right)\\
&\mathop{=}^{\eqref{asymp:Phiplus0},\eqref{asymp:psihatplus0},\eqref{val:B2inout0}} \frac{\left(\kappa_{0t}^2-1\right)^2 (-t)^{\sigma_{0t}}}{\kappa_{0t}^2 (q-1)^2 q t  (\kappa_{0t} \kappa_{t}-\kappa_{0}) (\kappa_{0t} \kappa_{1}-\kappa_{\infty})\epsilon_2(t)}.\label{tau28:final}
\end{align}
We repeat that the only information used to obtain \eqref{tau18:exp1} is the analytic structure of the shift matrices. 

Combining the expressions \eqref{rat:t7t8}, \eqref{tau18:final}, and \eqref{tau28:final} we find 
\begin{align}
    \left(\frac{\tau_3}{\tau_1} \right)\left(\frac{\tau_1}{\tau_2} \right) \left(\frac{\tau_4}{\tau_1} \right) &= \sff(t)\frac{\left(\kappa_{0t}^2-1\right)^4  (-t)^{2 (\sigma_{0t}-1)} (\kappa_{1}-\kappa_{0t} \kappa_{\infty})}{\kappa_{0} \kappa_{0t}^5 \kappa_{1} \left(\kappa_{\infty}^2-1\right) (q-1)^4 q^2  (\kappa_{0}-\kappa_{0t} \kappa_{t}) (\kappa_{0} \kappa_{0t} \kappa_{t}-1) (\kappa_{0t} \kappa_{1}-\kappa_{\infty}) \epsilon_1(t) \epsilon_2(t)}\\
    &\mathop{=}^{\eqref{id:epsilon12}} \sff(t)  (-t)^{(2\sigma_{0t}-1)}\frac{r_{0t}  \kappa_{0} \left(\kappa_{0t}^2-1\right)^2 \kappa_t   (\kappa_{1}-\kappa_{0t} \kappa_{\infty})}{\kappa_{0t}^3 (\kappa_{0}-\kappa_{0t} \kappa_t) (\kappa_{0} \kappa_{0t} \kappa_t-1) (\kappa_{0t} \kappa_{1}-\kappa_{\infty})},
\end{align}
proving the statement of Theorem~\ref{thm:f-tau}.

\subsection{Minor expansion and asymptotics}\label{subsec:proofminor}
In this subsection, we use the explicit form of the Fredholm determinant described in Proposition \ref{prop:Widom} to obtain asymptotics around $t=0$ of the tau function. 

\begin{proof}[Proof of Theorem \ref{thm:tau-asymp}]
   
Recall the structure of the tau function \eqref{prop:Widom}.
Let us begin by writing the series representation of the kernels $a(z,w)$, $b(z,w)$ in \eqref{exp:kernels} as \cite{GavrylenkoLisovyi2017,CGL2017}:
\begin{align}\label{minordef:a}
    a(z,w) = \sum_{p,q\in \mathbb{Z}'_{+}} a^{p}_{-q} z^{-\frac{1}{2}+p} w^{-\frac{1}{2}+q},
\end{align}
and
\begin{align}\label{minordef:b}
     (-t)^{\sigma_{0t} \sigma_3} \begin{bmatrix}
         r_{0t} & 0 \\
         0 & 1
     \end{bmatrix}b(z,w)\begin{bmatrix}
         r_{0t} & 0 \\
         0 & 1
     \end{bmatrix}^{-1}(-t)^{-\sigma_{0t} \sigma_3} = \sum_{p,q\in \mathbb{Z}'_{+}} b^{-p}_{q} z^{-\frac{1}{2}-p} w^{-\frac{1}{2}-q} t^{p+q}. 
\end{align}
The notation $\mathbb{Z}'_{+}$ denotes the set of positive half integers.

The expressions \eqref{minordef:a}, \eqref{minordef:b} are key to deriving the asymptotic expansion of the tau function as $t\to 0$ to any order. Let us rewrite the statement of \cite[Theorem 2.12]{GavrylenkoLisovyi2017} in the present context. For $Q\in \mathbb{Z}$, the terms up to order $t^{Q}$ of the Fredholm determinant \eqref{prop:Widom} are given by the expansion up to the $Q$th minor $\det \left(1+U_Q\right)$, where $U_{Q} = \left(\begin{array}{cc}
    0 & a_Q \\
    b_Q & 0
\end{array} \right)$. The block matrices $a_{Q}$, $b_{Q}$ are of order $2Q \times 2Q$, and are given by 
\begin{align}
   a_Q=
\begin{bmatrix}
a^{\,Q-\frac12}_{-\frac12} & 0 & \cdots & 0 \\
\vdots & a^{\,Q-\frac32}_{-\frac32} & \ddots & \vdots \\
a^{\,\frac32}_{-\frac12} & \ddots & \ddots & 0 \\
a^{\,\frac12}_{-\frac12} & a^{\,\frac12}_{-\frac32} & \cdots & a^{\,\frac12}_{\frac12-Q}
\end{bmatrix}, 
\end{align}
and
\begin{align}
b_Q
= (-t)^{-\sigma_{0t} \sigma_3} \begin{bmatrix}
         r_{0t} & 0 \\
         0 & 1
     \end{bmatrix}^{-1}
\begin{bmatrix}
b^{-\frac12}_{\,Q-\frac12}\, t^{Q}&\cdots& b^{-\frac12}_{\,\frac32}\, t^{2} & b^{-\frac12}_{\,\frac12}\, t \\
0 & \ddots & \cdot & b^{-\frac32}_{\,\frac12}\, t^{2} \\
\vdots & \cdot & b^{\,\frac32-Q}_{\,\frac32-Q}\, t^{Q} & \vdots \\
0 & \cdots & 0 & b^{\,\frac12-Q}_{\,\frac12-Q}\, t^{Q}
\end{bmatrix}
\begin{bmatrix}
         r_{0t} & 0 \\
         0 & 1
     \end{bmatrix} (-t)^{\sigma_{0t} \sigma_3}.
\end{align}
The conjugation by the matrix $(-t)^{-\sigma_{0t} \sigma_3} \begin{bmatrix}
         r_{0t} & 0 \\
         0 & 1
     \end{bmatrix}^{-1} $ in $d_Q$ is understood to act on each of the elements $b^{-p}_q$. Therefore, the asymptotic expansion of the tau function has the form \eqref{thm:tau-asymp} where the coefficients $\widehat{\tau}_{n,k}$ can  be computed term-wise.

Let us now compute the first few coefficients $\widehat{\tau}_{0,0}$, $\widehat{\tau}_{1,-1}, \widehat{\tau}_{1,0}$, and $\widehat{\tau}_{1,1}$. To this end, it suffices to compute the first few values of the coefficients $a^{p}_{-q}$ and $b^{-p}_q$, which can be derived by substituting the asymptotic expansions of $\widehat{\Phi}_{-}(z,t) = \widehat{\Phi}_{\infty}^{i}\left(\frac{z}{t}\right) (-t)^{\sigma_{0t} \sigma_3}$ and $\widehat{\Phi}_{+}(z) = \widehat{\Phi}_{0}^{e}(z)$ into the expressions \eqref{exp:kernels} and matching powers of the variables $z$ and $w$. We can write the expansion \eqref{asymp:psihatplus} as
\begin{align}
    \widehat{\Phi}_{0}^{e}(z) = \widehat{\Phi}_{+}^{(0)}\left(\mathbb{1} + \sum_{k=1}^{\infty} g^{e}_{0,k} z^k\right), 
\end{align}
and we have \eqref{asymp:Psimin}, which we can write as
\begin{align}
     \widehat{\Phi}_{\infty}^{i}\left(\frac{z}{t} \right) = \mathbb{1} + \sum_{k=1}^{\infty} g^{i}_{\infty,k} z^{-k} t^{k}.
\end{align}
Substituting above expressions in \eqref{exp:kernels} and comparing with the basis expansions \eqref{minordef:a}, \eqref{minordef:b}, we can obtain the desired coefficients. The first few expressions read
\begin{align}
    b_{\frac{1}{2}}^{-\frac{1}{2}} = g_{\infty,1}^{i}, &&b_{\frac{3}{2}}^{-\frac{1}{2}}= g_{\infty,2}^{i}, &&  b^{-\frac{3}{2}}_{\frac{1}{2}} = g_{\infty,2}^{i} - \left( g_{\infty,1}^{i}\right)^2,
\end{align}
\begin{align}
    a^{\frac{1}{2}}_{-\frac{1}{2}} = g_{0,1}^e, &&a^{\frac{1}{2}}_{-\frac{3}{2}} = g_{0,2}^e, && a^{\frac{3}{2}}_{-\frac{1}{2}}= g_{0,2}^e - \left( g_{0,1}^e \right)^2.
\end{align}
The coefficients $g^{e}_{0,k}$, $g^{i}_{\infty,k}$ can be calculated from the expressions \eqref{explicit:psi0e}, \eqref{exp:psihat0e}, and \eqref{eq:internalpsiinf}, \eqref{exp:Psihatinfi} respectively. The first couple of terms are 
\begin{align}
    g_{0,1}^e &= \begin{bmatrix}
 \frac{-\kappa_{\infty} \left(\kappa_{0t}^2 ((\kappa_{1}-1) \kappa_{1}+1)+\kappa_{1}\right)+\kappa_{0t} \kappa_{1} \kappa_{\infty}^2+\kappa_{0t} \kappa_{1}}{\left(\kappa_{0t}^2-1\right) \kappa_{1} \kappa_{\infty} (q-1)} & -\frac{(\kappa_{0t} \kappa_{\infty}-\kappa_{1}) (\kappa_{0t}-\kappa_{1} \kappa_{\infty})}{\left(\kappa_{0t}^2-1\right) \kappa_{1} \kappa_{\infty} \left(\kappa_{0t}^2 q-1\right)} \\
 \frac{\kappa_{0t}^2 \kappa_{\infty} \left(\frac{1}{\kappa_{\infty}}-\kappa_{0t} \kappa_{1}\right) \left(\frac{\kappa_{0t} \kappa_{1}}{\kappa_{\infty}}-1\right)}{\kappa_{1} \left(\kappa_{0t}^4-\kappa_{0t}^2 (q+1)+q\right)} & -\frac{\kappa_{0t} \left(-\kappa_{0t}+\kappa_{\infty}+\frac{1}{\kappa_{\infty}}\right)-\kappa_{1}-{\kappa_{1}^{-1}}+1}{\left(\kappa_{0t}^2-1\right) (q-1)} \\
\end{bmatrix} \label{eq:g01e}\\
    g_{\infty,1}^{i} & =\begin{bmatrix}
 \frac{q \left(\kappa_{0}^2 (-\kappa_{0t}) \kappa_{t}+\kappa_{0} \left(\kappa_{0t}^2 ((\kappa_{t}-1) \kappa_{t}+1)+\kappa_{t}\right)-\kappa_{0t} \kappa_{t}\right)}{\kappa_{0} \left(\kappa_{0t}^2-1\right) \kappa_{t} (q-1)} & \frac{\kappa_{0t}^2 q (\kappa_{0}-\kappa_{0t} \kappa_{t}) (\kappa_{0} \kappa_{0t} \kappa_{t}-1)}{\kappa_{0} \left(\kappa_{0t}^2-1\right) \kappa_{t} \left(q-\kappa_{0t}^2\right)} \\
 \frac{q (\kappa_{0} \kappa_{0t}-\kappa_{t}) (\kappa_{0t}-\kappa_{0} \kappa_{t})}{\kappa_{0} \left(\kappa_{0t}^2-1\right) \kappa_{t} \left(\kappa_{0t}^2 q-1\right)} & \frac{q \left(\kappa_{0}^2 \kappa_{0t} \kappa_{t}-\kappa_{0} \left(\kappa_{t} \left(\kappa_{0t}^2+\kappa_{t}-1\right)+1\right)+\kappa_{0t} \kappa_{t}\right)}{\kappa_{0} \left(\kappa_{0t}^2-1\right) \kappa_{t} (q-1)} \\
\end{bmatrix}\label{eq:ginf1i}
\end{align}
Then, the leading asymptotic for $t\to 0$ reads
\begin{align}
    \tau(t) &= \det\left[\mathbb{1} + \begin{bmatrix} 0 & a^{\frac{1}{2}}_{-\frac{1}{2}}\\
    \text{diag}\left((-t)^{-\sigma_{0t}} r_{0t}^{-1}, (-t)^{\sigma_{0t}}  \right) \left(t \,b_{\frac{1}{2}}^{-\frac{1}{2}}\right)\, \text{diag}\left((-t)^{\sigma_{0t}} r_{0t}, (-t)^{-\sigma_{0t}}  \right) & 0 \end{bmatrix} + \mathcal{O}(t^2)\right] \\
    &= \det\left(\mathbb{1}- t\, g_{0,1}^e\text{diag}\left((-t)^{-\sigma_{0t}} r_{0t}^{-1}, (-t)^{\sigma_{0t}}  \right)   g_{\infty,1}^{i} \text{diag}\left((-t)^{\sigma_{0t}} r_{0t}, (-t)^{-\sigma_{0t}}  \right)+ \mathcal{O}(t^{2(1-2\Re \sigma_{0t})}) \right).
\end{align}
Substituting the values \eqref{eq:g01e} and \eqref{eq:ginf1i} in the expression above gives us the values of the coefficients \eqref{eq:tau-asymp-coeff}.
\end{proof}
\begin{remark}
With the asymptotic behaviour of the tau function in  \eqref{thm:tau-asymp} and the relations between the tau functions and the $q$PVI transcendents \eqref{eqthm:g-tau} and \eqref{eqthm:f-tau}, one can verify the leading asymptotic behaviour for the qPVI transcendents ${\sf f}(t), {\sf g}(t)$ in Theorem \ref{thm:asymptotic}. One can further verify through deriving the $q\to 1$ limit of the matrices  \eqref{eq:g01e} and \eqref{eq:ginf1i} that the above expression reproduces the known asymptotic expansion of Painlev\'e VI tau function \cite[(2.47)]{gavrylenko2018fredholm} in the limit $q\to 1$, up to an overall factor of $ t^{\sigma^2 - \theta_0^2 - \theta_t^2 }$.     

\end{remark}

\section{Conclusion}\label{sec:conclusion}
The main result of this paper is the rigorous, analytic
construction of a tau function for $q\Psix$ as a Fredholm determinant, which is analytic on its domain of definition, whose divisor characterises the points at which the corresponding RHP is not solvable and equivalently where the transcendent takes value in a particular exceptional line. We further derived formulas for the $q\Psix$ dependent variables in terms of the tau function and three copies of it with some of the parameters shifted, as well as the asymptotics of the tau function around $t=0$. Let us now list some future directions leading from this work.

\begin{itemize}[leftmargin=*]
\item An immediate question is how the tau function defined in \cite{JNS2017}, which we denote by $\tau_\text{JNS}$, and $\tau_{W}$, in Definition \ref{def:widomconstant}, are related. Upon the following identification of parameters and variables\footnote{We introduced minus signs in the identification of $\theta_t$ and $\theta_1$ so that our indexing of exceptional lines (conjecturally) matches with the indexing of tau functions in \cite{JNS2017}.},
\begin{align*}
    \theta_0&=\theta_0^{\operatorname{JNS}}, & 
    \theta_t&=-\theta_t^{\operatorname{JNS}}, & 
    \theta_1&=-\theta_1^{\operatorname{JNS}}, & 
    \theta_\infty&=\theta_\infty^{\operatorname{JNS}}, \\
   \sigma_{0t}&=\sigma^{\operatorname{JNS}}, &
   f&=q^{1-\theta_1}y^{\operatorname{JNS}}, &
   g&=z^{\operatorname{JNS}}, &
   t&=q^{\theta_t+\theta_1}t^{\operatorname{JNS}},
\end{align*}
We expect the relation to take the form
\begin{align}
    \tau_{\text{JNS}} = \frac{\prod_{\epsilon, \epsilon' = \pm}G_q(1+\epsilon \sigma_{0t}  + \theta_t + \epsilon' \theta_0) G_q(1+\epsilon \theta_{\infty}  + \theta_1 + \epsilon' \sigma)}{G_q(1+2\sigma_{0t}) G_q(1-2\sigma_{0t})} t^{\sigma_{0t}^2 - \theta_t^2 - \theta_0^2} \tau_{W},
\end{align}
however, we were not able to verify this relation beyond the first few leading terms.
\item We have constructed four tau functions, $\tau_j$, $1\leq j\leq 4$, whose zeros correspond to the corresponding solution $(f,g)$ taking value in correspondingly indexed exceptional lines $E_j$, $1\leq j\leq 4$, see Theorem \ref{thm:fundamentalrelationtauRHP}. The definition of the tau function in \cite{tsudamasuda} suggests that appropriate shifts of the parameters would provide four further tau functions $\tau_j$, $5\leq j\leq 8$, similarly corresponding to the exceptional lines $E_j$, $5\leq j\leq 8$. Deriving the correct shifts and verifying the statement of Theorem \ref{thm:fundamentalrelationtauRHP} for them using Riemann-Hilbert theory would be of great interest. Furthermore, Tsuda and Masuda \cite[Theorem 2.3]{tsudamasuda}\footnote{Note the correspondence $(\sff, \sfg)=(g^{\operatorname{TM}},f^{\operatorname{TM}})$ between the notation for dependent variables in \cite{tsudamasuda} and ours.} give a formula for  $\sfg$ in terms of these tau functions. Their formula, combined with the asymptotic expansions of $\sfg$ and the tau function in Theorems \ref{thm:asymptotic} and \ref{eq:tau-asymp}, would suggest that
    \begin{align}\label{g-conj}
           \sfg(t) =  -r_{0t}^{-1}(-t)^{1-2 \sigma_{0t}}m_{\sfg}\frac{\tau_7(t) \tau_8 (t)}{\tau_5(t) \tau_6(t)},
    \end{align}
    with constant multiplier $m_{\sfg}=-q^{-1}\kappa_{0t}m_{\sff}$, where $m_{\sff}$ is defined in Theorem \ref{thm:f-tau}.
Deriving such a formula is another interesting problem for the future. We note that
a similar formula is conjectured in \cite[Conjecture 3.4]{JNS2017}. 

\item One of the main consequences of the Fredholm determinant representation of tau functions is that their asymptotic expansions can be studied through the minor expansions of Fredholm determinants. As a next step to our results, one may attempt to derive the asymptotic expansion of the $q\Psix$ tau functions for large $t$, and solve the connection problem for the tau functions. This may be obtained through studying the monodromy dependence of the tau function as in \cite{ILP2016} or by using certain symmetry relations of the transcendents and tau functions as in \cite{gavrylenko2025riemannhilbertproblemsfredholmdeterminants}. 

\item Finally, extension of the Fredholm determinant framework to other discrete Painlev\'e equations, including additive and elliptic difference ones, is another interesting direction for future research.
\end{itemize}

\printbibliography
\end{document}